\documentclass[aps,prx,superscriptaddress,twocolumn,longbibliography,floatfix]{revtex4-2}
\usepackage{units}
\usepackage{amsmath}
\usepackage{amsthm}
\usepackage{amssymb}
\usepackage{graphicx}
\usepackage{color}
\usepackage{xcolor}
\usepackage{bbold}
\usepackage{apptools}
\usepackage{appendix}

\definecolor{myurlcolor}{rgb}{0,0,0.7}
\definecolor{myrefcolor}{rgb}{0.1,0,0.9}

\usepackage[
	breaklinks,
	pdftex,
	colorlinks=true, 
	linkcolor=myrefcolor,
	citecolor=myurlcolor,
	urlcolor=myurlcolor
]{hyperref}

\newcommand{\SMLong}{Appendix}

\newcommand{\SM}{Appendix}

\newtheorem{theorem}{Theorem}
\newtheorem{lemma}{Lemma}

\AtAppendix{\counterwithin{theorem}{section}}
\AtAppendix{\counterwithin{lemma}{section}}
\AtAppendix{\counterwithin{fact}{section}}
\AtAppendix{\counterwithin{definition}{section}}

\usepackage[linesnumbered, ruled,vlined]{algorithm2e}

\graphicspath{{./images/},{./imagesAppendix/}}

\renewcommand{\eqref}[1]{Eq.~(\ref{#1})} %

\def\app#1#2{%
  \mathrel{%
    \setbox0=\hbox{$#1\sim$}%
    \setbox2=\hbox{%
      \rlap{\hbox{$#1\propto$}}%
      \lower1.1\ht0\box0%
    }%
    \raise0.25\ht2\box2%
  }%
}

\newtheorem{fact}{\protect\factname}
\ifx\proof\undefined
\newenvironment{proof}[1][\protect\proofname]{\par
	\normalfont\topsep6\p@\@plus6\p@\relax
	\trivlist
	\itemindent\parindent
	\item[\hskip\labelsep\scshape #1]\ignorespaces
}{%
	\endtrivlist\@endpefalse
}
\providecommand{\proofname}{Proof}
\fi

\newtheorem{proposition}{Proposition}
\newcommand{\bra}[1]{\langle #1|}
\newcommand{\ket}[1]{|#1 \rangle}
\makeatother
\newcommand{\braket}[2]{\langle #1 \vert #2 \rangle}

\newcommand{\tr}{\mathrm{tr}}

\providecommand{\factname}{Fact}
\providecommand{\theoremname}{Theorem}
\providecommand{\claimname}{Claim}
\providecommand{\lemmaname}{Lemma}
\providecommand{\definitionname}{Definition}
\providecommand{\corollaryname}{Corollary}
\providecommand{\conjecturename}{Conjecture}

\newcommand\norm[1]{\left\lVert#1\right\rVert}

\definecolor{KB}{rgb}{0.4,0.3,0.9}

\definecolor{THc}{rgb}{0.9,0.3,0.2}

\definecolor{daxcolor}{rgb}{1,0,0.1}

\definecolor{nbcolor}{rgb}{0.2, 0.5, 0.5}

\newcommand{\revA}[1]{{#1}}

\newtheorem{definition}{\protect\definitionname}

\newtheorem{corollary}{\protect\corollaryname}

\def\d{\mathrm{d}}

\newcommand{\subfigimg}[3][,]{%
	\setbox1=\hbox{\includegraphics[#1]{#3}}%
	\leavevmode\rlap{\usebox1}%
	\rlap{\hspace*{2pt}\raisebox{\dimexpr\ht1-0.5\baselineskip}{{\bfseries \large\textsf{#2}}}}%
	\phantom{\usebox1}%
}

\newcommand{\sectionMain}[1]{
\let\oldaddcontentsline\addcontentsline%
\renewcommand{\addcontentsline}[3]{}%
\section{#1}
\let\addcontentsline\oldaddcontentsline
}

\newcommand{\Eset}[1]{\underset{#1}{\mathbb{E}}}

\newcommand{\prlsection}[1]{\section{#1}}

\newcommand{\ihpc}{Institute of High Performance Computing (IHPC), Agency for Science, Technology and Research (A*STAR), 1 Fusionopolis Way, $\#$16-16 Connexis, Singapore 138632, Republic of Singapore}
\newcommand{\sutd}{Science, Mathematics and Technology Cluster, Singapore University of Technology and Design, 8 Somapah Road, Singapore 487372, Singapore}
\newcommand{\qinc}{
A*STAR Quantum Innovation Centre (Q.InC), Institute of High Performance Computing (IHPC), Agency for Science, Technology and Research (A*STAR), 1 Fusionopolis Way, \#16-16 Connexis, Singapore 138632, Singapore}

\newcommand{\CQuERE}{Centre for Quantum Engineering, Research and Education, TCG CREST, Sector V, Salt Lake, Kolkata 700091, India \looseness=-1}

\newcommand{\IISER}{Department of Physical Sciences, Indian Institute of Science Education and Research (IISER) Mohali, Sector 81, SAS Nagar, Mohali, Punjab, 140306, India\looseness=-1}

\allowdisplaybreaks

\begin{document}

\title{Pseudorandom density matrices}

\author{Nikhil Bansal}
\email{nikhilbansaliiser@gmail.com}
\affiliation{\ihpc}
\affiliation{\IISER}

\author{Wai-Keong Mok}
\email{darielmok@caltech.edu}
\affiliation{Institute for Quantum Information and Matter, California Institute of Technology, Pasadena, CA 91125, USA}

\author{Kishor Bharti}
\email{kishor.bharti1@gmail.com}
\affiliation{\ihpc}
\affiliation{\qinc}
\affiliation{\CQuERE}
\affiliation{\sutd}

\author{Dax Enshan Koh}
\email{dax\_koh@ihpc.a-star.edu.sg}
\affiliation{\ihpc}
\affiliation{\qinc}
\affiliation{\sutd}

\author{Tobias Haug}
\email{tobias.haug@u.nus.edu}
\affiliation{Quantum Research Center, Technology Innovation Institute, Abu Dhabi, UAE}

\begin{abstract}
Pseudorandom states (PRSs) are state ensembles that cannot be efficiently distinguished from Haar random states. However, the definition of PRSs has been limited to pure states and lacks robustness against noise. 
Here, we introduce pseudorandom density matrices (PRDMs), ensembles of $n$-qubit states that are computationally indistinguishable from the generalized Hilbert-Schmidt ensemble (GHSE), which is constructed from $(n+m)$-qubit Haar random states with $m$ qubits traced out. For $m=0$, PRDMs are equivalent to PRSs, whereas for $m=\omega(\log n)$, PRDMs are computationally indistinguishable from the maximally mixed state. PRDMs with $m=\omega(\log n)$ are robust to unital noise channels and separated in terms of security from PRS.
PRDMs disguise valuable quantum resources, possessing near-maximal entanglement, magic and coherence, while being computationally indistinguishable from resource-free states. PRDMs exhibit a pseudoresource gap of $\Theta(n)$ vs $0$, surpassing previously found gaps. 
We also render EFI pairs, a fundamental cryptographic primitive, robust to strong mixed unitary noise.
Our work has major implications on quantum resource theory: We show that entanglement, magic and coherence cannot be efficiently tested, and that black-box resource distillation requires a superpolynomial number of copies. We also establish lower bounds on the purity needed for efficient testing and black-box distillation.
Finally, we introduce memoryless PRSs, a noise-robust notion of PRS which are indistinguishable to Haar random states for efficient algorithms without quantum memory, as well as noise-robust quantum money. Our work provides a comprehensive framework of pseudorandomness for mixed states, which yields powerful quantum cryptographic primitives and fundamental bounds on quantum resource theories.
\end{abstract}

\maketitle

 \let\oldaddcontentsline\addcontentsline%
\renewcommand{\addcontentsline}[3]{}%

\begin{figure*}[t]
	\centering	
	\subfigimg[width=0.75\textwidth]{}{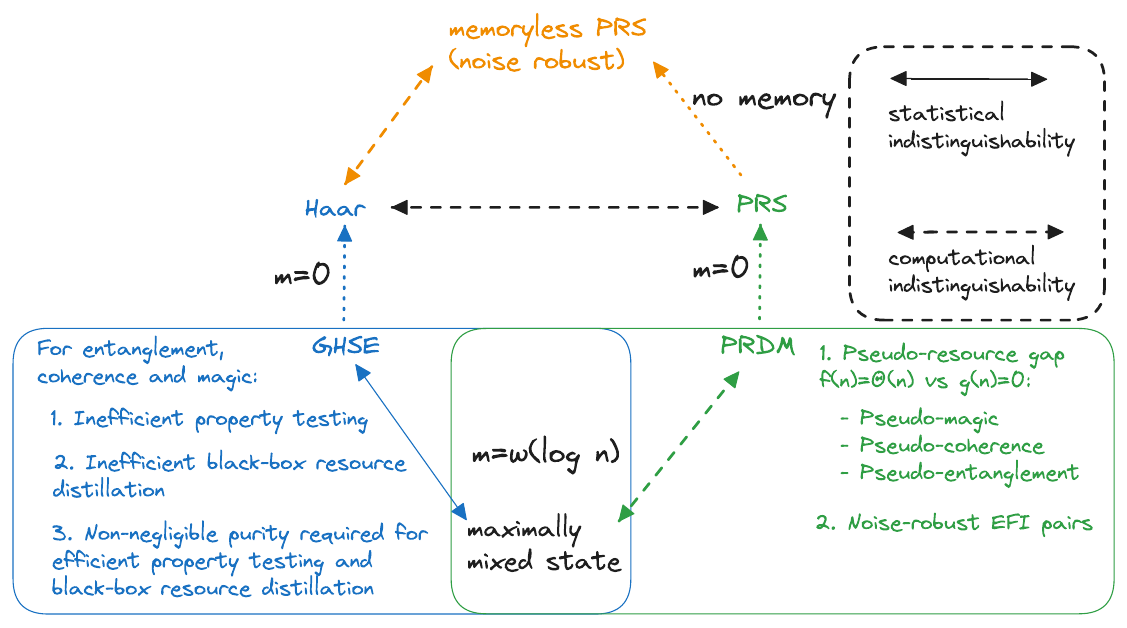}
	\caption{Overview of relationship between $n$-qubit GHSE, PRDM, Haar random states, PRS, and memoryless PRS depending on mixedness parameter $m$. Solid arrows indicate statistical indistinguishability, dashed lines are computational indistinguishability, and dotted line denote specializations of definitions. We show main results derived from statistical indistinguishability of GHSE and maximally mixed state for $m=\omega(\log n)$ (blue), as well as results due to computational indistinguishability between PRDM and maximally mixed state (green). 
 We also introduce memoryless PRS, which are noise robust and computationally indistinguishable from Haar random states for algorithms without access to quantum memory. 
	}
	\label{fig:PRDM}
\end{figure*}

\prlsection{Introduction} 
Research at the pairwise intersections of cryptography, quantum theory, information theory, and complexity theory has revealed many fascinating insights. These fields have evolved significantly over time, giving rise to new areas of study such as quantum information theory (arising from quantum theory and information theory), quantum complexity theory (arising from quantum theory and complexity theory), and modern cryptography (arising from cryptography and complexity theory).

In contrast, considering the comprehensive intersection of all four fields---cryptography, complexity theory, information theory, and quantum theory---has led to new developments. A key example is quantum pseudorandom states (PRSs), which are efficiently preparable state ensembles that are indistinguishable from Haar random states for any efficient quantum algorithm~\cite{ji2018pseudorandom}. 
Further developments includes the generation of certifiable true randomness~\cite{brakerski2021certifying}, pseudorandom unitaries~\cite{ji2018pseudorandom}, homogeneous space pseudorandomness~\cite{arvind2023quantumtugwarrandomness}, pseudorandom isometries~\cite{ananth2023pseudorandomisometry}, pseudorandom state scramblers~\cite{lu2023quantum}, computational entanglement theory~\cite{arnon2023computational} and new cryptographic principles under minimal assumptions~\cite{kretschmer2021quantum}. 
Similar to the emergence of modern cryptography from classical cryptography and complexity theory~\cite{katz2007introduction}, there is hope for establishing a new sub-field called ``modern quantum information theory'' (see \SM{}~\ref{sec:motivation}), where the computational aspects of objects from quantum information theory are rigorously explored. There is the potential for breakthrough similar to the one cryptography experienced in 1980s. This calls for the development of new primitives, which extend our understanding beyond current notions of pseudorandomness.

An example where pseudorandomness has led to novel insights into quantum information presents in the field of resource theories. Quantum information processing tasks rest on the availability of precious quantum resources such as entanglement, magic and coherence~\cite{chitambar2019quantum}.
It has been shown that PRSs allow one to hide quantum resources in plain sight and mask low resource states as highly complex ones~\cite{bouland2022quantum}. This counter-intuitive phenomena, dubbed \textit{pseudoresources}~\cite{haug2023pseudorandom}, has been established in different resource theories such as pseudoentanglement~\cite{bouland2022quantum}, pseudomagic~\cite{gu2023little}, and pseudocoherence~\cite{haug2023pseudorandom}. 
As a result, the existence of pseudoresources has imposed fundamental limits on testing whether a state contains quantum resources~\cite{montanaro2013survey,wright2016learn,haug2023pseudorandom}, and the distillation of resource-rich states from noisy states~\cite{bennett1996concentrating,regula2021fundamental,gu2023little}.

\revA{However, the concept of PRS is only well defined for pure states: When the PRS is subject to even the lowest amounts of noise during state preparation, it can be efficiently distinguished from Haar random states via the SWAP test~\cite{haug2023pseudorandom}.}
Evidently, the most general quantum state is not a pure state; instead, it is a convex combination of pure states, known as a mixed state or a density matrix. Mixedness arises naturally whenever one does not keep track of information about the state, for example when the state interacts with an uncontrolled environment, or when one randomizes the state preparation protocol. Finding a definition of pseudorandomness that is based on density matrices could generate notions of pseudorandomness that are robust to noise, establish a general theory of pseudoresources, and find improved bounds on property testing.

In this work, we provide a step towards shaping the aforementioned field of modern quantum information theory.
We introduce pseudorandom density matrices (PRDMs) as mixed-state generalization of PRS. PRDMs are efficiently preparable $n$-qubit state ensembles that are computationally indistinguishable from the mixed state ensemble obtained by tracing out $m$ qubits of $(n+m)$-qubit Haar random states, which we refer to as the generalized Hilbert-Schmidt ensemble (GHSE)~\cite{braunstein1996geometry,hall1998random,Zyczkowski_2001,hayden2006aspects,Zyczkowski_2011,sarkar2019bures}.
The GHSE corresponds to Haar random states for $m=0$, while for $m=\omega(\log n)$, we show that it is statistically indistinguishable from the maximally mixed state. Yet, the GHSE with $m$ scaling polylogarithmically with $n$ has near-maximal entanglement, magic and coherence. 
Similarly, PRDMs with mixedness parameter $m=0$ correspond to PRS and become computationally indistinguishable from the maximally mixed state for $m=\omega(\log n)$. In this regime, PRDMs become robust to unital noise channels, \revA{and are fundamentally separated in security against inefficient distinguishers.}
Surprisingly, while such PRDMs appear trivial to any efficient observer, they can have asymptotically maximal entanglement, coherence and magic.
We construct pseudoentangled, pseudocoherent and pseudomagic state ensembles, which are two ensembles that are computationally indistinguishable, yet possess a maximal gap in entanglement, magic and coherence of $\Theta(n)$ vs $0$, an improvement over previous constructions which were believed to be optimal~\cite{arnon2023computational,gu2023little,haug2023stabilizer}. The reason we are able to improve the pseudoresource gap bounds is that we consider general mixed states, while previous bounds only considered pure states. 
Furthermore, we establish new constructions for EFI pairs, which are statistically far yet computationally indistinguishable ensembles~\cite{brakerski2022computational} and serve as important cryptographic primitives~\cite{yan2022general,bartusek2021one,ananth2022cryptography,ananth2021concurrent}. We show that PRDMs and the maximally mixed state can form noise-robust EFI pairs, which remain EFI pairs even when subjected to mixed unitary noise channels, including local depolarizing noise up to a relatively high noise probability of $p\lesssim1/4$.

\revA{We show that
testing entanglement, magic or coherence is inefficient: Given an unknown state, testing whether it has a lot of resources or none requires a superpolynomial number of copies. This implies that entanglement, magic and coherence are in general not efficiently physically observable. In contrast, testing pure states requires only $O(1)$ copies. 
Additionally, we show that black-box distillation of entanglement, magic and coherence is inefficient. 
We also place lower bounds on the purity needed for efficient testing and black-box distillation, where at least a non-negligible amount of purity is necessary.}

\revA{We also show that even if PRS is not noise-robust, its applications can be made robust to noise. In particular, we provide a noise-robust private-key quantum money based on PRS. Even under noise, quantum banknotes cannot be counterfeited, yet are still accepted by the bank. As technical contribution, we propose a modified completeness amplification scheme based on Ref.~\cite{aaronson_quantum_money_2012_arxiv}, which now is secure against embezzling attacks that in the original scheme can be used to counterfeit banknotes.}

Finally, we define a weaker notion of PRS called memoryless PRSs, which are indistinguishable from Haar random states only for efficient algorithms without quantum memory. We show that memoryless PRS are robust to unital noise, contrary to general PRSs.

Our work generalizes the notion of PRS to mixed states, promising applications in cryptography and quantum resource theories. The main results of this paper and relationships between the concepts introduced are shown in Fig.~\ref{fig:PRDM}. Sample complexity of testing is summarized in Fig.~\ref{fig:testing}, and pseudoresource gaps of PRDMs and PRS in Table~\ref{tab:pseudoresource}.

\begin{figure}[t]
	\centering	
	\subfigimg[width=0.4\textwidth]{}{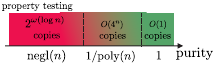}
	\caption{Copies needed to test whether a state $\rho$ contains $Q(\rho)=0$ or $Q(\rho)=\Theta(n)$ magic, coherence or entanglement, as function of purity $\text{tr}(\rho^2)$. We show testing is inefficient for negligible purity, while efficient protocols are known for pure states, and only inefficient tomography protocols have been known so for $1/\text{poly}(n)$ purity. 
	}
	\label{fig:testing}
\end{figure}

\begin{table}[htbp]\centering
\begin{tabular}{| c|c|c| }
    \hline    
  Resource $Q$  &  PRS $f(n)$ vs $g(n)$ & PRDM $f(n)$ vs $g(n)$\\
    \hline\hline
    Magic  & $\Theta(n)$ vs $\omega(\log n)$~\cite{gu2023little} &   $\Theta(n)$ vs $0$ [Thm.~\ref{thm:maximalgap}]\\ \hline
    Coherence  & $\Theta(n)$ vs $\omega(\log n)$~\cite{haug2023pseudorandom}&   $\Theta(n)$ vs $0$ [Thm.~\ref{thm:maximalgap}]\\ \hline
        Entanglement   &   $\Theta(n)$ vs $\omega(\log n)$~\cite{bouland2022quantum}&  $\Theta(n)$ vs $0$ [Thm.~\ref{thm:maximalgap}]  \\
        \hline
\end{tabular}
\caption{Pseudoresource gap $f(n)$ vs $g(n)$ for different resources $Q$ using pure PRS and mixed PRDM.  
}
\label{tab:pseudoresource}
\end{table}

\prlsection{Quantum resources} 
Performing specific non-trivial tasks in quantum information require quantum resources as a fuel to run the process.
To characterize the resource content of a given state $\rho$, quantum resource monotones $Q(\rho)$ have been defined~\cite{chitambar2019quantum}. Resource monotones are characterized by free operations $F_Q$ which cannot increase the resource, i.e.\ $Q(F_Q(\rho))\leq Q(\rho)$, and a set of free states $\sigma\in S_Q$ with $Q(\sigma)=0$. Intuitively, free states and operations are `easy' in the context of the resource and are readily available. To perform non-trivial tasks, one needs `expensive' resource states $\rho\notin S_Q$ which have $Q(\rho)>0$, or non-free operations which can increase the resource. We consider sub-additive resource monotones which are bounded as $0\leq Q \leq \Theta(n)$.

Depending on the task, different resource theories are relevant.
In fault-tolerant quantum computing, stabilizer states and Clifford operations are easy to perform~\cite{bravyi2005universal}. They constitute the free states and operations, while the key resource needed to perform universal quantum computation is called \textit{magic}. A commonly used magic monotone for qubits is the log-free robustness of magic~\cite{howard2017application,liu2022many}
\begin{equation}
    \text{LR}(\rho)=\log\Big(\min \vert c_\phi\vert \text{ s.t }  \rho=\sum_{\phi \in \text{STAB}}c_\phi \phi\Big)
\end{equation}
where $\text{STAB}$ is the set of all pure $n$-qubit stabilizer states. 
Another key resource in quantum information is coherence, which describes the degree in which the state is a superposition of computational basis states~\cite{baumgratz2014quantifying,streltsov2017colloquium}. The free operations are diagonal operations, while the free states are density matrices with only diagonal entries in the computational basis.
A commonly used coherence monotone is the relative entropy of coherence~\cite{baumgratz2014quantifying,streltsov2017colloquium}
\begin{equation}
    C(\rho) = S(\Delta[\rho]) - S(\rho)
\end{equation}
where $S(\rho)=-\text{tr}(\rho \log \rho)$ denotes the von Neumann entropy and  $\Delta[\rho]=\sum_i \ket{i}\bra{i}\rho\ket{i}\bra{i}$ is the fully dephasing channel applied on $\rho$, where $\{\ket{i}\}_i$ are the computational basis states.
Finally, for quantum communication tasks one requires entanglement as resource, while local operations and classical communication (LOCC) are free operations, and separable states are the free states. 
\revA{To characterize the entanglement between bipartition $A$ and $B$ we use the distillable entanglement~\cite{plenio2005introduction}
\begin{equation}
    E_\text{D}(\rho)=\text{sup}\left\{r : \lim_{s\rightarrow\infty} \left(\inf_{\Lambda}\Vert \Lambda(\rho^{\otimes s})-\Phi_2^{\otimes rs} \Vert_1\right)=0\right\}\,,
\end{equation}
which characterizes the number of Bell pairs $\Phi_2=\ket{\psi^+}\bra{\psi^+}$, $\ket{\psi^+}=\frac{1}{\sqrt{2}}(\ket{00}+\ket{11})$ that can be asymptotically distilled from $\rho$ via LOCC operations $\Lambda$ with respect to bipartition $A$, $B$. Our results also apply to entanglement cost $E_\text{C}$ which is an upper bound to the distillable entanglement $E_\text{C}\geq E_\text{D}$. }

\prlsection{Generalized Hilbert-Schmidt ensemble (GHSE)}
For pure quantum states, one can define a unique ensemble of random states, namely Haar random states: One draws pure states $\ket{\psi}\in \mathcal{S}(2^n)$ from the space of $n$-qubit quantum states $\mathcal{S}(2^n)$ according to the Haar measure $\mu_n$. 
However, general quantum states are not pure, but probabilistic mixtures of quantum states described by density matrices $\rho=\sum_j p_j\ket{\psi_j}\bra{\psi_j}$ with probability $p_j$ for state $\ket{\psi_j}$. 

We now study a random ensemble of quantum states that  interpolates between Haar random states and maximally mixed state $I_n/2^n$ via a mixedness parameter $m$, where $I_n$ is the identity matrix over $n$-qubits. In particular, we consider the random ensemble of mixed states induced by the partial trace of $m$ qubits over $(n+m)$-qubit Haar random states, which we refer to as the GHSE~\cite{braunstein1996geometry,hall1998random,Zyczkowski_2001,hayden2006aspects,Zyczkowski_2011,sarkar2019bures}:
\begin{definition}[Generalized Hilbert-Schmidt ensemble (GHSE) \SM{}~\ref{sec:HS}]\label{def:tracenensemble}
The $(n,m)$ GHSE is an ensemble of $n$-qubit states 
    \begin{equation}\label{eq:eta}
        \eta_{n,m}=\{\operatorname{tr}_{m}(\ket{\psi}\bra{\psi})\}_{\psi\in\mu_{n+m}}
    \end{equation}
    generated by tracing out $m$ qubits from $(n+m)$-qubit states drawn from the Haar measure $\mu_{n+m}$.
\end{definition}
The case $m=0$ corresponds to Haar random states, while $m=n$ corresponds to the Hilbert-Schmidt ensemble~\cite{Zyczkowski_2001}. 

Haar random states are known to contain near-maximal coherence~\cite{singh2016average}, magic~\cite{liu2019quantum} and entanglement~\cite{hayden2006aspects}.
Increasing mixedness by tracing out $m$ qubits in general reduces the quantumness of states. For example, the limit $m \to \infty$ yields an equal mixture of all quantum states, the maximally mixed state $I_n/2^n$, which has zero magic, coherence and entanglement.

We find that states drawn from the GHSE possess two key behaviors: First, they are statistically indistinguishable from the maximally mixed state when tracing out $m=\omega(\log n)$ qubits. Second, for $m=O(\text{polylog}(n))$, GHSE states remain highly resourceful in terms of magic, coherence and entanglement. Here, $\operatorname{polylog}(n)$ denotes a polynomial of the logarithm, e.g.  $(\log n)^c$ with $c \geq1$.
Thus, we find surprisingly that GHSE with $m=\text{polylog}(n)$ hold two properties which on first glance appear counter-intuitive: They are statistically indistinguishable from the maximally mixed state (which has no quantum resources), yet possess a large amount of quantum resources. 
This means that GHSE for $m=\text{polylog}(n)$ are effectively hiding quantum resources in plain sight, pretending to be trivial states while having near-maximal quantum resources.

We summarize the properties of GHSE in the following, while the proofs are deferred to the \SMLong{}:
\begin{theorem}[Properties of GHSE]\label{thm:propGHSE}
    The $(n,m)$ GHSE $\eta_{n,m}$ is statistically close to the maximally mixed state (\SM{}~\ref{sec:stat_indist}), with\begin{equation}\label{eq:indist_ghse}
        \operatorname{TD}\left(\Eset{\rho \in \eta_{n,m}}[\rho^{\otimes t}], \left(I_n/2^n\right)^{\otimes t}\right)=O\left(t^2/2^m\right) ,
    \end{equation}
    where $\operatorname{TD}$ denotes the trace distance. On average, states $\rho\in \eta_{n,m}$  have log-free robustness of magic (\SM{}~\ref{sec:magic})
    \begin{equation}
        \text{LR}(\rho)\geq n-m-2\log(n+m) -1\,,
    \end{equation}
     relative entropy of coherence (\SM{}~\ref{sec:coherence})
        \begin{equation}
        C(\rho)\geq n-m-1\,,
    \end{equation}
    and \revA{distillable entanglement between bipartition $n_A=\vert A\vert$ and $n_B=\vert B\vert$ with $n_A\leq n_B$ (\SM{}~\ref{sec:entanglement}, see also~\cite{hayden2006aspects})
     \begin{equation}
        E_\text{D}(\rho)\geq n_A-m-1\,.
     \end{equation}}
\end{theorem}
We note that resource monotones are not unique. However, similar bounds can be found for most other practical definition of resource monotones. For example, the entanglement of formation and logarithmic negativity, which are upper bounds of $E_\text{D}$, have been shown to exhibit a similar scaling as the logarithmic negativity for the GHSE~\cite{hayden2006aspects,smith2006typical,montanaro2013survey,shapourian2021entanglement,bhosale2012entanglement}.

\prlsection{Pseudorandom density matrix (PRDM)} 
GHSE states  cannot be efficiently prepared in general, as Haar random states are known to be hard to prepare~\cite{nielsen2002quantum}. 
However, one can achieve efficient preparation by relaxing to a weaker notion of quantum randomness.
In particular, PRS have been recently proposed as a notion of quantum pseudorandomness that is efficiently preparable, but only computationally indistinguishable from Haar random states~\cite{ji2018pseudorandom}. Here, indistinguishability is defined from a computational perspective: There is no efficient quantum algorithm that can tell  PRS and Haar random states apart given a polynomial number of copies of the state. While computational pseudorandomness is weaker than statistical randomness, for practical applications  indistinguishability with respect to efficient quantum algorithms is already sufficient.

PRS has only been defined for pure states. We now generalize this concept to general density matrices. We define PRDM as efficiently preparable states that are computational indistinguishable from GHSE.
\begin{definition}[Pseudo-random density matrix (PRDM)]\label{def:PRDM}
    Let $\kappa=\operatorname{poly}(n)$ be the security parameter with keyspace $\mathcal{K}=\{0,1\}^{\kappa}$. A keyed family of $n$-qubit density matrices $\{\rho_{k,m}\}_{k \in \mathcal{K}}$ is defined as the pseudorandom density matrix (PRDM) ensemble with mixedness parameter $m$ if:
    \begin{enumerate}
        \item {Efficiently preparable}: There exists an efficient quantum algorithm $\mathcal{G}$ such that $\mathcal{G}(1^{\kappa}, k,m) = \rho_{k,m}$.
        \item {Computational Indistinguishability}: $t=\mathrm{poly}(n)$ copies of $\rho_{k,m}$ are computationally indistinguishable, i.e. for any quantum polynomial time adversary $\mathcal{A}$, from the GHSE $\eta_{n,m}$ of Def.~\ref{def:tracenensemble}
        \begin{equation}
            \Big{|}\Pr_{k \leftarrow \mathcal{K}}[\mathcal{A}(\rho_{k,m}^{\otimes t}) = 1] - \Pr_{\rho \leftarrow \eta_{n,m}}[\mathcal{A}(\rho^{\otimes t}) = 1]\Big{|} = \operatorname{negl}(n).
        \end{equation}
    \end{enumerate}
\end{definition}
Here, $\operatorname{negl}(\cdot)$ are functions that decay faster than any inverse polynomial (see \SM{}~\ref{sec:definitions}).
Depending on $m$, PRDMs are indistinguishable from different notions of randomness:
For $m=0$, PRDMs are computationally indistinguishable from Haar random states $\{\ket{\psi}\bra{\psi}\}_{\psi\in \mu_n}$ and thus equivalent to PRS. For $m=\omega(\log n)$, we find that PRDMs become indistinguishable from the maximally mixed state for any efficient observer:
\begin{theorem}[Computational indistinguishability of PRDMs from maximally mixed state]\label{thm:PRDM_indisting}
PRDMs $\rho_k\equiv \rho_{k,\omega(\log n)}$ with $t=\mathrm{poly}(n)$ copies and $m=\omega(\log n)$ are indistinguishable from the maximally mixed state $\sigma_n=I_n/2^{-n}$ for any efficient quantum algorithm $\mathcal{A}$,
  \begin{equation}
        \Big{|}\Pr_{k \leftarrow \mathcal{K}}[\mathcal{A}(\rho_k^{\otimes t}) = 1] - \Pr [\mathcal{A}((\sigma_n)^{\otimes t}) = 1]\Big{|} = \operatorname{negl}(n).
    \end{equation}
\end{theorem}
To see this, note that PRDM and GHSE are computationally indistinguishable by definition. Further, the maximally mixed state and GHSE are statistically indistinguishable for $m=\omega(\log n)$ due to~\eqref{eq:indist_ghse}. Then, Thm.~\ref{thm:PRDM_indisting} follows directly from the triangle inequality.

\revA{Next, we discuss the noise robustness of PRS and PRDM.
PRS are defined as being computationally indistinguishable from Haar. If the same noise channel is applied on both PRS and Haar random state (e.g. via the distinguishing protocol), due to the data-processing inequality both states remain of course indistinguishable. 
However, a more practical scenario is that only the PRS is subject to noise during its preparation. In this case, PRS are extremely fragile against noise~\cite{haug2023pseudorandom}. }
In particular, after applying depolarizing noise on only a single qubit, PRSs become distinguishable from Haar random states and thus are not PRS anymore. One can easily check that this susceptibility to noise carries over to PRDMs with small $m$.
In particular, for $m=O(\log n)$ PRDMs are not robust to noise as the SWAP test~\cite{barenco1997stabilization}, which requires only two copies, can efficiently distinguish the state before and after application of the noise channel due to non-negligible purity~\cite{haug2023pseudorandom}. 

As noise is ubiquitous in quantum systems, we would like to have pseudorandomness that can survive even with noise:
\begin{definition}[Noise-robust PRDM]\label{def:noise-robust}
    A PRDM is noise robust to channel $\Gamma(\cdot)$ if it remains a PRDM after application of the noise channel $\Gamma(\cdot)$ \revA{on the PRDM}, i.e. if $\{\rho_k\}_{k \in \mathcal{K}}$ is PRDM, then $\{\Gamma(\rho_k)\}_{k \in \mathcal{K}}$ is also PRDM.
\end{definition}
We now show that PRDMs with $m=\omega(\log n)$ are robust to arbitrary unital noise channels that can be efficiently implemented, which includes common noise models such as depolarizing noise:
\begin{theorem}
    PRDMs are robust to efficiently implementable unital noise channels, i.e. channels where the identity is the fixed point $\Gamma(I)=I$, if and only if $m=\omega(\log n)$. 
\end{theorem}
This follows from the computational indistinguishability from the maximally mixed state and the fixed point condition of unital channels as shown in \SM{}~\ref{sec:noise}.

\revA{How secure are PRDMs against attackers with superpolynomial computational power? For PRS (i.e. $m=0$), Ref.~\cite{kretschmer2021quantum} showed that one can sample-efficiently (i.e. using polynomial many copies of the state) distinguish PRS from Haar random states using classical shadows and an (inefficient) $\mathsf{PostBQP}$ oracle. Here, $\mathsf{PostBQP}$ is the class of quantum algorithms combined with post-selection, which is known to be equivalent to $\mathsf{PP}$~\cite{aaronson2005quantum}.
However, we find that PRDMs for $m=\omega(\log n)$ are robust to this $\mathsf{PostBQP}$ attack (see \SM{}~\ref{sec:kretschmer}), i.e. it requires superpolynomially many copies. Thus, there are attacks that break PRS yet do not work for PRDMs.
However, the opposite is not true: Any sample-efficient attack that can distinguish PRDM from GHSE can also distinguish PRS from Haar random states. %
Thus, there is a fundamental separation in security between PRS and PRDMs:
\begin{theorem}[Security separation between PRDM and PRS (see \SM{}~\ref{sec:kretschmer})]\label{thm:separation}
    Every sample-efficient algorithm that distinguishes PRDMs from GHSE also distinguishes PRSs from Haar random states, while there exist $\mathsf{PostBQP}$ oracles that sample-efficiently distinguish PRSs from Haar random states yet fail for PRDM and GHSE.
\end{theorem}
Finally, we note that whether PRDMs are generally secure against $\mathsf{PostBQP}$ attacks remains an open question.}

\revA{To efficiently construct PRDMs $\rho_{k,m}$, one can harness any existing construction of PRS: We take an $(n+m)$-qubit PRS $\ket{\psi_k}\bra{\psi_k}$, and trace out $m$ qubits:
\begin{equation}\label{eq:tracePRDM}
\rho_{k,m}=\operatorname{tr}_m(\ket{\psi_k}\bra{\psi_k})\,.
\end{equation}
This implies that PRDMs can be prepared with very short circuit depth: PRS (and pseudorandom unitaries) can be prepared in $O(\text{polylog}\,n)$ depth using one-dimensional local circuits, and $O(\text{polylog}\log n)$ depth assuming arbitrary connectivity~\cite{schuster2024randomunitariesextremelylow,ma2024construct}.} 
In the following, we now consider an explicit PRDM construction that we can use for applications. In particular, we construct a real-valued PRDM by taking partial trace of PRS constructed from binary phase states~\cite{ji2018pseudorandom,brakerski2019pseudo,bouland2022quantum, ananth2022pseudorandom}:
\begin{fact}[Binary phase PRDM]\label{thm:PRDMSconstruct}
The $n$-qubit binary phase PRDM $\rho_{k,m}=\operatorname{tr}_m(\ket{\psi_k}\bra{\psi_k})$ is efficiently prepared by tracing out $m$ qubits of the binary phase PRS of $n+m$ qubits $\ket{\psi_k}=2^{-(n+m)/2}\sum_{x\in \{0,1\}^{n+m}} (-1)^{f_k(x)}\ket{x}$~\cite{ji2018pseudorandom,brakerski2019pseudo,bouland2022quantum, ananth2022pseudorandom} with pseudorandom function $f_k(x): \{0,1\}^{n+m}\rightarrow\{0,1\}$, key $k\in\{0,1\}^{\kappa}$ and security parameter $\kappa=O(n+m)$. 
\end{fact}
This construction of PRDM requires the existence of PRS, which in turn have been proven to exist assuming quantum-secure one-way functions exist~\cite{zhandry2012how}. However, it has been noted that PRS can exist under even weaker assumptions~\cite{kretschmer2021quantum}. \revA{It is also possible that PRDMs \revA{for $m=\omega(\log n)$} could be constructed without using PRS at all. If this is the case, then PRDMs could require even weaker assumptions than PRS.}

\prlsection{Pseudoresources} 
Can one efficiently mask the fact that a state contains quantum resources? Recently, pseudoresources have been proposed as efficiently preparable ensembles which are computationally indistinguishable, yet possess substantially different amount of resources~\cite{bouland2022quantum,gu2023little,haug2023pseudorandom}:
\begin{definition}[Pseudoresources]\label{def:pseudoresource}
Let $Q$ be a quantum resource monotone. A pseudoresource pair with gap $f(n)$ vs. $g(n)$ (where $f(n)>g(n)$) consists of two efficiently preparable state ensembles:
\begin{enumerate}
    \item a `high resource' ensemble of $n$-qubit quantum states $\{\rho_{k_1}\}$ such that $Q(\rho_{k_1})=
    f(n)$ with high probability over key ${k_1}$, and
    \item a `low resource' ensemble of $n$-qubit quantum states $\{\sigma_{k_2}\}$ such that $Q(\sigma_{k_2})=
    g(n)$ with high probability over key ${k_2}$,
\end{enumerate}
such that the two ensembles are computationally indistinguishable given $t=\mathrm{poly}(n)$ copies. %
\end{definition}
Pseudoresource ensembles allow one to efficiently generate states masquerading as high resource states with $f(n)$ resource, yet actually contain only $g(n)$ of a given resource.
For example, a pseudoresource gap of $f(n)=\Theta(n)$ vs $g(n)=\omega(\log n)$ has been found for pseudoentanglement~\cite{bouland2022quantum}, pseudomagic~\cite{gu2023little} and pseudocoherence~\cite{haug2023pseudorandom}. 
For pure states (i.e. PRSs) these pseudoresource gaps are indeed maximal. 
However, it turns out that for mixed states we can hide quantum resources even better: In particular, PRDMs as constructed via binary phase states have near-maximal entanglement, coherence and magic, yet are computationally indistinguishable from the maximally mixed states. Thus, these two ensembles have asymptotically maximal pseudoresource gaps:
\revA{\begin{theorem}[Optimal pseudoentanglement, pseudocoherence and pseudomagic gap]\label{thm:maximalgap}
The PRDM ensemble of Fact~\ref{thm:PRDMSconstruct} with $m=\operatorname{polylog}(n)$, and the maximally mixed state $I_n/2^n$, form a pseudoentanglement (\SM{}~\ref{sec:pseudoentanglement}),  pseudomagic (\SM{}~\ref{sec:pseudomagic}) and pseudocoherence (\SM{}~\ref{sec:pseudocoherence})  ensemble with gap $f(n)=\Theta(n)$ vs $g(n)=0$, which is asymptotically optimal. 
\end{theorem}
We show the pseudoresource gaps in terms of distillable entanglement $E_\text{D}$, log-free robustness of magic $\text{LR}$, and relative entropy of coherence $C$. Note that the same gaps also apply for other good resource monotones, e.g. for the entanglement cost $E_\text{C}$.
We summarize the pseudoresource gaps for pure and mixed states in Table~\ref{tab:pseudoresource}.
We note that a similar pseudoentanglement gap for mixed states was also found in Ref.~\cite{goulao2024pseudo}. However, their result only applies to a more restrictive definition of pseudoentanglement, where the bipartition to compute the entanglement is fixed beforehand, whereas our result holds for randomly chosen bipartitions.
}

\prlsection{Noise-robust EFI pairs} 
We now consider another application of pseudorandomness:
EFI pairs are efficiently preparable ensembles which are statistical far, yet computational indistinguishable~\cite{brakerski2022computational}. They are an important cryptographic primitive for various applications, such as for bit commitment~\cite{yan2022general}, quantum oblivious transfer~\cite{bartusek2021one}, multiparty quantum computation~\cite{ananth2022cryptography}, and zero knowledge proofs~\cite{ananth2021concurrent}. EFI pairs can be constructed from PRS~\cite{brakerski2022computational} and single-copy PRS~\cite{morimae2022quantum}.
However, for practical purposes we would like  EFI pairs to still function under noise, while being constructed with weaker assumptions than quantum-secure one-way functions. Here, we show that one can construct noise-robust EFI pairs using PRDMs and the maximally mixed state:
\begin{theorem}[Noise-robust EFI pair]\label{thm:EFI}
    PRDMs constructed by tracing out $m$ qubits from $(n+m)$ qubit PRS with security parameter $\kappa=c(n+m)$ with $0<c<1$, and the maximally mixed state $I_n/2^n$ are EFI pairs when $\omega(\log n)< m<\frac{n}{2}(1-c)-\frac{1}{2}$. They remain EFI pairs after applying efficient mixed unitary channel $\Phi(\rho)=\sum_{i=1}^{r} p_i U_i \rho U_i^\dag$ with unitaries $U_i$ and probabilities $\{p_i\}_i$  whenever its Shannon entropy is bounded as
    \begin{equation}
        H(\{p_i\}_i)\leq n(1-c)-m-2\,.
    \end{equation}
\end{theorem}
The proof in \SM{}~\ref{sec:EFI} follows from the Fannes-Audenaert inequality~\cite{audenaert2007sharp}. Intuitively, the inequality can be interpreted as saying that two states with very different von Neumann entropies must also be far apart in trace distance. The maximally mixed state has an entropy of $n$, while for PRDMs it is upper bounded by $m+\kappa$. Thus, whenever $m+\kappa<n$, these states are far apart in trace distance~\cite{morimae2022quantum}, and remain distinguishable even after application of noise as long as the entropy remains sufficiently below that of the maximally mixed state. While our proof of noise-robustness is restricted to mixed unital noise channels, we believe it can be extended to arbitrary efficiently implementable unital channels.

Our result implies that EFI pairs from PRDMs are robust to many realistic noise models, such as dephasing, depolarizing or Pauli channels.
For noisy intermediate-scale quantum computers and quantum error correction models, the most commonly used noise model is local depolarizing noise $\Lambda_p(\rho)^{\otimes n}$ which acts on all $n$ qubits~\cite{bharti2021noisy,chen2023complexity}. Here, $\Lambda_p(\rho)=p/4\sum_{\alpha\in\{x,y,z\}}\sigma^\alpha \rho \sigma^\alpha+(1-3p/4)\rho$ is the local depolarizing channel, $\sigma^\alpha$ with $\alpha\in\{x,y,z\}$ Pauli operators, and $p$ the depolarizing probability. Here, we find noise robustness even for relatively high noise rates of $p \lesssim \frac{1}{4}$ (see SM~\ref{sec:EFI}):
\begin{corollary}[Noise-robust EFI pair against local depolarizing noise]
EFI pairs of Thm.~\ref{thm:EFI} remain EFI pairs after applying the local depolarizing channel on all $n$ qubits $\Lambda_p^{\otimes n}(\rho)$ whenever $H(\{1-3p/4,p/4,p/4,p/4\}) \leq (1-c) -m/n- 2/n$. In particular, for $m 
= \operatorname{polylog}(n)$ and $c=10^{-4}$, we have robustness for all $p < \frac{1}{4} -O(\operatorname{polylog}(n)/n)$.
\end{corollary}
Further, we find that there is a fundamental connection between pseudoentanglement and construction of EFI pairs. In particular, EFI pairs have to be statistical far, which due to the Fannes-Audenaert inequality is easier to satisfy whenever the  PRS used to construct the PRDM has low entanglement. In particular, when constructing the PRDM by tracing $m$ qubits from pseudoentangled PRS such as introduced in Ref.~\cite{bouland2022quantum}, we can achieve EFI pairs for any $m=cn$ with $c>0$.

\revA{We note that for applications of EFI pairs, it is often desirable that they can be efficiently given key $k$~\cite{brakerski2022computational}. A recent work extended our construction of noise-robust EFI pairs from minimal assumptions to also include efficient verification~\cite{haug2024pseudorandom}.}

\prlsection{Property testing of resources}
Given an unknown object, are there experiments that can verify its properties? Property testing deals exactly with this question:
A property tester is a quantum algorithm that checks whether a given quantum state does not have a particular resource, or contains a lot of it (see Refs.~\cite{rubinfeld1996robust,goldreich1998property,buhrman2008quantum,montanaro2013survey} or SM~\ref{sec:definitions}).
Depending on the property, it has been known that the property tester may need a vastly different number of copies of the states to succeed~\cite{montanaro2013survey}. For example, testing whether the description of a quantum state has imaginary numbers requires exponentially many copies~\cite{haug2023pseudorandom}. In contrast, checking whether a quantum state is pure or highly mixed can be done efficiently with only $O(1)$ copies~\cite{barenco1997stabilization}.

Entanglement, magic and coherence are fundamental signatures of quantum complexity and strongly affect the structure of quantum states. Thus, it is natural to ask whether they can be tested efficiently. For pure states, efficient property testers using $O(1)$ copies are known for entanglement~\cite{ekert2002direct,bouland2022quantum}, magic~\cite{gross2021schur,haug2022scalable,haug2023efficient} and coherence~\cite{haug2023pseudorandom}.
Can efficient property testers also exist for general mixed states? 
Here, we answer this question in negative, showing that testing distillable entanglement $E_\text{D}$, log-free robustness of magic $\text{LR}$, and relative entropy of coherence $C$ is inefficient:
\begin{theorem}[Testing entanglement, coherence and magic is inefficient]\label{thm:inefficient_testing}
    For resource monotones $Q=\{E,\text{LR},C\}$, testing whether a given state $\rho$ has resource $Q(\rho)=0$  or $Q(\rho)=\Theta(n)$ requires $2^{\omega(\log n)}$ many copies .
\end{theorem}
We show this result in \SM{}~\ref{sec:magic}, \ref{sec:coherence}, \ref{sec:entanglement} by contradiction: If a property tester using polynomial copies exist, then it could efficiently distinguish the GHSE with $m=\text{polylog}(n)$ and $Q(\rho)=\Theta(n)$ from the maximally mixed state with $Q(I_n/2^n)=0$, which contradicts the indistinguishability shown in Thm.~\ref{thm:propGHSE}. We note that inefficiency of property testing from a computational point of view can be similarly established using our pseudoresource ensembles.
We note that the inefficiency of testing entanglement was already pointed out in Ref.~\cite{montanaro2013survey} regarding the entanglement of formation.

Crucially, we can identify purity as the resource that makes property testing hard: Pure states $\text{tr}(\rho^2)=1$ can be efficiently tested, while testing highly mixed states $\text{tr}(\rho^2)=\text{negl}(n)$ is inefficient. Thus, one requires at least inverse polynomial purity to efficiently test states:
\begin{corollary}[Lower bound on purity for efficient testing of entanglement, coherence and magic (\SM{}~\ref{sec:magic}, \ref{sec:coherence}, \ref{sec:entanglement})]\label{thm:puritybound}
    Efficient testing for resource $Q=\{E,\text{LR},C\}$ requires at least a non-negligible purity $\text{tr}(\rho^2)=\Omega(n^{-c})$ with $c>0$. 
\end{corollary}
For the sample complexity for purity $\text{tr}(\rho^2)=1/\text{poly}(n)$, the best known upper bounds are from tomography as $O(4^n)$~\cite{haah2016sample,o2016efficient}. Note that similarly one finds that efficient testing requires $\text{rank}(\rho)=O(\text{poly}(n))$.
In Fig.~\ref{fig:testing}, we summarize the currently known copies complexities of property testing for magic, coherence and entanglement as function of purity.

\prlsection{Black-box resource distillation}
Another key task in quantum information is to generate special resource states that are useful for information processing.
For example, quantum teleportation needs entangled Bell states~\cite{bennett1996concentrating,bennett1996purification,bennett1996mixed}, universal quantum computing uses magical $T$-states~\cite{bravyi2005universal} and creating quantum superpositions requires coherence~\cite{winter2016operational,fang2018probabilistic,regula2018one}. However, generating these resource states directly is often difficult. This necessitates resource distillation, which takes in many copies of a noisy resource state, and applies the free operations of the resource theory to prepare one noise-free resource state.

Commonly, resource distillation is considered for the case where the noisy input state is known beforehand. Indeed, for classes of mixed states limits on distillation have been shown~\cite{krastanov2019optimized,fang2020no,marvian2020coherence,fang2022no}. However, often noise and input states are not well characterized. Yet, if we are guaranteed that these noisy states are resourceful, can one still extract noise-free resource states from them? 
Such an algorithm that distills resources from unknown, but resourceful states is called black-box resource distillation~\cite{gu2023little}:
\begin{definition}[Black-box resource distillation]\label{def:blackbox}
A black-box resource distillation algorithm $\mathcal{D}_Q$ uses the free operations of resource theory $Q$ on arbitrary input state $\rho$, which is guaranteed to contain resource $Q(\rho)\geq Q_\text{in}$. $\mathcal{D}_Q$ prepares pure resource state $\ket{\psi}$ with $Q(\ket{\psi})\geq Q_\text{out}$, where $\ket{\psi}$ can depend on $\rho$.
\end{definition}
Here, we show that black-box resource distillation of entanglement, magic and coherence is inefficient:
\begin{theorem}[Inefficiency of black-box resource distillation]\label{thm:distill}
Any black-box resource distillation algorithm $\mathcal{D}_Q$ for resource monotones $Q=\{E,\text{LR},C\}$ with $Q_\text{in}=\Theta(n)$, $Q_\text{out}=\Omega(n^{-c})$ and $c>0$, requires a superpolynomial number of input states $\rho$.
\end{theorem}
This is proven by contradiction in \SM{}~\ref{sec:magic}, \ref{sec:coherence}, \ref{sec:entanglement}: if black-box distillation algorithms exist, they could be used for efficient property testing which contradicts Thm.~\ref{thm:inefficient_testing}. This argument can be extended even to probabilistic black-box distillation algorithms, which prepare the resource state with non-negligible probability. 

As we have done for property testing, we can also place a lower bound on the purity required for efficient black-box resource distillation:
\begin{corollary}[Purity is necessary for black-box distillation]\label{thm:distill_purity}
Any efficient black-box resource distillation $D_Q$ with $Q=\{E,\text{LR}, C\}$ requires input states $\rho$ with purity $\text{tr}(\rho^2)=\Omega(n^{-c})$ to prepare non-trivial resource states.  
\end{corollary}
Similarly, efficient black-box distillation requires input states with at most $\text{rank}(\rho)=O(\text{poly}(n))$.

\prlsection{Memoryless PRS}
Finally, we revisit the question of noise-robustness of PRS. So far, PRS and PRDMs were defined via indistinguishability with respect to arbitrary efficient algorithms. However, in many practical scenarios the distinguishing algorithm has even stronger limitations on its computational power. We show that such restricted algorithms allow one to define a weaker notion of PRS that is noise robust.
In particular, we propose memoryless PRS, which are indistinguishable from Haar random states for any efficient algorithms which have no access to quantum memory~\cite{aharonov2021quantum,chen2022exponential}. Algorithms without quantum memory can only perform measurements on a  single copy at a time, and do not have a register to store quantum information between measurements. However, they can adaptively choose the measurement depending on the previous measurement outcomes. 
\begin{theorem}[Noise robustness of memoryless PRS (\SM{}~\ref{Pseudorandomness without quantum memory})]\label{thm:memorylessPRS}
    Algorithms without quantum memory require superpolynomial number of copies to distinguish memoryless PRS subject to unital noise channels and Haar random states.
\end{theorem}
We give the formal definition of memoryless PRS and the proof of Thm.~\ref{thm:memorylessPRS} in \SM{}~\ref{Pseudorandomness without quantum memory} using the techniques of Ref.~\cite{chen2022exponential}. In particular, without quantum memory one cannot distinguish Haar random states and maximally mixed state, which we use to establish noise robustness. Memoryless PRS can be prepared by the same state preparation algorithms as for PRS. However, potentially easier state constructions with non-trivial quantum resources may exist.
No access to quantum memory can naturally occur in different scenarios of quantum information processing. Here, we give three explicit examples: First, in a distributed communication scenario, each agent receives only one copy of the state and can only communicate classical information with other agents. 
Second, in  near-term quantum computing, the number of available qubits is often limited. In the case where the state has almost as many qubits as supported by the near-term quantum computer, then it only can perform memoryless algorithms.
Third, we can also consider the case in which the distinguishing algorithm receives one copy of the memoryless PRS at a time, with a time interval longer than the coherence time of the quantum register storing the state. These correspond to algorithms without quantum memory. 
In these scenarios, our weaker notion of memoryless PRS is already sufficient for practical applications, which could potentially be easier to prepare in experiments than PRS or PRDMs.
Note that any memoryless PRS is also a single-copy PRS, where the distinguisher algorithm is given only a single copy the state~\cite{morimae2022quantum,chen2024power}, while there are single-copy PRS which are not memoryless PRS (see \SM{}~\ref{Pseudorandomness without quantum memory}).

Finally, in \SM{}~\ref{Pseudorandomness with noisy quantum memory}, we also show that when the distinguisher algorithm has access to bad quantum memory, i.e. quantum memory subject to noise, this does not yield any notion of noise-robust PRS.

\revA{\prlsection{Noise-robust quantum money}
The standard definition of PRS is not noise-robust, however does this imply that  applications of PRS do not work when subject to noise? 
An application of PRS is private-key quantum money~\cite{ji2018pseudorandom}. Here, each quantum banknote is a PRS issues by the bank. As PRS are indistinguishable from Haar random states, a counterfeiter cannot create more banknotes, yet the bank can efficiently verify the banknotes as it knows to the secret key. 
However, does this scheme still work when the banknote is subject to noise?

We find that the original quantum-money based on PRS by Ref.~\cite{ji2018pseudorandom} has negligible soundness error (i.e. banknotes cannot be counterfeited) even when subject to noise.
However, the scheme by Ref.~\cite{ji2018pseudorandom} has a completeness error that increases with noise, i.e. the bank is likely to reject quantum banknotes with increasing noise. This happens because the bank verifies the banknote by a projection, which with noise is likely to fail.

To solve this, we propose a noise-robust private-key quantum money based on PRS which is combined with an improved completeness amplification scheme (see SM~\ref{sec:quantummoney}).
Completeness amplification was original proposed in Ref.~\cite{aaronson_quantum_money_2012_arxiv} to preserve the value of quantum money by composing banknotes of multiple smaller banknotes. However, we find that the original completeness amplification scheme as described in the arXiv version of Ref.~\cite{aaronson_quantum_money_2012_arxiv} has a critical vulnerability: By embezzling smaller banknotes from multiple composite banknotes, one can create a counterfeit composite banknote. This is possible since in the original scheme, every smaller banknote has its own serial number, allowing one to extract valid small banknotes from the composite while preserving its value. 

We propose a modified completeness amplification protocol that is secure against such attacks. We achieve this by assigning a single serial number $s$   to the composite banknote given as
\begin{equation}
    \$=(s,\ket{\psi_{q_1}},\ket{\psi_{q_2}},\dots,\ket{\psi_{q_L}})\,,
\end{equation}
where $\ket{\psi_{q_i}}$ are in total $L=\text{poly}(n)$ PRS and $q_1,\dots,q_L=f_k(s)$ are generated from a quantum-secure pseudorandom function.
Our banknotes retain their value under noise, with a completeness error $\epsilon=\text{negl}(n)$ for any noise channel $\Gamma(.)$ which has fidelity $F(\Gamma(\ket{\psi_{q_i}}))>1/2+1/\text{poly}(n)$ for all $i=1,\dots,L$. 
We find security against aforementioned embezzling attack and attacks involving single banknotes.
While we do not provide a proof of security against all possible attacks involving many banknotes, we believe generic security is likely to hold and leave the formal proof as an open problem.
}

\prlsection{Discussion} 
We have introduced PRDMs as ensembles of quantum states which are indistinguishable from the GHSE by any efficient algorithm, generalizing the concept of PRS to mixed states. Here, the GHSE is the ensemble induced by the $m$-qubit partial trace over $n+m$ qubit Haar random states. 
PRDMs with $m=0$ correspond to PRSs, while PRDMs with $m=\omega(\log n)$ are indistinguishable from the maximally mixed states for any efficient algorithm. 
Further, we show that PRDMs with $m=\omega(\log n)$ are robust to arbitrary unital noise, in contrast to PRS.

PRDMs are a non-trivial generalization as for $m=\omega(\log n)$ they appear to any efficient observer as trivial states, yet can contain a near-maximal amount of entanglement, magic and coherence. \revA{These resources can only extracted when knowing the key of the PRDM.}
With PRDMs, we can completely hide quantum resources in plain sight, where we find a class of \revA{pseudoentangled}, pseudomagic and pseudocoherent ensembles with maximal resource gap of $f(n)=\Theta(n)$ vs $g(n)=0$. 
Pseudoresource ensembles based on PRDMs possess larger gaps than the ones based on PRS, which are limited to $f(n)=\Theta(n)$ vs $g(n)=\omega(\log n)$~\cite{bouland2022quantum,haug2023pseudorandom,gu2023little}. For pure states, one can only hide resources up to an amount of $\omega(\log n)$ due to the existence of efficient property testing algorithms, while PRDMs allow one to hide the asymptotically maximal possible amount of resources.

We show that PRDMs yield noise-robust EFI pairs~\cite{brakerski2022computational}, which remain EFI pairs even when mixed unitary channels are applied. This includes local depolarizing noise, the most commonly used noise model for noisy-intermediate scale quantum computers~\cite{bharti2021noisy,chen2023complexity} and quantum error correction. 
Curiously, EFI pairs can tolerate error rates of up to $p\approx 1/4$, which is higher than the corresponding threshold of surface codes~\cite{bombin2012strong}. Thus, EFI pairs can potentially exist even in noise regimes where quantum error correction fails.
Further, high noise robustness opens up potential application of EFI pairs on near-term quantum computers~\cite{bharti2021noisy}. 
Our proof of noise-robustness is limited to PRDMs with security parameter $\kappa<n$, which future work could improve upon.  The concept of a quantum one-way state generator, similar to the classical one-way function, has been proposed recently~\cite{morimae2022one} and shown to be equivalent to quantum commitments~\cite{batra2024commitments}. Furthermore, quantum commitments are known to be equivalent to EFI~\cite{brakerski2022computational,yan2022general}. Since PRDMs yield noise robust EFI pairs, it would be interesting to explore if the noise robustness can be translated to quantum one way state generators and quantum commitment schemes.

Our work contributes to several new understandings on learning properties of quantum states. Cryptography and machine learning are often seen as opposites~\cite{rivest1991cryptography,kearns1994cryptographic}: cryptography hides patterns, while machine learning aims to reveal them. This connection between cryptography and machine learning can establish fundamental bounds on property testing. 
For pure states, pseudoresources imply that one can efficiently test whether a state has $O(1)$ or $\omega(\log n)$ of a given resource~\cite{bouland2022quantum, gu2023little, haug2023pseudorandom}, where efficient testing for entanglement~\cite{ekert2002direct,bouland2022quantum}, magic~\cite{gross2021schur,haug2022scalable,haug2023efficient} and coherence~\cite{haug2023pseudorandom} has been demonstrated.
However, we find that there is a fundamental difference whether one tests pure state or general mixed states: For mixed states testing resources is not efficient, as it requires superpolynomially many copies. This follows from the existence of resource-rich mixed states which are indistinguishable from the resource-free maximally mixed state.
Entanglement, magic and coherence are the most fundamental resources that separate the quantum and classical world, and directly impact most applications of quantum technologies. Yet, it turns out that these quantum resources are in general not physically observable even when having access to polynomial number of copies and performing arbitrary measurements. 
We show that purity is a fundamental condition needed for testing: States with negligible purity cannot be efficiently tested, while efficient tests are known for pure states. It remains an open question whether efficient tests can exist for inverse polynomial purity. It would also be interesting to explore testing of physical many-body quantum systems such as Gibbs states.

\revA{We show that black-box distillation, i.e. distilling pure resource states from noisy states, is inefficient: %
Entanglement, magic and coherence are inefficient for black-box distillation, requiring a superpolynomial number of input states.}
This implies that efficient distillation of these resources always requires specific knowledge about the input state.  In particular, we show that to efficiently perform black box distillation, one must guarantee that the state has non-negligible purity.

We note that the definition of random mixed states is not unique, in contrast to pure states~\cite{Zyczkowski_2011, Zyczkowski_2001,sarkar2019bures}. 
\revA{For example, as an alternative to GHSE, one could consider the ensemble of random $n$-qubit mixed states $\rho$ with $\text{rank}(\rho)=2^m$ and identical eigenvalues. However, it has been shown that this  ensemble is statistically indistinguishable from the GHSE~\cite{haug2024pseudorandom}. Thus, both ensembles yield the same definition of PRDM.
}

\revA{While PRS require negligible noise, we show that applications of PRS can still function even under noise. In particular, we propose a noise-robust private key quantum money schemes based on PRS~\cite{ji2018pseudorandom} via completeness amplification~\cite{aaronson_quantum_money_2012_arxiv}. The quantum money is secure and is accepted by the bank as long as the noisy PRS has at least 50\% fidelity with the noise-free PRS. 
As a technical contribution of importance beyond our work, we show that the completeness amplification scheme of Ref.~\cite{aaronson_quantum_money_2012_arxiv} can be broken by an embezzling attack, while our modified scheme is secure against such attacks (\SM{}~\ref{sec:quantummoney}). Our modified scheme is also relevant for quantum bitcoins~\cite{jogenfors2019quantum} and other quantum money schemes~\cite{sano2022quantum}.}
We also define memoryless PRS, which are secure to observers without quantum memory. Restrictions to observers without quantum memory can be found in communication scenarios or noisy quantum computers with limited qubit number. We show that these memoryless PRS are robust to unital noise.

\revA{We also prove a separation in security between PRS and PRDM, where PRDMs are secure to a strictly larger class of (inefficient) attacks than PRS. PRS are known to be vulnerable to $\mathsf{PostBQP}$ attacks~\cite{kretschmer2021quantum}, whether such attacks also exist for PRDMs is left as an open question. }

\revA{We note that a follow-up work introduced the notion of verifiable PRDM (VPRDM). They are PRDMs where the correct state preparation can be efficiently verified~\cite{haug2024pseudorandom}, a useful property for many applications, such as encryption, authentication, or hiding quantum resources. For completeness, we give a short review on VPRDMs in Appendix~\ref{sec:VPRDM}. }

Future work can generalize our work to noisy pseudorandom unitaries~\cite{ji2018pseudorandom} and isometries~\cite{ananth2023pseudorandomisometry}. \revA{It would also be interesting to study potential applications of GSHE and PRDM for $m=O(\log n)$.}
Further, we leave proving the noise-robustness of EFI pairs to arbitrary unital noise, and improving the bounds on purity for property testing and black-box distillation as open problems. Finally, we believe PRDM and memoryless pseudorandomness can yield promising applications in cryptography, learning theory, resource theory and experimental demonstrations.

\begin{acknowledgments}
KB thanks Rahul Jain for interesting discussions.
This research is supported by A*STAR C230917003.
NB acknowledges INSPIRE-SHE scholarship by DST, India. The Institute for Quantum Information and Matter is an NSF Physics Frontiers Center. DEK acknowledges funding support from the Agency for Science, Technology and Research (A*STAR) Central Research Fund (CRF) Award.

\end{acknowledgments}

\bibliography{noise}

\begin{thebibliography}{103}%
\makeatletter
\providecommand \@ifxundefined [1]{%
 \@ifx{#1\undefined}
}%
\providecommand \@ifnum [1]{%
 \ifnum #1\expandafter \@firstoftwo
 \else \expandafter \@secondoftwo
 \fi
}%
\providecommand \@ifx [1]{%
 \ifx #1\expandafter \@firstoftwo
 \else \expandafter \@secondoftwo
 \fi
}%
\providecommand \natexlab [1]{#1}%
\providecommand \enquote  [1]{``#1''}%
\providecommand \bibnamefont  [1]{#1}%
\providecommand \bibfnamefont [1]{#1}%
\providecommand \citenamefont [1]{#1}%
\providecommand \href@noop [0]{\@secondoftwo}%
\providecommand \href [0]{\begingroup \@sanitize@url \@href}%
\providecommand \@href[1]{\@@startlink{#1}\@@href}%
\providecommand \@@href[1]{\endgroup#1\@@endlink}%
\providecommand \@sanitize@url [0]{\catcode `\\12\catcode `\$12\catcode
  `\&12\catcode `\#12\catcode `\^12\catcode `\_12\catcode `\%12\relax}%
\providecommand \@@startlink[1]{}%
\providecommand \@@endlink[0]{}%
\providecommand \url  [0]{\begingroup\@sanitize@url \@url }%
\providecommand \@url [1]{\endgroup\@href {#1}{\urlprefix }}%
\providecommand \urlprefix  [0]{URL }%
\providecommand \Eprint [0]{\href }%
\providecommand \doibase [0]{https://doi.org/}%
\providecommand \selectlanguage [0]{\@gobble}%
\providecommand \bibinfo  [0]{\@secondoftwo}%
\providecommand \bibfield  [0]{\@secondoftwo}%
\providecommand \translation [1]{[#1]}%
\providecommand \BibitemOpen [0]{}%
\providecommand \bibitemStop [0]{}%
\providecommand \bibitemNoStop [0]{.\EOS\space}%
\providecommand \EOS [0]{\spacefactor3000\relax}%
\providecommand \BibitemShut  [1]{\csname bibitem#1\endcsname}%
\let\auto@bib@innerbib\@empty
\bibitem [{\citenamefont {Ji}\ \emph {et~al.}(2018)\citenamefont {Ji},
  \citenamefont {Liu},\ and\ \citenamefont {Song}}]{ji2018pseudorandom}%
  \BibitemOpen
  \bibfield  {author} {\bibinfo {author} {\bibfnamefont {Z.}~\bibnamefont
  {Ji}}, \bibinfo {author} {\bibfnamefont {Y.-K.}\ \bibnamefont {Liu}},\ and\
  \bibinfo {author} {\bibfnamefont {F.}~\bibnamefont {Song}},\ }\bibfield
  {title} {\bibinfo {title} {Pseudorandom quantum states},\ }in\ \href
  {https://doi.org/10.1007/978-3-319-96878-0_5} {\emph {\bibinfo {booktitle}
  {Annual International Cryptology Conference}}}\ (\bibinfo {organization}
  {Springer},\ \bibinfo {year} {2018})\ pp.\ \bibinfo {pages}
  {126--152}\BibitemShut {NoStop}%
\bibitem [{\citenamefont {Brakerski}\ \emph {et~al.}(2021)\citenamefont
  {Brakerski}, \citenamefont {Christiano}, \citenamefont {Mahadev},
  \citenamefont {Vazirani},\ and\ \citenamefont
  {Vidick}}]{brakerski2021certifying}%
  \BibitemOpen
  \bibfield  {author} {\bibinfo {author} {\bibfnamefont {Z.}~\bibnamefont
  {Brakerski}}, \bibinfo {author} {\bibfnamefont {P.}~\bibnamefont
  {Christiano}}, \bibinfo {author} {\bibfnamefont {U.}~\bibnamefont {Mahadev}},
  \bibinfo {author} {\bibfnamefont {U.}~\bibnamefont {Vazirani}},\ and\
  \bibinfo {author} {\bibfnamefont {T.}~\bibnamefont {Vidick}},\ }\bibfield
  {title} {\bibinfo {title} {A cryptographic test of quantumness and
  certifiable randomness from a single quantum device},\ }\href
  {https://doi.org/10.1145/3441309} {\bibfield  {journal} {\bibinfo  {journal}
  {J. ACM}\ }\textbf {\bibinfo {volume} {68}},\ \bibinfo {pages} {31} (\bibinfo
  {year} {2021})}\BibitemShut {NoStop}%
\bibitem [{\citenamefont {Arvind}\ \emph {et~al.}(2023)\citenamefont {Arvind},
  \citenamefont {Bharti}, \citenamefont {Khoo}, \citenamefont {Koh},\ and\
  \citenamefont {Kong}}]{arvind2023quantumtugwarrandomness}%
  \BibitemOpen
  \bibfield  {author} {\bibinfo {author} {\bibfnamefont {R.}~\bibnamefont
  {Arvind}}, \bibinfo {author} {\bibfnamefont {K.}~\bibnamefont {Bharti}},
  \bibinfo {author} {\bibfnamefont {J.~Y.}\ \bibnamefont {Khoo}}, \bibinfo
  {author} {\bibfnamefont {D.~E.}\ \bibnamefont {Koh}},\ and\ \bibinfo {author}
  {\bibfnamefont {J.~F.}\ \bibnamefont {Kong}},\ }\bibfield  {title} {\bibinfo
  {title} {A quantum tug of war between randomness and symmetries on
  homogeneous spaces},\ }\href {https://doi.org/10.48550/arXiv.2309.05253}
  {\bibfield  {journal} {\bibinfo  {journal} {arXiv preprint arXiv:2309.05253}\
  } (\bibinfo {year} {2023})}\BibitemShut {NoStop}%
\bibitem [{\citenamefont {Ananth}\ \emph {et~al.}(2023)\citenamefont {Ananth},
  \citenamefont {Gulati}, \citenamefont {Kaleoglu},\ and\ \citenamefont
  {Lin}}]{ananth2023pseudorandomisometry}%
  \BibitemOpen
  \bibfield  {author} {\bibinfo {author} {\bibfnamefont {P.}~\bibnamefont
  {Ananth}}, \bibinfo {author} {\bibfnamefont {A.}~\bibnamefont {Gulati}},
  \bibinfo {author} {\bibfnamefont {F.}~\bibnamefont {Kaleoglu}},\ and\
  \bibinfo {author} {\bibfnamefont {Y.-T.}\ \bibnamefont {Lin}},\ }\bibfield
  {title} {\bibinfo {title} {Pseudorandom isometries},\ }\href
  {https://doi.org/10.48550/arXiv.2311.02901} {\bibfield  {journal} {\bibinfo
  {journal} {arXiv preprint arXiv:2311.02901}\ } (\bibinfo {year}
  {2023})}\BibitemShut {NoStop}%
\bibitem [{\citenamefont {Lu}\ \emph {et~al.}(2023)\citenamefont {Lu},
  \citenamefont {Qin}, \citenamefont {Song}, \citenamefont {Yao},\ and\
  \citenamefont {Zhao}}]{lu2023quantum}%
  \BibitemOpen
  \bibfield  {author} {\bibinfo {author} {\bibfnamefont {C.}~\bibnamefont
  {Lu}}, \bibinfo {author} {\bibfnamefont {M.}~\bibnamefont {Qin}}, \bibinfo
  {author} {\bibfnamefont {F.}~\bibnamefont {Song}}, \bibinfo {author}
  {\bibfnamefont {P.}~\bibnamefont {Yao}},\ and\ \bibinfo {author}
  {\bibfnamefont {M.}~\bibnamefont {Zhao}},\ }\bibfield  {title} {\bibinfo
  {title} {Quantum pseudorandom scramblers},\ }\href
  {https://doi.org/10.48550/arXiv.2309.08941} {\bibfield  {journal} {\bibinfo
  {journal} {arXiv preprint arXiv:2309.08941}\ } (\bibinfo {year}
  {2023})}\BibitemShut {NoStop}%
\bibitem [{\citenamefont {Arnon-Friedman}\ \emph {et~al.}(2023)\citenamefont
  {Arnon-Friedman}, \citenamefont {Brakerski},\ and\ \citenamefont
  {Vidick}}]{arnon2023computational}%
  \BibitemOpen
  \bibfield  {author} {\bibinfo {author} {\bibfnamefont {R.}~\bibnamefont
  {Arnon-Friedman}}, \bibinfo {author} {\bibfnamefont {Z.}~\bibnamefont
  {Brakerski}},\ and\ \bibinfo {author} {\bibfnamefont {T.}~\bibnamefont
  {Vidick}},\ }\bibfield  {title} {\bibinfo {title} {Computational entanglement
  theory},\ }\href {https://doi.org/10.48550/arXiv.2310.02783} {\bibfield
  {journal} {\bibinfo  {journal} {arXiv preprint arXiv:2310.02783}\ } (\bibinfo
  {year} {2023})}\BibitemShut {NoStop}%
\bibitem [{\citenamefont {Kretschmer}(2021)}]{kretschmer2021quantum}%
  \BibitemOpen
  \bibfield  {author} {\bibinfo {author} {\bibfnamefont {W.}~\bibnamefont
  {Kretschmer}},\ }\bibfield  {title} {\bibinfo {title} {{Quantum
  Pseudorandomness and Classical Complexity}},\ }in\ \href
  {https://doi.org/10.4230/LIPIcs.TQC.2021.2} {\emph {\bibinfo {booktitle}
  {16th Conference on the Theory of Quantum Computation, Communication and
  Cryptography (TQC 2021)}}},\ \bibinfo {series} {Leibniz International
  Proceedings in Informatics (LIPIcs)}, Vol.\ \bibinfo {volume} {197},\
  \bibinfo {editor} {edited by\ \bibinfo {editor} {\bibfnamefont {M.-H.}\
  \bibnamefont {Hsieh}}}\ (\bibinfo  {publisher} {Schloss Dagstuhl --
  Leibniz-Zentrum f{\"u}r Informatik},\ \bibinfo {address} {Dagstuhl,
  Germany},\ \bibinfo {year} {2021})\ pp.\ \bibinfo {pages}
  {2:1--2:20}\BibitemShut {NoStop}%
\bibitem [{\citenamefont {Katz}\ and\ \citenamefont
  {Lindell}(2007)}]{katz2007introduction}%
  \BibitemOpen
  \bibfield  {author} {\bibinfo {author} {\bibfnamefont {J.}~\bibnamefont
  {Katz}}\ and\ \bibinfo {author} {\bibfnamefont {Y.}~\bibnamefont {Lindell}},\
  }\href@noop {} {\emph {\bibinfo {title} {Introduction to modern cryptography:
  principles and protocols}}}\ (\bibinfo  {publisher} {Chapman and Hall/CRC},\
  \bibinfo {year} {2007})\BibitemShut {NoStop}%
\bibitem [{\citenamefont {Chitambar}\ and\ \citenamefont
  {Gour}(2019)}]{chitambar2019quantum}%
  \BibitemOpen
  \bibfield  {author} {\bibinfo {author} {\bibfnamefont {E.}~\bibnamefont
  {Chitambar}}\ and\ \bibinfo {author} {\bibfnamefont {G.}~\bibnamefont
  {Gour}},\ }\bibfield  {title} {\bibinfo {title} {Quantum resource theories},\
  }\href {https://doi.org/10.1103/RevModPhys.91.025001} {\bibfield  {journal}
  {\bibinfo  {journal} {Rev. Mod. Phys.}\ }\textbf {\bibinfo {volume} {91}},\
  \bibinfo {pages} {025001} (\bibinfo {year} {2019})}\BibitemShut {NoStop}%
\bibitem [{\citenamefont {Aaronson}\ \emph {et~al.}(2022)\citenamefont
  {Aaronson}, \citenamefont {Bouland}, \citenamefont {Fefferman}, \citenamefont
  {Ghosh}, \citenamefont {Vazirani}, \citenamefont {Zhang},\ and\ \citenamefont
  {Zhou}}]{bouland2022quantum}%
  \BibitemOpen
  \bibfield  {author} {\bibinfo {author} {\bibfnamefont {S.}~\bibnamefont
  {Aaronson}}, \bibinfo {author} {\bibfnamefont {A.}~\bibnamefont {Bouland}},
  \bibinfo {author} {\bibfnamefont {B.}~\bibnamefont {Fefferman}}, \bibinfo
  {author} {\bibfnamefont {S.}~\bibnamefont {Ghosh}}, \bibinfo {author}
  {\bibfnamefont {U.}~\bibnamefont {Vazirani}}, \bibinfo {author}
  {\bibfnamefont {C.}~\bibnamefont {Zhang}},\ and\ \bibinfo {author}
  {\bibfnamefont {Z.}~\bibnamefont {Zhou}},\ }\bibfield  {title} {\bibinfo
  {title} {Quantum pseudoentanglement},\ }\href
  {https://doi.org/10.48550/arXiv.2211.00747} {\bibfield  {journal} {\bibinfo
  {journal} {arXiv preprint arXiv:2211.00747}\ } (\bibinfo {year}
  {2022})}\BibitemShut {NoStop}%
\bibitem [{\citenamefont {Haug}\ \emph {et~al.}(2023)\citenamefont {Haug},
  \citenamefont {Bharti},\ and\ \citenamefont {Koh}}]{haug2023pseudorandom}%
  \BibitemOpen
  \bibfield  {author} {\bibinfo {author} {\bibfnamefont {T.}~\bibnamefont
  {Haug}}, \bibinfo {author} {\bibfnamefont {K.}~\bibnamefont {Bharti}},\ and\
  \bibinfo {author} {\bibfnamefont {D.~E.}\ \bibnamefont {Koh}},\ }\bibfield
  {title} {\bibinfo {title} {Pseudorandom unitaries are neither real nor sparse
  nor noise-robust},\ }\href {https://doi.org/10.48550/arXiv.2306.11677}
  {\bibfield  {journal} {\bibinfo  {journal} {arXiv preprint arXiv:2306.11677}\
  } (\bibinfo {year} {2023})}\BibitemShut {NoStop}%
\bibitem [{\citenamefont {Gu}\ \emph {et~al.}(2024)\citenamefont {Gu},
  \citenamefont {Leone}, \citenamefont {Ghosh}, \citenamefont {Eisert},
  \citenamefont {Yelin},\ and\ \citenamefont {Quek}}]{gu2023little}%
  \BibitemOpen
  \bibfield  {author} {\bibinfo {author} {\bibfnamefont {A.}~\bibnamefont
  {Gu}}, \bibinfo {author} {\bibfnamefont {L.}~\bibnamefont {Leone}}, \bibinfo
  {author} {\bibfnamefont {S.}~\bibnamefont {Ghosh}}, \bibinfo {author}
  {\bibfnamefont {J.}~\bibnamefont {Eisert}}, \bibinfo {author} {\bibfnamefont
  {S.~F.}\ \bibnamefont {Yelin}},\ and\ \bibinfo {author} {\bibfnamefont
  {Y.}~\bibnamefont {Quek}},\ }\bibfield  {title} {\bibinfo {title}
  {Pseudomagic quantum states},\ }\href
  {https://doi.org/10.1103/PhysRevLett.132.210602} {\bibfield  {journal}
  {\bibinfo  {journal} {Physical Review Letters}\ }\textbf {\bibinfo {volume}
  {132}},\ \bibinfo {pages} {210602} (\bibinfo {year} {2024})}\BibitemShut
  {NoStop}%
\bibitem [{\citenamefont {Montanaro}\ and\ \citenamefont
  {de~Wolf}(2016)}]{montanaro2013survey}%
  \BibitemOpen
  \bibfield  {author} {\bibinfo {author} {\bibfnamefont {A.}~\bibnamefont
  {Montanaro}}\ and\ \bibinfo {author} {\bibfnamefont {R.}~\bibnamefont
  {de~Wolf}},\ }\href {https://doi.org/10.4086/toc.gs.2016.007} {\emph
  {\bibinfo {title} {A Survey of Quantum Property Testing}}},\ \bibinfo
  {series} {Graduate Surveys}\ No.~\bibinfo {number} {7}\ (\bibinfo
  {publisher} {Theory of Computing Library},\ \bibinfo {year} {2016})\ pp.\
  \bibinfo {pages} {1--81}\BibitemShut {NoStop}%
\bibitem [{\citenamefont {Wright}(2016)}]{wright2016learn}%
  \BibitemOpen
  \bibfield  {author} {\bibinfo {author} {\bibfnamefont {J.}~\bibnamefont
  {Wright}},\ }\emph {\bibinfo {title} {How to learn a quantum state}},\
  \href@noop {} {Ph.D. thesis},\ \bibinfo  {school} {Carnegie Mellon
  University} (\bibinfo {year} {2016})\BibitemShut {NoStop}%
\bibitem [{\citenamefont {Bennett}\ \emph
  {et~al.}(1996{\natexlab{a}})\citenamefont {Bennett}, \citenamefont
  {Bernstein}, \citenamefont {Popescu},\ and\ \citenamefont
  {Schumacher}}]{bennett1996concentrating}%
  \BibitemOpen
  \bibfield  {author} {\bibinfo {author} {\bibfnamefont {C.~H.}\ \bibnamefont
  {Bennett}}, \bibinfo {author} {\bibfnamefont {H.~J.}\ \bibnamefont
  {Bernstein}}, \bibinfo {author} {\bibfnamefont {S.}~\bibnamefont {Popescu}},\
  and\ \bibinfo {author} {\bibfnamefont {B.}~\bibnamefont {Schumacher}},\
  }\bibfield  {title} {\bibinfo {title} {Concentrating partial entanglement by
  local operations},\ }\href {https://doi.org/10.1103/PhysRevA.53.2046}
  {\bibfield  {journal} {\bibinfo  {journal} {Phys. Rev. A}\ }\textbf {\bibinfo
  {volume} {53}},\ \bibinfo {pages} {2046} (\bibinfo {year}
  {1996}{\natexlab{a}})}\BibitemShut {NoStop}%
\bibitem [{\citenamefont {Regula}\ and\ \citenamefont
  {Takagi}(2021)}]{regula2021fundamental}%
  \BibitemOpen
  \bibfield  {author} {\bibinfo {author} {\bibfnamefont {B.}~\bibnamefont
  {Regula}}\ and\ \bibinfo {author} {\bibfnamefont {R.}~\bibnamefont
  {Takagi}},\ }\bibfield  {title} {\bibinfo {title} {Fundamental limitations on
  distillation of quantum channel resources},\ }\href
  {https://doi.org/10.1038/s41467-021-24699-0} {\bibfield  {journal} {\bibinfo
  {journal} {Nature Communications}\ }\textbf {\bibinfo {volume} {12}},\
  \bibinfo {pages} {4411} (\bibinfo {year} {2021})}\BibitemShut {NoStop}%
\bibitem [{\citenamefont {Braunstein}(1996)}]{braunstein1996geometry}%
  \BibitemOpen
  \bibfield  {author} {\bibinfo {author} {\bibfnamefont {S.~L.}\ \bibnamefont
  {Braunstein}},\ }\bibfield  {title} {\bibinfo {title} {Geometry of quantum
  inference},\ }\href {https://doi.org/10.1016/0375-9601(96)00365-9} {\bibfield
   {journal} {\bibinfo  {journal} {Physics Letters A}\ }\textbf {\bibinfo
  {volume} {219}},\ \bibinfo {pages} {169} (\bibinfo {year}
  {1996})}\BibitemShut {NoStop}%
\bibitem [{\citenamefont {Hall}(1998)}]{hall1998random}%
  \BibitemOpen
  \bibfield  {author} {\bibinfo {author} {\bibfnamefont {M.~J.}\ \bibnamefont
  {Hall}},\ }\bibfield  {title} {\bibinfo {title} {Random quantum correlations
  and density operator distributions},\ }\href
  {https://doi.org/10.1016/S0375-9601(98)00190-X} {\bibfield  {journal}
  {\bibinfo  {journal} {Physics Letters A}\ }\textbf {\bibinfo {volume}
  {242}},\ \bibinfo {pages} {123} (\bibinfo {year} {1998})}\BibitemShut
  {NoStop}%
\bibitem [{\citenamefont {Zyczkowski}\ and\ \citenamefont
  {Sommers}(2001)}]{Zyczkowski_2001}%
  \BibitemOpen
  \bibfield  {author} {\bibinfo {author} {\bibfnamefont {K.}~\bibnamefont
  {Zyczkowski}}\ and\ \bibinfo {author} {\bibfnamefont {H.-J.}\ \bibnamefont
  {Sommers}},\ }\bibfield  {title} {\bibinfo {title} {Induced measures in the
  space of mixed quantum states},\ }\href
  {https://doi.org/10.1088/0305-4470/34/35/335} {\bibfield  {journal} {\bibinfo
   {journal} {Journal of Physics A: Mathematical and General}\ }\textbf
  {\bibinfo {volume} {34}},\ \bibinfo {pages} {7111} (\bibinfo {year}
  {2001})}\BibitemShut {NoStop}%
\bibitem [{\citenamefont {Hayden}\ \emph {et~al.}(2006)\citenamefont {Hayden},
  \citenamefont {Leung},\ and\ \citenamefont {Winter}}]{hayden2006aspects}%
  \BibitemOpen
  \bibfield  {author} {\bibinfo {author} {\bibfnamefont {P.}~\bibnamefont
  {Hayden}}, \bibinfo {author} {\bibfnamefont {D.~W.}\ \bibnamefont {Leung}},\
  and\ \bibinfo {author} {\bibfnamefont {A.}~\bibnamefont {Winter}},\
  }\bibfield  {title} {\bibinfo {title} {Aspects of generic entanglement},\
  }\href {https://doi.org/10.1007/s00220-006-1535-6} {\bibfield  {journal}
  {\bibinfo  {journal} {Communications in mathematical physics}\ }\textbf
  {\bibinfo {volume} {265}},\ \bibinfo {pages} {95} (\bibinfo {year}
  {2006})}\BibitemShut {NoStop}%
\bibitem [{\citenamefont {Zyczkowski}\ \emph {et~al.}(2011)\citenamefont
  {Zyczkowski}, \citenamefont {Penson}, \citenamefont {Nechita},\ and\
  \citenamefont {Collins}}]{Zyczkowski_2011}%
  \BibitemOpen
  \bibfield  {author} {\bibinfo {author} {\bibfnamefont {K.}~\bibnamefont
  {Zyczkowski}}, \bibinfo {author} {\bibfnamefont {K.~A.}\ \bibnamefont
  {Penson}}, \bibinfo {author} {\bibfnamefont {I.}~\bibnamefont {Nechita}},\
  and\ \bibinfo {author} {\bibfnamefont {B.}~\bibnamefont {Collins}},\
  }\bibfield  {title} {\bibinfo {title} {{Generating random density
  matrices}},\ }\href {https://doi.org/10.1063/1.3595693} {\bibfield  {journal}
  {\bibinfo  {journal} {Journal of Mathematical Physics}\ }\textbf {\bibinfo
  {volume} {52}},\ \bibinfo {pages} {062201} (\bibinfo {year}
  {2011})}\BibitemShut {NoStop}%
\bibitem [{\citenamefont {Sarkar}\ and\ \citenamefont
  {Kumar}(2019)}]{sarkar2019bures}%
  \BibitemOpen
  \bibfield  {author} {\bibinfo {author} {\bibfnamefont {A.}~\bibnamefont
  {Sarkar}}\ and\ \bibinfo {author} {\bibfnamefont {S.}~\bibnamefont {Kumar}},\
  }\bibfield  {title} {\bibinfo {title} {Bures–{H}all ensemble: spectral
  densities and average entropies},\ }\href
  {https://doi.org/10.1088/1751-8121/ab2675} {\bibfield  {journal} {\bibinfo
  {journal} {Journal of Physics A: Mathematical and Theoretical}\ }\textbf
  {\bibinfo {volume} {52}},\ \bibinfo {pages} {295203} (\bibinfo {year}
  {2019})}\BibitemShut {NoStop}%
\bibitem [{\citenamefont {Haug}\ and\ \citenamefont
  {Piroli}(2023)}]{haug2023stabilizer}%
  \BibitemOpen
  \bibfield  {author} {\bibinfo {author} {\bibfnamefont {T.}~\bibnamefont
  {Haug}}\ and\ \bibinfo {author} {\bibfnamefont {L.}~\bibnamefont {Piroli}},\
  }\bibfield  {title} {\bibinfo {title} {Stabilizer entropies and
  nonstabilizerness monotones},\ }\href
  {https://doi.org/10.22331/q-2023-08-28-1092} {\bibfield  {journal} {\bibinfo
  {journal} {{Quantum}}\ }\textbf {\bibinfo {volume} {7}},\ \bibinfo {pages}
  {1092} (\bibinfo {year} {2023})}\BibitemShut {NoStop}%
\bibitem [{\citenamefont {Brakerski}\ \emph {et~al.}(2022)\citenamefont
  {Brakerski}, \citenamefont {Canetti},\ and\ \citenamefont
  {Qian}}]{brakerski2022computational}%
  \BibitemOpen
  \bibfield  {author} {\bibinfo {author} {\bibfnamefont {Z.}~\bibnamefont
  {Brakerski}}, \bibinfo {author} {\bibfnamefont {R.}~\bibnamefont {Canetti}},\
  and\ \bibinfo {author} {\bibfnamefont {L.}~\bibnamefont {Qian}},\ }\bibfield
  {title} {\bibinfo {title} {On the computational hardness needed for quantum
  cryptography},\ }\href {https://doi.org/10.48550/arXiv.2209.04101} {\bibfield
   {journal} {\bibinfo  {journal} {arXiv preprint arXiv:2209.04101}\ }
  (\bibinfo {year} {2022})}\BibitemShut {NoStop}%
\bibitem [{\citenamefont {Yan}(2022)}]{yan2022general}%
  \BibitemOpen
  \bibfield  {author} {\bibinfo {author} {\bibfnamefont {J.}~\bibnamefont
  {Yan}},\ }\bibfield  {title} {\bibinfo {title} {General properties of quantum
  bit commitments},\ }in\ \href {https://doi.org/10.1007/978-3-031-22972-5_22}
  {\emph {\bibinfo {booktitle} {International Conference on the Theory and
  Application of Cryptology and Information Security}}}\ (\bibinfo
  {organization} {Springer},\ \bibinfo {year} {2022})\ pp.\ \bibinfo {pages}
  {628--657}\BibitemShut {NoStop}%
\bibitem [{\citenamefont {Bartusek}\ \emph {et~al.}(2021)\citenamefont
  {Bartusek}, \citenamefont {Coladangelo}, \citenamefont {Khurana},\ and\
  \citenamefont {Ma}}]{bartusek2021one}%
  \BibitemOpen
  \bibfield  {author} {\bibinfo {author} {\bibfnamefont {J.}~\bibnamefont
  {Bartusek}}, \bibinfo {author} {\bibfnamefont {A.}~\bibnamefont
  {Coladangelo}}, \bibinfo {author} {\bibfnamefont {D.}~\bibnamefont
  {Khurana}},\ and\ \bibinfo {author} {\bibfnamefont {F.}~\bibnamefont {Ma}},\
  }\bibfield  {title} {\bibinfo {title} {One-way functions imply secure
  computation in a quantum world},\ }in\ \href
  {https://doi.org/10.1007/978-3-030-84242-0_17} {\emph {\bibinfo {booktitle}
  {Advances in Cryptology--CRYPTO 2021: 41st Annual International Cryptology
  Conference, CRYPTO 2021, Virtual Event, August 16--20, 2021, Proceedings,
  Part I 41}}}\ (\bibinfo {organization} {Springer},\ \bibinfo {year} {2021})\
  pp.\ \bibinfo {pages} {467--496}\BibitemShut {NoStop}%
\bibitem [{\citenamefont {Ananth}\ \emph
  {et~al.}(2022{\natexlab{a}})\citenamefont {Ananth}, \citenamefont {Qian},\
  and\ \citenamefont {Yuen}}]{ananth2022cryptography}%
  \BibitemOpen
  \bibfield  {author} {\bibinfo {author} {\bibfnamefont {P.}~\bibnamefont
  {Ananth}}, \bibinfo {author} {\bibfnamefont {L.}~\bibnamefont {Qian}},\ and\
  \bibinfo {author} {\bibfnamefont {H.}~\bibnamefont {Yuen}},\ }\bibfield
  {title} {\bibinfo {title} {Cryptography from pseudorandom quantum states},\
  }in\ \href {https://doi.org/10.1007/978-3-031-15802-5_8} {\emph {\bibinfo
  {booktitle} {Advances in Cryptology--CRYPTO 2022: 42nd Annual International
  Cryptology Conference, CRYPTO 2022, Santa Barbara, CA, USA, August 15--18,
  2022, Proceedings, Part I}}}\ (\bibinfo {organization} {Springer},\ \bibinfo
  {year} {2022})\ pp.\ \bibinfo {pages} {208--236}\BibitemShut {NoStop}%
\bibitem [{\citenamefont {Ananth}\ \emph {et~al.}(2021)\citenamefont {Ananth},
  \citenamefont {Chung},\ and\ \citenamefont {Placa}}]{ananth2021concurrent}%
  \BibitemOpen
  \bibfield  {author} {\bibinfo {author} {\bibfnamefont {P.}~\bibnamefont
  {Ananth}}, \bibinfo {author} {\bibfnamefont {K.-M.}\ \bibnamefont {Chung}},\
  and\ \bibinfo {author} {\bibfnamefont {R.~L.~L.}\ \bibnamefont {Placa}},\
  }\bibfield  {title} {\bibinfo {title} {On the concurrent composition of
  quantum zero-knowledge},\ }in\ \href
  {https://doi.org/10.1007/978-3-030-84242-0_13} {\emph {\bibinfo {booktitle}
  {Advances in Cryptology--CRYPTO 2021: 41st Annual International Cryptology
  Conference, CRYPTO 2021, Virtual Event, August 16--20, 2021, Proceedings,
  Part I 41}}}\ (\bibinfo {organization} {Springer},\ \bibinfo {year} {2021})\
  pp.\ \bibinfo {pages} {346--374}\BibitemShut {NoStop}%
\bibitem [{\citenamefont {Aaronson}\ and\ \citenamefont
  {Christiano}(2012)}]{aaronson_quantum_money_2012_arxiv}%
  \BibitemOpen
  \bibfield  {author} {\bibinfo {author} {\bibfnamefont {S.}~\bibnamefont
  {Aaronson}}\ and\ \bibinfo {author} {\bibfnamefont {P.}~\bibnamefont
  {Christiano}},\ }\bibfield  {title} {\bibinfo {title} {Quantum money from
  hidden subspaces},\ }\bibfield  {journal} {\bibinfo  {journal} {arXiv
  preprint arXiv:1203.4740}\ }\href {https://doi.org/10.48550/arXiv.1203.4740}
  {10.48550/arXiv.1203.4740} (\bibinfo {year} {2012})\BibitemShut {NoStop}%
\bibitem [{\citenamefont {Bravyi}\ and\ \citenamefont
  {Kitaev}(2005)}]{bravyi2005universal}%
  \BibitemOpen
  \bibfield  {author} {\bibinfo {author} {\bibfnamefont {S.}~\bibnamefont
  {Bravyi}}\ and\ \bibinfo {author} {\bibfnamefont {A.}~\bibnamefont
  {Kitaev}},\ }\bibfield  {title} {\bibinfo {title} {Universal quantum
  computation with ideal {C}lifford gates and noisy ancillas},\ }\href
  {https://doi.org/10.1103/PhysRevA.71.022316} {\bibfield  {journal} {\bibinfo
  {journal} {Phys. Rev. A}\ }\textbf {\bibinfo {volume} {71}},\ \bibinfo
  {pages} {022316} (\bibinfo {year} {2005})}\BibitemShut {NoStop}%
\bibitem [{\citenamefont {Howard}\ and\ \citenamefont
  {Campbell}(2017)}]{howard2017application}%
  \BibitemOpen
  \bibfield  {author} {\bibinfo {author} {\bibfnamefont {M.}~\bibnamefont
  {Howard}}\ and\ \bibinfo {author} {\bibfnamefont {E.}~\bibnamefont
  {Campbell}},\ }\bibfield  {title} {\bibinfo {title} {Application of a
  resource theory for magic states to fault-tolerant quantum computing},\
  }\href {https://doi.org/10.1103/PhysRevLett.118.090501} {\bibfield  {journal}
  {\bibinfo  {journal} {Phys. Rev. Lett.}\ }\textbf {\bibinfo {volume} {118}},\
  \bibinfo {pages} {090501} (\bibinfo {year} {2017})}\BibitemShut {NoStop}%
\bibitem [{\citenamefont {Liu}\ and\ \citenamefont
  {Winter}(2022)}]{liu2022many}%
  \BibitemOpen
  \bibfield  {author} {\bibinfo {author} {\bibfnamefont {Z.-W.}\ \bibnamefont
  {Liu}}\ and\ \bibinfo {author} {\bibfnamefont {A.}~\bibnamefont {Winter}},\
  }\bibfield  {title} {\bibinfo {title} {Many-body quantum magic},\ }\href
  {https://doi.org/10.1103/PRXQuantum.3.020333} {\bibfield  {journal} {\bibinfo
   {journal} {PRX Quantum}\ }\textbf {\bibinfo {volume} {3}},\ \bibinfo {pages}
  {020333} (\bibinfo {year} {2022})}\BibitemShut {NoStop}%
\bibitem [{\citenamefont {Baumgratz}\ \emph {et~al.}(2014)\citenamefont
  {Baumgratz}, \citenamefont {Cramer},\ and\ \citenamefont
  {Plenio}}]{baumgratz2014quantifying}%
  \BibitemOpen
  \bibfield  {author} {\bibinfo {author} {\bibfnamefont {T.}~\bibnamefont
  {Baumgratz}}, \bibinfo {author} {\bibfnamefont {M.}~\bibnamefont {Cramer}},\
  and\ \bibinfo {author} {\bibfnamefont {M.~B.}\ \bibnamefont {Plenio}},\
  }\bibfield  {title} {\bibinfo {title} {Quantifying coherence},\ }\href
  {https://doi.org/10.1103/PhysRevLett.113.140401} {\bibfield  {journal}
  {\bibinfo  {journal} {Phys. Rev. Lett.}\ }\textbf {\bibinfo {volume} {113}},\
  \bibinfo {pages} {140401} (\bibinfo {year} {2014})}\BibitemShut {NoStop}%
\bibitem [{\citenamefont {Streltsov}\ \emph {et~al.}(2017)\citenamefont
  {Streltsov}, \citenamefont {Adesso},\ and\ \citenamefont
  {Plenio}}]{streltsov2017colloquium}%
  \BibitemOpen
  \bibfield  {author} {\bibinfo {author} {\bibfnamefont {A.}~\bibnamefont
  {Streltsov}}, \bibinfo {author} {\bibfnamefont {G.}~\bibnamefont {Adesso}},\
  and\ \bibinfo {author} {\bibfnamefont {M.~B.}\ \bibnamefont {Plenio}},\
  }\bibfield  {title} {\bibinfo {title} {Colloquium: Quantum coherence as a
  resource},\ }\href {https://doi.org/10.1103/RevModPhys.89.041003} {\bibfield
  {journal} {\bibinfo  {journal} {Rev. Mod. Phys.}\ }\textbf {\bibinfo {volume}
  {89}},\ \bibinfo {pages} {041003} (\bibinfo {year} {2017})}\BibitemShut
  {NoStop}%
\bibitem [{\citenamefont {Plenio}\ and\ \citenamefont
  {Virmani}(2005)}]{plenio2005introduction}%
  \BibitemOpen
  \bibfield  {author} {\bibinfo {author} {\bibfnamefont {M.~B.}\ \bibnamefont
  {Plenio}}\ and\ \bibinfo {author} {\bibfnamefont {S.}~\bibnamefont
  {Virmani}},\ }\bibfield  {title} {\bibinfo {title} {An introduction to
  entanglement measures},\ }\href@noop {} {\bibfield  {journal} {\bibinfo
  {journal} {arXiv preprint quant-ph/0504163}\ } (\bibinfo {year}
  {2005})}\BibitemShut {NoStop}%
\bibitem [{\citenamefont {Singh}\ \emph {et~al.}(2016)\citenamefont {Singh},
  \citenamefont {Zhang},\ and\ \citenamefont {Pati}}]{singh2016average}%
  \BibitemOpen
  \bibfield  {author} {\bibinfo {author} {\bibfnamefont {U.}~\bibnamefont
  {Singh}}, \bibinfo {author} {\bibfnamefont {L.}~\bibnamefont {Zhang}},\ and\
  \bibinfo {author} {\bibfnamefont {A.~K.}\ \bibnamefont {Pati}},\ }\bibfield
  {title} {\bibinfo {title} {Average coherence and its typicality for random
  pure states},\ }\href {https://doi.org/10.1103/PhysRevA.93.032125} {\bibfield
   {journal} {\bibinfo  {journal} {Phys. Rev. A}\ }\textbf {\bibinfo {volume}
  {93}},\ \bibinfo {pages} {032125} (\bibinfo {year} {2016})}\BibitemShut
  {NoStop}%
\bibitem [{\citenamefont {Liu}\ \emph {et~al.}(2019)\citenamefont {Liu},
  \citenamefont {Yuan}, \citenamefont {Lu},\ and\ \citenamefont
  {Wang}}]{liu2019quantum}%
  \BibitemOpen
  \bibfield  {author} {\bibinfo {author} {\bibfnamefont {J.}~\bibnamefont
  {Liu}}, \bibinfo {author} {\bibfnamefont {H.}~\bibnamefont {Yuan}}, \bibinfo
  {author} {\bibfnamefont {X.-M.}\ \bibnamefont {Lu}},\ and\ \bibinfo {author}
  {\bibfnamefont {X.}~\bibnamefont {Wang}},\ }\bibfield  {title} {\bibinfo
  {title} {Quantum {F}isher information matrix and multiparameter estimation},\
  }\href {https://doi.org/10.1088/1751-8121/ab5d4d} {\bibfield  {journal}
  {\bibinfo  {journal} {Journal of Physics A: Mathematical and Theoretical}\
  }\textbf {\bibinfo {volume} {53}},\ \bibinfo {pages} {023001} (\bibinfo
  {year} {2019})}\BibitemShut {NoStop}%
\bibitem [{\citenamefont {Smith}\ and\ \citenamefont
  {Leung}(2006)}]{smith2006typical}%
  \BibitemOpen
  \bibfield  {author} {\bibinfo {author} {\bibfnamefont {G.}~\bibnamefont
  {Smith}}\ and\ \bibinfo {author} {\bibfnamefont {D.}~\bibnamefont {Leung}},\
  }\bibfield  {title} {\bibinfo {title} {Typical entanglement of stabilizer
  states},\ }\href {https://doi.org/10.1103/PhysRevA.74.062314} {\bibfield
  {journal} {\bibinfo  {journal} {Phys. Rev. A}\ }\textbf {\bibinfo {volume}
  {74}},\ \bibinfo {pages} {062314} (\bibinfo {year} {2006})}\BibitemShut
  {NoStop}%
\bibitem [{\citenamefont {Shapourian}\ \emph {et~al.}(2021)\citenamefont
  {Shapourian}, \citenamefont {Liu}, \citenamefont {Kudler-Flam},\ and\
  \citenamefont {Vishwanath}}]{shapourian2021entanglement}%
  \BibitemOpen
  \bibfield  {author} {\bibinfo {author} {\bibfnamefont {H.}~\bibnamefont
  {Shapourian}}, \bibinfo {author} {\bibfnamefont {S.}~\bibnamefont {Liu}},
  \bibinfo {author} {\bibfnamefont {J.}~\bibnamefont {Kudler-Flam}},\ and\
  \bibinfo {author} {\bibfnamefont {A.}~\bibnamefont {Vishwanath}},\ }\bibfield
   {title} {\bibinfo {title} {Entanglement negativity spectrum of random mixed
  states: A diagrammatic approach},\ }\href
  {https://doi.org/10.1103/PRXQuantum.2.030347} {\bibfield  {journal} {\bibinfo
   {journal} {PRX Quantum}\ }\textbf {\bibinfo {volume} {2}},\ \bibinfo {pages}
  {030347} (\bibinfo {year} {2021})}\BibitemShut {NoStop}%
\bibitem [{\citenamefont {Bhosale}\ \emph {et~al.}(2012)\citenamefont
  {Bhosale}, \citenamefont {Tomsovic},\ and\ \citenamefont
  {Lakshminarayan}}]{bhosale2012entanglement}%
  \BibitemOpen
  \bibfield  {author} {\bibinfo {author} {\bibfnamefont {U.~T.}\ \bibnamefont
  {Bhosale}}, \bibinfo {author} {\bibfnamefont {S.}~\bibnamefont {Tomsovic}},\
  and\ \bibinfo {author} {\bibfnamefont {A.}~\bibnamefont {Lakshminarayan}},\
  }\bibfield  {title} {\bibinfo {title} {Entanglement between two subsystems,
  the {W}igner semicircle and extreme-value statistics},\ }\href
  {https://doi.org/10.1103/PhysRevA.85.062331} {\bibfield  {journal} {\bibinfo
  {journal} {Phys. Rev. A}\ }\textbf {\bibinfo {volume} {85}},\ \bibinfo
  {pages} {062331} (\bibinfo {year} {2012})}\BibitemShut {NoStop}%
\bibitem [{\citenamefont {Nielsen}\ and\ \citenamefont
  {Chuang}(2010)}]{nielsen2002quantum}%
  \BibitemOpen
  \bibfield  {author} {\bibinfo {author} {\bibfnamefont {M.~A.}\ \bibnamefont
  {Nielsen}}\ and\ \bibinfo {author} {\bibfnamefont {I.~L.}\ \bibnamefont
  {Chuang}},\ }\href {https://doi.org/10.1017/CBO9780511976667} {\emph
  {\bibinfo {title} {Quantum computation and quantum information}}}\ (\bibinfo
  {publisher} {Cambridge University Press},\ \bibinfo {year}
  {2010})\BibitemShut {NoStop}%
\bibitem [{\citenamefont {Barenco}\ \emph {et~al.}(1997)\citenamefont
  {Barenco}, \citenamefont {Berthiaume}, \citenamefont {Deutsch}, \citenamefont
  {Ekert}, \citenamefont {Jozsa},\ and\ \citenamefont
  {Macchiavello}}]{barenco1997stabilization}%
  \BibitemOpen
  \bibfield  {author} {\bibinfo {author} {\bibfnamefont {A.}~\bibnamefont
  {Barenco}}, \bibinfo {author} {\bibfnamefont {A.}~\bibnamefont {Berthiaume}},
  \bibinfo {author} {\bibfnamefont {D.}~\bibnamefont {Deutsch}}, \bibinfo
  {author} {\bibfnamefont {A.}~\bibnamefont {Ekert}}, \bibinfo {author}
  {\bibfnamefont {R.}~\bibnamefont {Jozsa}},\ and\ \bibinfo {author}
  {\bibfnamefont {C.}~\bibnamefont {Macchiavello}},\ }\bibfield  {title}
  {\bibinfo {title} {Stabilization of quantum computations by symmetrization},\
  }\href {https://doi.org/10.1137/S0097539796302452} {\bibfield  {journal}
  {\bibinfo  {journal} {SIAM Journal on Computing}\ }\textbf {\bibinfo {volume}
  {26}},\ \bibinfo {pages} {1541} (\bibinfo {year} {1997})}\BibitemShut
  {NoStop}%
\bibitem [{\citenamefont {Aaronson}(2005)}]{aaronson2005quantum}%
  \BibitemOpen
  \bibfield  {author} {\bibinfo {author} {\bibfnamefont {S.}~\bibnamefont
  {Aaronson}},\ }\bibfield  {title} {\bibinfo {title} {Quantum computing,
  postselection, and probabilistic polynomial-time},\ }\href
  {https://doi.org/10.1098/rspa.2005.1546} {\bibfield  {journal} {\bibinfo
  {journal} {Proceedings of the Royal Society A: Mathematical, Physical and
  Engineering Sciences}\ }\textbf {\bibinfo {volume} {461}},\ \bibinfo {pages}
  {3473} (\bibinfo {year} {2005})}\BibitemShut {NoStop}%
\bibitem [{\citenamefont {Schuster}\ \emph {et~al.}(2024)\citenamefont
  {Schuster}, \citenamefont {Haferkamp},\ and\ \citenamefont
  {Huang}}]{schuster2024randomunitariesextremelylow}%
  \BibitemOpen
  \bibfield  {author} {\bibinfo {author} {\bibfnamefont {T.}~\bibnamefont
  {Schuster}}, \bibinfo {author} {\bibfnamefont {J.}~\bibnamefont
  {Haferkamp}},\ and\ \bibinfo {author} {\bibfnamefont {H.-Y.}\ \bibnamefont
  {Huang}},\ }\bibfield  {title} {\bibinfo {title} {Random unitaries in
  extremely low depth},\ }\href {https://arxiv.org/abs/2407.07754} {\bibfield
  {journal} {\bibinfo  {journal} {arXiv preprint arXiv:2407.07754}\ } (\bibinfo
  {year} {2024})}\BibitemShut {NoStop}%
\bibitem [{\citenamefont {Ma}\ and\ \citenamefont
  {Huang}(2024)}]{ma2024construct}%
  \BibitemOpen
  \bibfield  {author} {\bibinfo {author} {\bibfnamefont {F.}~\bibnamefont
  {Ma}}\ and\ \bibinfo {author} {\bibfnamefont {H.-Y.}\ \bibnamefont {Huang}},\
  }\bibfield  {title} {\bibinfo {title} {How to construct random unitaries},\
  }\bibfield  {journal} {\bibinfo  {journal} {arXiv preprint arXiv:2410.10116}\
  }\href {https://doi.org/10.48550/arXiv.2410.10116}
  {10.48550/arXiv.2410.10116} (\bibinfo {year} {2024})\BibitemShut {NoStop}%
\bibitem [{\citenamefont {Brakerski}\ and\ \citenamefont
  {Shmueli}(2019)}]{brakerski2019pseudo}%
  \BibitemOpen
  \bibfield  {author} {\bibinfo {author} {\bibfnamefont {Z.}~\bibnamefont
  {Brakerski}}\ and\ \bibinfo {author} {\bibfnamefont {O.}~\bibnamefont
  {Shmueli}},\ }\bibfield  {title} {\bibinfo {title} {({P}seudo) random quantum
  states with binary phase},\ }in\ \href
  {https://doi.org/10.1007/978-3-030-36030-6_10} {\emph {\bibinfo {booktitle}
  {Theory of Cryptography Conference}}}\ (\bibinfo {organization} {Springer},\
  \bibinfo {year} {2019})\ pp.\ \bibinfo {pages} {229--250}\BibitemShut
  {NoStop}%
\bibitem [{\citenamefont {Ananth}\ \emph
  {et~al.}(2022{\natexlab{b}})\citenamefont {Ananth}, \citenamefont {Gulati},
  \citenamefont {Qian},\ and\ \citenamefont {Yuen}}]{ananth2022pseudorandom}%
  \BibitemOpen
  \bibfield  {author} {\bibinfo {author} {\bibfnamefont {P.}~\bibnamefont
  {Ananth}}, \bibinfo {author} {\bibfnamefont {A.}~\bibnamefont {Gulati}},
  \bibinfo {author} {\bibfnamefont {L.}~\bibnamefont {Qian}},\ and\ \bibinfo
  {author} {\bibfnamefont {H.}~\bibnamefont {Yuen}},\ }\bibfield  {title}
  {\bibinfo {title} {Pseudorandom (function-like) quantum state generators: New
  definitions and applications},\ }in\ \href
  {https://doi.org/10.1007/978-3-031-22318-1_9} {\emph {\bibinfo {booktitle}
  {Theory of Cryptography Conference}}}\ (\bibinfo {organization} {Springer},\
  \bibinfo {year} {2022})\ pp.\ \bibinfo {pages} {237--265}\BibitemShut
  {NoStop}%
\bibitem [{\citenamefont {Zhandry}(2021)}]{zhandry2012how}%
  \BibitemOpen
  \bibfield  {author} {\bibinfo {author} {\bibfnamefont {M.}~\bibnamefont
  {Zhandry}},\ }\bibfield  {title} {\bibinfo {title} {How to construct quantum
  random functions},\ }\href {https://doi.org/10.1109/FOCS.2012.37} {\bibfield
  {journal} {\bibinfo  {journal} {Journal of the ACM (JACM)}\ }\textbf
  {\bibinfo {volume} {68}},\ \bibinfo {pages} {1} (\bibinfo {year}
  {2021})}\BibitemShut {NoStop}%
\bibitem [{\citenamefont {Goul{\~a}o}\ and\ \citenamefont
  {Elkouss}(2024)}]{goulao2024pseudo}%
  \BibitemOpen
  \bibfield  {author} {\bibinfo {author} {\bibfnamefont {M.}~\bibnamefont
  {Goul{\~a}o}}\ and\ \bibinfo {author} {\bibfnamefont {D.}~\bibnamefont
  {Elkouss}},\ }\bibfield  {title} {\bibinfo {title} {Pseudo-entanglement is
  necessary for efi pairs},\ }\href@noop {} {\bibfield  {journal} {\bibinfo
  {journal} {arXiv preprint arXiv:2406.06881}\ } (\bibinfo {year}
  {2024})}\BibitemShut {NoStop}%
\bibitem [{\citenamefont {Morimae}\ and\ \citenamefont
  {Yamakawa}(2022{\natexlab{a}})}]{morimae2022quantum}%
  \BibitemOpen
  \bibfield  {author} {\bibinfo {author} {\bibfnamefont {T.}~\bibnamefont
  {Morimae}}\ and\ \bibinfo {author} {\bibfnamefont {T.}~\bibnamefont
  {Yamakawa}},\ }\bibfield  {title} {\bibinfo {title} {Quantum commitments and
  signatures without one-way functions},\ }in\ \href
  {https://doi.org/10.1007/978-3-031-15802-5_10} {\emph {\bibinfo {booktitle}
  {Advances in Cryptology--CRYPTO 2022: 42nd Annual International Cryptology
  Conference, CRYPTO 2022, Santa Barbara, CA, USA, August 15--18, 2022,
  Proceedings, Part I}}}\ (\bibinfo {organization} {Springer},\ \bibinfo {year}
  {2022})\ pp.\ \bibinfo {pages} {269--295}\BibitemShut {NoStop}%
\bibitem [{\citenamefont {Audenaert}(2007)}]{audenaert2007sharp}%
  \BibitemOpen
  \bibfield  {author} {\bibinfo {author} {\bibfnamefont {K.~M.}\ \bibnamefont
  {Audenaert}},\ }\bibfield  {title} {\bibinfo {title} {A sharp continuity
  estimate for the von neumann entropy},\ }\href
  {https://doi.org/10.1088/1751-8113/40/28/S18} {\bibfield  {journal} {\bibinfo
   {journal} {Journal of Physics A: Mathematical and Theoretical}\ }\textbf
  {\bibinfo {volume} {40}},\ \bibinfo {pages} {8127} (\bibinfo {year}
  {2007})}\BibitemShut {NoStop}%
\bibitem [{\citenamefont {Bharti}\ \emph {et~al.}(2022)\citenamefont {Bharti},
  \citenamefont {Cervera-Lierta}, \citenamefont {Kyaw}, \citenamefont {Haug},
  \citenamefont {Alperin-Lea}, \citenamefont {Anand}, \citenamefont {Degroote},
  \citenamefont {Heimonen}, \citenamefont {Kottmann}, \citenamefont {Menke},
  \citenamefont {Mok}, \citenamefont {Sim}, \citenamefont {Kwek},\ and\
  \citenamefont {Aspuru-Guzik}}]{bharti2021noisy}%
  \BibitemOpen
  \bibfield  {author} {\bibinfo {author} {\bibfnamefont {K.}~\bibnamefont
  {Bharti}}, \bibinfo {author} {\bibfnamefont {A.}~\bibnamefont
  {Cervera-Lierta}}, \bibinfo {author} {\bibfnamefont {T.~H.}\ \bibnamefont
  {Kyaw}}, \bibinfo {author} {\bibfnamefont {T.}~\bibnamefont {Haug}}, \bibinfo
  {author} {\bibfnamefont {S.}~\bibnamefont {Alperin-Lea}}, \bibinfo {author}
  {\bibfnamefont {A.}~\bibnamefont {Anand}}, \bibinfo {author} {\bibfnamefont
  {M.}~\bibnamefont {Degroote}}, \bibinfo {author} {\bibfnamefont
  {H.}~\bibnamefont {Heimonen}}, \bibinfo {author} {\bibfnamefont {J.~S.}\
  \bibnamefont {Kottmann}}, \bibinfo {author} {\bibfnamefont {T.}~\bibnamefont
  {Menke}}, \bibinfo {author} {\bibfnamefont {W.-K.}\ \bibnamefont {Mok}},
  \bibinfo {author} {\bibfnamefont {S.}~\bibnamefont {Sim}}, \bibinfo {author}
  {\bibfnamefont {L.-C.}\ \bibnamefont {Kwek}},\ and\ \bibinfo {author}
  {\bibfnamefont {A.}~\bibnamefont {Aspuru-Guzik}},\ }\bibfield  {title}
  {\bibinfo {title} {Noisy intermediate-scale quantum algorithms},\ }\href
  {https://doi.org/10.1103/RevModPhys.94.015004} {\bibfield  {journal}
  {\bibinfo  {journal} {Rev. Mod. Phys.}\ }\textbf {\bibinfo {volume} {94}},\
  \bibinfo {pages} {015004} (\bibinfo {year} {2022})}\BibitemShut {NoStop}%
\bibitem [{\citenamefont {Chen}\ \emph {et~al.}(2023)\citenamefont {Chen},
  \citenamefont {Cotler}, \citenamefont {Huang},\ and\ \citenamefont
  {Li}}]{chen2023complexity}%
  \BibitemOpen
  \bibfield  {author} {\bibinfo {author} {\bibfnamefont {S.}~\bibnamefont
  {Chen}}, \bibinfo {author} {\bibfnamefont {J.}~\bibnamefont {Cotler}},
  \bibinfo {author} {\bibfnamefont {H.-Y.}\ \bibnamefont {Huang}},\ and\
  \bibinfo {author} {\bibfnamefont {J.}~\bibnamefont {Li}},\ }\bibfield
  {title} {\bibinfo {title} {The complexity of {NISQ}},\ }\href
  {https://doi.org/10.1038/s41467-023-41217-6} {\bibfield  {journal} {\bibinfo
  {journal} {Nature Communications}\ }\textbf {\bibinfo {volume} {14}},\
  \bibinfo {pages} {6001} (\bibinfo {year} {2023})}\BibitemShut {NoStop}%
\bibitem [{\citenamefont {Haug}\ \emph {et~al.}(2025)\citenamefont {Haug},
  \citenamefont {Bansal}, \citenamefont {Mok}, \citenamefont {Koh},\ and\
  \citenamefont {Bharti}}]{haug2024pseudorandom}%
  \BibitemOpen
  \bibfield  {author} {\bibinfo {author} {\bibfnamefont {T.}~\bibnamefont
  {Haug}}, \bibinfo {author} {\bibfnamefont {N.}~\bibnamefont {Bansal}},
  \bibinfo {author} {\bibfnamefont {W.-K.}\ \bibnamefont {Mok}}, \bibinfo
  {author} {\bibfnamefont {D.~E.}\ \bibnamefont {Koh}},\ and\ \bibinfo {author}
  {\bibfnamefont {K.}~\bibnamefont {Bharti}},\ }\bibfield  {title} {\bibinfo
  {title} {Pseudorandom quantum authentication},\ }\bibfield  {journal}
  {\bibinfo  {journal} {arXiv preprint arXiv:2501.00951}\ }\href
  {https://doi.org/10.48550/arXiv.2501.00951} {10.48550/arXiv.2501.00951}
  (\bibinfo {year} {2025})\BibitemShut {NoStop}%
\bibitem [{\citenamefont {Rubinfeld}\ and\ \citenamefont
  {Sudan}(1996)}]{rubinfeld1996robust}%
  \BibitemOpen
  \bibfield  {author} {\bibinfo {author} {\bibfnamefont {R.}~\bibnamefont
  {Rubinfeld}}\ and\ \bibinfo {author} {\bibfnamefont {M.}~\bibnamefont
  {Sudan}},\ }\bibfield  {title} {\bibinfo {title} {Robust characterizations of
  polynomials with applications to program testing},\ }\href
  {https://doi.org/10.1137/S0097539793255151} {\bibfield  {journal} {\bibinfo
  {journal} {SIAM Journal on Computing}\ }\textbf {\bibinfo {volume} {25}},\
  \bibinfo {pages} {252} (\bibinfo {year} {1996})}\BibitemShut {NoStop}%
\bibitem [{\citenamefont {Goldreich}\ \emph {et~al.}(1998)\citenamefont
  {Goldreich}, \citenamefont {Goldwasser},\ and\ \citenamefont
  {Ron}}]{goldreich1998property}%
  \BibitemOpen
  \bibfield  {author} {\bibinfo {author} {\bibfnamefont {O.}~\bibnamefont
  {Goldreich}}, \bibinfo {author} {\bibfnamefont {S.}~\bibnamefont
  {Goldwasser}},\ and\ \bibinfo {author} {\bibfnamefont {D.}~\bibnamefont
  {Ron}},\ }\bibfield  {title} {\bibinfo {title} {Property testing and its
  connection to learning and approximation},\ }\href
  {https://doi.org/10.1145/285055.285060} {\bibfield  {journal} {\bibinfo
  {journal} {Journal of the ACM (JACM)}\ }\textbf {\bibinfo {volume} {45}},\
  \bibinfo {pages} {653} (\bibinfo {year} {1998})}\BibitemShut {NoStop}%
\bibitem [{\citenamefont {Buhrman}\ \emph {et~al.}(2008)\citenamefont
  {Buhrman}, \citenamefont {Fortnow}, \citenamefont {Newman},\ and\
  \citenamefont {R{\"o}hrig}}]{buhrman2008quantum}%
  \BibitemOpen
  \bibfield  {author} {\bibinfo {author} {\bibfnamefont {H.}~\bibnamefont
  {Buhrman}}, \bibinfo {author} {\bibfnamefont {L.}~\bibnamefont {Fortnow}},
  \bibinfo {author} {\bibfnamefont {I.}~\bibnamefont {Newman}},\ and\ \bibinfo
  {author} {\bibfnamefont {H.}~\bibnamefont {R{\"o}hrig}},\ }\bibfield  {title}
  {\bibinfo {title} {Quantum property testing},\ }\href
  {https://doi.org/10.1137/S009753970444241} {\bibfield  {journal} {\bibinfo
  {journal} {SIAM Journal on Computing}\ }\textbf {\bibinfo {volume} {37}},\
  \bibinfo {pages} {1387} (\bibinfo {year} {2008})}\BibitemShut {NoStop}%
\bibitem [{\citenamefont {Ekert}\ \emph {et~al.}(2002)\citenamefont {Ekert},
  \citenamefont {Alves}, \citenamefont {Oi}, \citenamefont {Horodecki},
  \citenamefont {Horodecki},\ and\ \citenamefont {Kwek}}]{ekert2002direct}%
  \BibitemOpen
  \bibfield  {author} {\bibinfo {author} {\bibfnamefont {A.~K.}\ \bibnamefont
  {Ekert}}, \bibinfo {author} {\bibfnamefont {C.~M.}\ \bibnamefont {Alves}},
  \bibinfo {author} {\bibfnamefont {D.~K.~L.}\ \bibnamefont {Oi}}, \bibinfo
  {author} {\bibfnamefont {M.}~\bibnamefont {Horodecki}}, \bibinfo {author}
  {\bibfnamefont {P.}~\bibnamefont {Horodecki}},\ and\ \bibinfo {author}
  {\bibfnamefont {L.~C.}\ \bibnamefont {Kwek}},\ }\bibfield  {title} {\bibinfo
  {title} {Direct estimations of linear and nonlinear functionals of a quantum
  state},\ }\href {https://doi.org/10.1103/PhysRevLett.88.217901} {\bibfield
  {journal} {\bibinfo  {journal} {Phys. Rev. Lett.}\ }\textbf {\bibinfo
  {volume} {88}},\ \bibinfo {pages} {217901} (\bibinfo {year}
  {2002})}\BibitemShut {NoStop}%
\bibitem [{\citenamefont {Gross}\ \emph {et~al.}(2021)\citenamefont {Gross},
  \citenamefont {Nezami},\ and\ \citenamefont {Walter}}]{gross2021schur}%
  \BibitemOpen
  \bibfield  {author} {\bibinfo {author} {\bibfnamefont {D.}~\bibnamefont
  {Gross}}, \bibinfo {author} {\bibfnamefont {S.}~\bibnamefont {Nezami}},\ and\
  \bibinfo {author} {\bibfnamefont {M.}~\bibnamefont {Walter}},\ }\bibfield
  {title} {\bibinfo {title} {Schur--{W}eyl duality for the {C}lifford group
  with applications: Property testing, a robust {H}udson theorem, and de
  {F}inetti representations},\ }\href
  {https://doi.org/10.1007/s00220-021-04118-7} {\bibfield  {journal} {\bibinfo
  {journal} {Communications in Mathematical Physics}\ }\textbf {\bibinfo
  {volume} {385}},\ \bibinfo {pages} {1325} (\bibinfo {year}
  {2021})}\BibitemShut {NoStop}%
\bibitem [{\citenamefont {Haug}\ and\ \citenamefont
  {Kim}(2023)}]{haug2022scalable}%
  \BibitemOpen
  \bibfield  {author} {\bibinfo {author} {\bibfnamefont {T.}~\bibnamefont
  {Haug}}\ and\ \bibinfo {author} {\bibfnamefont {M.}~\bibnamefont {Kim}},\
  }\bibfield  {title} {\bibinfo {title} {Scalable measures of magic resource
  for quantum computers},\ }\href {https://doi.org/10.1103/PRXQuantum.4.010301}
  {\bibfield  {journal} {\bibinfo  {journal} {PRX Quantum}\ }\textbf {\bibinfo
  {volume} {4}},\ \bibinfo {pages} {010301} (\bibinfo {year}
  {2023})}\BibitemShut {NoStop}%
\bibitem [{\citenamefont {Haug}\ \emph {et~al.}(2024)\citenamefont {Haug},
  \citenamefont {Lee},\ and\ \citenamefont {Kim}}]{haug2023efficient}%
  \BibitemOpen
  \bibfield  {author} {\bibinfo {author} {\bibfnamefont {T.}~\bibnamefont
  {Haug}}, \bibinfo {author} {\bibfnamefont {S.}~\bibnamefont {Lee}},\ and\
  \bibinfo {author} {\bibfnamefont {M.~S.}\ \bibnamefont {Kim}},\ }\bibfield
  {title} {\bibinfo {title} {Efficient quantum algorithms for stabilizer
  entropies},\ }\href {https://doi.org/10.1103/PhysRevLett.132.240602}
  {\bibfield  {journal} {\bibinfo  {journal} {Phys. Rev. Lett.}\ }\textbf
  {\bibinfo {volume} {132}},\ \bibinfo {pages} {240602} (\bibinfo {year}
  {2024})}\BibitemShut {NoStop}%
\bibitem [{\citenamefont {Haah}\ \emph {et~al.}(2016)\citenamefont {Haah},
  \citenamefont {Harrow}, \citenamefont {Ji}, \citenamefont {Wu},\ and\
  \citenamefont {Yu}}]{haah2016sample}%
  \BibitemOpen
  \bibfield  {author} {\bibinfo {author} {\bibfnamefont {J.}~\bibnamefont
  {Haah}}, \bibinfo {author} {\bibfnamefont {A.~W.}\ \bibnamefont {Harrow}},
  \bibinfo {author} {\bibfnamefont {Z.}~\bibnamefont {Ji}}, \bibinfo {author}
  {\bibfnamefont {X.}~\bibnamefont {Wu}},\ and\ \bibinfo {author}
  {\bibfnamefont {N.}~\bibnamefont {Yu}},\ }\bibfield  {title} {\bibinfo
  {title} {Sample-optimal tomography of quantum states},\ }in\ \href
  {https://doi.org/10.1145/2897518.2897585} {\emph {\bibinfo {booktitle}
  {Proceedings of the forty-eighth annual ACM symposium on Theory of
  Computing}}}\ (\bibinfo {year} {2016})\ pp.\ \bibinfo {pages}
  {913--925}\BibitemShut {NoStop}%
\bibitem [{\citenamefont {O'Donnell}\ and\ \citenamefont
  {Wright}(2016)}]{o2016efficient}%
  \BibitemOpen
  \bibfield  {author} {\bibinfo {author} {\bibfnamefont {R.}~\bibnamefont
  {O'Donnell}}\ and\ \bibinfo {author} {\bibfnamefont {J.}~\bibnamefont
  {Wright}},\ }\bibfield  {title} {\bibinfo {title} {Efficient quantum
  tomography},\ }in\ \href
  {https://doi.org/https://doi.org/10.1145/2897518.289754} {\emph {\bibinfo
  {booktitle} {Proceedings of the forty-eighth annual ACM symposium on Theory
  of Computing}}}\ (\bibinfo {year} {2016})\ pp.\ \bibinfo {pages}
  {899--912}\BibitemShut {NoStop}%
\bibitem [{\citenamefont {Bennett}\ \emph
  {et~al.}(1996{\natexlab{b}})\citenamefont {Bennett}, \citenamefont
  {Brassard}, \citenamefont {Popescu}, \citenamefont {Schumacher},
  \citenamefont {Smolin},\ and\ \citenamefont
  {Wootters}}]{bennett1996purification}%
  \BibitemOpen
  \bibfield  {author} {\bibinfo {author} {\bibfnamefont {C.~H.}\ \bibnamefont
  {Bennett}}, \bibinfo {author} {\bibfnamefont {G.}~\bibnamefont {Brassard}},
  \bibinfo {author} {\bibfnamefont {S.}~\bibnamefont {Popescu}}, \bibinfo
  {author} {\bibfnamefont {B.}~\bibnamefont {Schumacher}}, \bibinfo {author}
  {\bibfnamefont {J.~A.}\ \bibnamefont {Smolin}},\ and\ \bibinfo {author}
  {\bibfnamefont {W.~K.}\ \bibnamefont {Wootters}},\ }\bibfield  {title}
  {\bibinfo {title} {Purification of noisy entanglement and faithful
  teleportation via noisy channels},\ }\href
  {https://doi.org/10.1103/PhysRevLett.76.722} {\bibfield  {journal} {\bibinfo
  {journal} {Phys. Rev. Lett.}\ }\textbf {\bibinfo {volume} {76}},\ \bibinfo
  {pages} {722} (\bibinfo {year} {1996}{\natexlab{b}})}\BibitemShut {NoStop}%
\bibitem [{\citenamefont {Bennett}\ \emph
  {et~al.}(1996{\natexlab{c}})\citenamefont {Bennett}, \citenamefont
  {DiVincenzo}, \citenamefont {Smolin},\ and\ \citenamefont
  {Wootters}}]{bennett1996mixed}%
  \BibitemOpen
  \bibfield  {author} {\bibinfo {author} {\bibfnamefont {C.~H.}\ \bibnamefont
  {Bennett}}, \bibinfo {author} {\bibfnamefont {D.~P.}\ \bibnamefont
  {DiVincenzo}}, \bibinfo {author} {\bibfnamefont {J.~A.}\ \bibnamefont
  {Smolin}},\ and\ \bibinfo {author} {\bibfnamefont {W.~K.}\ \bibnamefont
  {Wootters}},\ }\bibfield  {title} {\bibinfo {title} {Mixed-state entanglement
  and quantum error correction},\ }\href
  {https://doi.org/10.1103/PhysRevA.54.3824} {\bibfield  {journal} {\bibinfo
  {journal} {Phys. Rev. A}\ }\textbf {\bibinfo {volume} {54}},\ \bibinfo
  {pages} {3824} (\bibinfo {year} {1996}{\natexlab{c}})}\BibitemShut {NoStop}%
\bibitem [{\citenamefont {Winter}\ and\ \citenamefont
  {Yang}(2016)}]{winter2016operational}%
  \BibitemOpen
  \bibfield  {author} {\bibinfo {author} {\bibfnamefont {A.}~\bibnamefont
  {Winter}}\ and\ \bibinfo {author} {\bibfnamefont {D.}~\bibnamefont {Yang}},\
  }\bibfield  {title} {\bibinfo {title} {Operational resource theory of
  coherence},\ }\href {https://doi.org/10.1103/PhysRevLett.116.120404}
  {\bibfield  {journal} {\bibinfo  {journal} {Phys. Rev. Lett.}\ }\textbf
  {\bibinfo {volume} {116}},\ \bibinfo {pages} {120404} (\bibinfo {year}
  {2016})}\BibitemShut {NoStop}%
\bibitem [{\citenamefont {Fang}\ \emph {et~al.}(2018)\citenamefont {Fang},
  \citenamefont {Wang}, \citenamefont {Lami}, \citenamefont {Regula},\ and\
  \citenamefont {Adesso}}]{fang2018probabilistic}%
  \BibitemOpen
  \bibfield  {author} {\bibinfo {author} {\bibfnamefont {K.}~\bibnamefont
  {Fang}}, \bibinfo {author} {\bibfnamefont {X.}~\bibnamefont {Wang}}, \bibinfo
  {author} {\bibfnamefont {L.}~\bibnamefont {Lami}}, \bibinfo {author}
  {\bibfnamefont {B.}~\bibnamefont {Regula}},\ and\ \bibinfo {author}
  {\bibfnamefont {G.}~\bibnamefont {Adesso}},\ }\bibfield  {title} {\bibinfo
  {title} {Probabilistic distillation of quantum coherence},\ }\href
  {https://doi.org/10.1103/PhysRevLett.121.070404} {\bibfield  {journal}
  {\bibinfo  {journal} {Phys. Rev. Lett.}\ }\textbf {\bibinfo {volume} {121}},\
  \bibinfo {pages} {070404} (\bibinfo {year} {2018})}\BibitemShut {NoStop}%
\bibitem [{\citenamefont {Regula}\ \emph {et~al.}(2018)\citenamefont {Regula},
  \citenamefont {Fang}, \citenamefont {Wang},\ and\ \citenamefont
  {Adesso}}]{regula2018one}%
  \BibitemOpen
  \bibfield  {author} {\bibinfo {author} {\bibfnamefont {B.}~\bibnamefont
  {Regula}}, \bibinfo {author} {\bibfnamefont {K.}~\bibnamefont {Fang}},
  \bibinfo {author} {\bibfnamefont {X.}~\bibnamefont {Wang}},\ and\ \bibinfo
  {author} {\bibfnamefont {G.}~\bibnamefont {Adesso}},\ }\bibfield  {title}
  {\bibinfo {title} {One-shot coherence distillation},\ }\href
  {https://doi.org/10.1103/PhysRevLett.121.010401} {\bibfield  {journal}
  {\bibinfo  {journal} {Phys. Rev. Lett.}\ }\textbf {\bibinfo {volume} {121}},\
  \bibinfo {pages} {010401} (\bibinfo {year} {2018})}\BibitemShut {NoStop}%
\bibitem [{\citenamefont {Krastanov}\ \emph {et~al.}(2019)\citenamefont
  {Krastanov}, \citenamefont {Albert},\ and\ \citenamefont
  {Jiang}}]{krastanov2019optimized}%
  \BibitemOpen
  \bibfield  {author} {\bibinfo {author} {\bibfnamefont {S.}~\bibnamefont
  {Krastanov}}, \bibinfo {author} {\bibfnamefont {V.~V.}\ \bibnamefont
  {Albert}},\ and\ \bibinfo {author} {\bibfnamefont {L.}~\bibnamefont
  {Jiang}},\ }\bibfield  {title} {\bibinfo {title} {Optimized {E}ntanglement
  {P}urification},\ }\href {https://doi.org/10.22331/q-2019-02-18-123}
  {\bibfield  {journal} {\bibinfo  {journal} {{Quantum}}\ }\textbf {\bibinfo
  {volume} {3}},\ \bibinfo {pages} {123} (\bibinfo {year} {2019})}\BibitemShut
  {NoStop}%
\bibitem [{\citenamefont {Fang}\ and\ \citenamefont {Liu}(2020)}]{fang2020no}%
  \BibitemOpen
  \bibfield  {author} {\bibinfo {author} {\bibfnamefont {K.}~\bibnamefont
  {Fang}}\ and\ \bibinfo {author} {\bibfnamefont {Z.-W.}\ \bibnamefont {Liu}},\
  }\bibfield  {title} {\bibinfo {title} {No-go theorems for quantum resource
  purification},\ }\href {https://doi.org/10.1103/PhysRevLett.125.060405}
  {\bibfield  {journal} {\bibinfo  {journal} {Phys. Rev. Lett.}\ }\textbf
  {\bibinfo {volume} {125}},\ \bibinfo {pages} {060405} (\bibinfo {year}
  {2020})}\BibitemShut {NoStop}%
\bibitem [{\citenamefont {Marvian}(2020)}]{marvian2020coherence}%
  \BibitemOpen
  \bibfield  {author} {\bibinfo {author} {\bibfnamefont {I.}~\bibnamefont
  {Marvian}},\ }\bibfield  {title} {\bibinfo {title} {Coherence distillation
  machines are impossible in quantum thermodynamics},\ }\href
  {https://doi.org/10.1038/s41467-019-13846-3} {\bibfield  {journal} {\bibinfo
  {journal} {Nature communications}\ }\textbf {\bibinfo {volume} {11}},\
  \bibinfo {pages} {25} (\bibinfo {year} {2020})}\BibitemShut {NoStop}%
\bibitem [{\citenamefont {Fang}\ and\ \citenamefont {Liu}(2022)}]{fang2022no}%
  \BibitemOpen
  \bibfield  {author} {\bibinfo {author} {\bibfnamefont {K.}~\bibnamefont
  {Fang}}\ and\ \bibinfo {author} {\bibfnamefont {Z.-W.}\ \bibnamefont {Liu}},\
  }\bibfield  {title} {\bibinfo {title} {No-go theorems for quantum resource
  purification: New approach and channel theory},\ }\href
  {https://doi.org/10.1103/PRXQuantum.3.010337} {\bibfield  {journal} {\bibinfo
   {journal} {PRX Quantum}\ }\textbf {\bibinfo {volume} {3}},\ \bibinfo {pages}
  {010337} (\bibinfo {year} {2022})}\BibitemShut {NoStop}%
\bibitem [{\citenamefont {Aharonov}\ \emph {et~al.}(2022)\citenamefont
  {Aharonov}, \citenamefont {Cotler},\ and\ \citenamefont
  {Qi}}]{aharonov2021quantum}%
  \BibitemOpen
  \bibfield  {author} {\bibinfo {author} {\bibfnamefont {D.}~\bibnamefont
  {Aharonov}}, \bibinfo {author} {\bibfnamefont {J.}~\bibnamefont {Cotler}},\
  and\ \bibinfo {author} {\bibfnamefont {X.-L.}\ \bibnamefont {Qi}},\
  }\bibfield  {title} {\bibinfo {title} {Quantum algorithmic measurement},\
  }\href {https://doi.org/10.1038/s41467-021-27922-0} {\bibfield  {journal}
  {\bibinfo  {journal} {Nature communications}\ }\textbf {\bibinfo {volume}
  {13}},\ \bibinfo {pages} {887} (\bibinfo {year} {2022})}\BibitemShut
  {NoStop}%
\bibitem [{\citenamefont {Chen}\ \emph {et~al.}(2022)\citenamefont {Chen},
  \citenamefont {Cotler}, \citenamefont {Huang},\ and\ \citenamefont
  {Li}}]{chen2022exponential}%
  \BibitemOpen
  \bibfield  {author} {\bibinfo {author} {\bibfnamefont {S.}~\bibnamefont
  {Chen}}, \bibinfo {author} {\bibfnamefont {J.}~\bibnamefont {Cotler}},
  \bibinfo {author} {\bibfnamefont {H.-Y.}\ \bibnamefont {Huang}},\ and\
  \bibinfo {author} {\bibfnamefont {J.}~\bibnamefont {Li}},\ }\bibfield
  {title} {\bibinfo {title} {Exponential separations between learning with and
  without quantum memory},\ }in\ \href
  {https://doi.org/10.1109/FOCS52979.2021.00063} {\emph {\bibinfo {booktitle}
  {2021 IEEE 62nd Annual Symposium on Foundations of Computer Science
  (FOCS)}}}\ (\bibinfo {organization} {IEEE},\ \bibinfo {year} {2022})\ pp.\
  \bibinfo {pages} {574--585}\BibitemShut {NoStop}%
\bibitem [{\citenamefont {Chen}\ \emph {et~al.}(2024)\citenamefont {Chen},
  \citenamefont {Coladangelo},\ and\ \citenamefont {Sattath}}]{chen2024power}%
  \BibitemOpen
  \bibfield  {author} {\bibinfo {author} {\bibfnamefont {B.}~\bibnamefont
  {Chen}}, \bibinfo {author} {\bibfnamefont {A.}~\bibnamefont {Coladangelo}},\
  and\ \bibinfo {author} {\bibfnamefont {O.}~\bibnamefont {Sattath}},\
  }\bibfield  {title} {\bibinfo {title} {The power of a single {H}aar random
  state: constructing and separating quantum pseudorandomness},\ }\href
  {10.48550/arXiv.2404.03295} {\bibfield  {journal} {\bibinfo  {journal} {arXiv
  preprint arXiv:2404.03295}\ } (\bibinfo {year} {2024})}\BibitemShut {NoStop}%
\bibitem [{\citenamefont {Bombin}\ \emph {et~al.}(2012)\citenamefont {Bombin},
  \citenamefont {Andrist}, \citenamefont {Ohzeki}, \citenamefont {Katzgraber},\
  and\ \citenamefont {Martin-Delgado}}]{bombin2012strong}%
  \BibitemOpen
  \bibfield  {author} {\bibinfo {author} {\bibfnamefont {H.}~\bibnamefont
  {Bombin}}, \bibinfo {author} {\bibfnamefont {R.~S.}\ \bibnamefont {Andrist}},
  \bibinfo {author} {\bibfnamefont {M.}~\bibnamefont {Ohzeki}}, \bibinfo
  {author} {\bibfnamefont {H.~G.}\ \bibnamefont {Katzgraber}},\ and\ \bibinfo
  {author} {\bibfnamefont {M.~A.}\ \bibnamefont {Martin-Delgado}},\ }\bibfield
  {title} {\bibinfo {title} {Strong resilience of topological codes to
  depolarization},\ }\href {https://doi.org/10.1103/PhysRevX.2.021004}
  {\bibfield  {journal} {\bibinfo  {journal} {Phys. Rev. X}\ }\textbf {\bibinfo
  {volume} {2}},\ \bibinfo {pages} {021004} (\bibinfo {year}
  {2012})}\BibitemShut {NoStop}%
\bibitem [{\citenamefont {Morimae}\ and\ \citenamefont
  {Yamakawa}(2022{\natexlab{b}})}]{morimae2022one}%
  \BibitemOpen
  \bibfield  {author} {\bibinfo {author} {\bibfnamefont {T.}~\bibnamefont
  {Morimae}}\ and\ \bibinfo {author} {\bibfnamefont {T.}~\bibnamefont
  {Yamakawa}},\ }\bibfield  {title} {\bibinfo {title} {One-wayness in quantum
  cryptography},\ }\href {https://doi.org/10.48550/arXiv.2210.03394} {\bibfield
   {journal} {\bibinfo  {journal} {arXiv preprint arXiv:2210.03394}\ }
  (\bibinfo {year} {2022}{\natexlab{b}})}\BibitemShut {NoStop}%
\bibitem [{\citenamefont {Batra}\ and\ \citenamefont
  {Jain}(2024)}]{batra2024commitments}%
  \BibitemOpen
  \bibfield  {author} {\bibinfo {author} {\bibfnamefont {R.}~\bibnamefont
  {Batra}}\ and\ \bibinfo {author} {\bibfnamefont {R.}~\bibnamefont {Jain}},\
  }\bibfield  {title} {\bibinfo {title} {Commitments are equivalent to one-way
  state generators},\ }\href {https://doi.org/10.48550/arXiv.2404.03220}
  {\bibfield  {journal} {\bibinfo  {journal} {arXiv preprint arXiv:2404.03220}\
  } (\bibinfo {year} {2024})}\BibitemShut {NoStop}%
\bibitem [{\citenamefont {Rivest}(1991)}]{rivest1991cryptography}%
  \BibitemOpen
  \bibfield  {author} {\bibinfo {author} {\bibfnamefont {R.~L.}\ \bibnamefont
  {Rivest}},\ }\bibfield  {title} {\bibinfo {title} {Cryptography and machine
  learning},\ }in\ \href {https://doi.org/10.1007/3-540-57332-1_36} {\emph
  {\bibinfo {booktitle} {International Conference on the Theory and Application
  of Cryptology}}}\ (\bibinfo {organization} {Springer},\ \bibinfo {year}
  {1991})\ pp.\ \bibinfo {pages} {427--439}\BibitemShut {NoStop}%
\bibitem [{\citenamefont {Kearns}\ and\ \citenamefont
  {Valiant}(1994)}]{kearns1994cryptographic}%
  \BibitemOpen
  \bibfield  {author} {\bibinfo {author} {\bibfnamefont {M.}~\bibnamefont
  {Kearns}}\ and\ \bibinfo {author} {\bibfnamefont {L.}~\bibnamefont
  {Valiant}},\ }\bibfield  {title} {\bibinfo {title} {Cryptographic limitations
  on learning {B}oolean formulae and finite automata},\ }\href
  {https://doi.org/10.1145/174644.174647} {\bibfield  {journal} {\bibinfo
  {journal} {Journal of the ACM (JACM)}\ }\textbf {\bibinfo {volume} {41}},\
  \bibinfo {pages} {67} (\bibinfo {year} {1994})}\BibitemShut {NoStop}%
\bibitem [{\citenamefont {Jogenfors}(2019)}]{jogenfors2019quantum}%
  \BibitemOpen
  \bibfield  {author} {\bibinfo {author} {\bibfnamefont {J.}~\bibnamefont
  {Jogenfors}},\ }\bibfield  {title} {\bibinfo {title} {Quantum bitcoin: an
  anonymous, distributed, and secure currency secured by the no-cloning theorem
  of quantum mechanics},\ }in\ \href@noop {} {\emph {\bibinfo {booktitle} {2019
  ieee international conference on blockchain and cryptocurrency (ICBC)}}}\
  (\bibinfo {organization} {IEEE},\ \bibinfo {year} {2019})\ pp.\ \bibinfo
  {pages} {245--252}\BibitemShut {NoStop}%
\bibitem [{\citenamefont {Sano}(2022)}]{sano2022quantum}%
  \BibitemOpen
  \bibfield  {author} {\bibinfo {author} {\bibfnamefont {Y.}~\bibnamefont
  {Sano}},\ }\bibfield  {title} {\bibinfo {title} {Quantum money generated by
  multiple untrustworthy banks},\ }\href@noop {} {\bibfield  {journal}
  {\bibinfo  {journal} {arXiv preprint arXiv:2205.09303}\ } (\bibinfo {year}
  {2022})}\BibitemShut {NoStop}%
\bibitem [{\citenamefont {Shannon}(1949)}]{shannon1949communication}%
  \BibitemOpen
  \bibfield  {author} {\bibinfo {author} {\bibfnamefont {C.~E.}\ \bibnamefont
  {Shannon}},\ }\bibfield  {title} {\bibinfo {title} {Communication theory of
  secrecy systems},\ }\href
  {https://doi.org/10.1002/j.1538-7305.1949.tb00928.x} {\bibfield  {journal}
  {\bibinfo  {journal} {The Bell System Technical Journal}\ }\textbf {\bibinfo
  {volume} {28}},\ \bibinfo {pages} {656} (\bibinfo {year} {1949})}\BibitemShut
  {NoStop}%
\bibitem [{\citenamefont {Shannon}(1948)}]{shannon1948mathematical}%
  \BibitemOpen
  \bibfield  {author} {\bibinfo {author} {\bibfnamefont {C.~E.}\ \bibnamefont
  {Shannon}},\ }\bibfield  {title} {\bibinfo {title} {A mathematical theory of
  communication},\ }\href {https://doi.org/10.1002/j.1538-7305.1948.tb01338.x}
  {\bibfield  {journal} {\bibinfo  {journal} {The Bell system technical
  journal}\ }\textbf {\bibinfo {volume} {27}},\ \bibinfo {pages} {379}
  (\bibinfo {year} {1948})}\BibitemShut {NoStop}%
\bibitem [{\citenamefont {Diffie}\ and\ \citenamefont
  {Hellman}(1976)}]{diffiehellman_1976}%
  \BibitemOpen
  \bibfield  {author} {\bibinfo {author} {\bibfnamefont {W.}~\bibnamefont
  {Diffie}}\ and\ \bibinfo {author} {\bibfnamefont {M.}~\bibnamefont
  {Hellman}},\ }\bibfield  {title} {\bibinfo {title} {New directions in
  cryptography},\ }\href {https://doi.org/10.1109/TIT.1976.1055638} {\bibfield
  {journal} {\bibinfo  {journal} {IEEE Transactions on Information Theory}\
  }\textbf {\bibinfo {volume} {22}},\ \bibinfo {pages} {644} (\bibinfo {year}
  {1976})}\BibitemShut {NoStop}%
\bibitem [{\citenamefont {Goldwasser}\ and\ \citenamefont
  {Micali}(1982)}]{goldwasser1982probabilistic}%
  \BibitemOpen
  \bibfield  {author} {\bibinfo {author} {\bibfnamefont {S.}~\bibnamefont
  {Goldwasser}}\ and\ \bibinfo {author} {\bibfnamefont {S.}~\bibnamefont
  {Micali}},\ }\bibfield  {title} {\bibinfo {title} {Probabilistic encryption
  \& how to play mental poker keeping secret all partial information},\ }in\
  \href {https://doi.org/10.1145/800070.802212} {\emph {\bibinfo {booktitle}
  {Proceedings of the Fourteenth Annual ACM Symposium on Theory of
  Computing}}},\ \bibinfo {series and number} {STOC '82}\ (\bibinfo
  {publisher} {Association for Computing Machinery},\ \bibinfo {address} {New
  York, NY, USA},\ \bibinfo {year} {1982})\ p.\ \bibinfo {pages}
  {365–377}\BibitemShut {NoStop}%
\bibitem [{\citenamefont {Nechita}(2007)}]{Nechita_2007}%
  \BibitemOpen
  \bibfield  {author} {\bibinfo {author} {\bibfnamefont {I.}~\bibnamefont
  {Nechita}},\ }\bibfield  {title} {\bibinfo {title} {Asymptotics of random
  density matrices},\ }\href {https://doi.org/10.1007/s00023-007-0345-5}
  {\bibfield  {journal} {\bibinfo  {journal} {Annales Henri Poincar{\'{e}}}\
  }\textbf {\bibinfo {volume} {8}},\ \bibinfo {pages} {1521} (\bibinfo {year}
  {2007})}\BibitemShut {NoStop}%
\bibitem [{\citenamefont {Childs}\ \emph {et~al.}(2007)\citenamefont {Childs},
  \citenamefont {Harrow},\ and\ \citenamefont {Wocjan}}]{childs2007weak}%
  \BibitemOpen
  \bibfield  {author} {\bibinfo {author} {\bibfnamefont {A.~M.}\ \bibnamefont
  {Childs}}, \bibinfo {author} {\bibfnamefont {A.~W.}\ \bibnamefont {Harrow}},\
  and\ \bibinfo {author} {\bibfnamefont {P.}~\bibnamefont {Wocjan}},\
  }\bibfield  {title} {\bibinfo {title} {Weak {F}ourier-{S}chur sampling, the
  hidden subgroup problem, and the quantum collision problem},\ }in\ \href
  {https://doi.org/10.1007/978-3-540-70918-3_51} {\emph {\bibinfo {booktitle}
  {STACS 2007: 24th Annual Symposium on Theoretical Aspects of Computer
  Science, Aachen, Germany, February 22-24, 2007. Proceedings 24}}}\ (\bibinfo
  {organization} {Springer},\ \bibinfo {year} {2007})\ pp.\ \bibinfo {pages}
  {598--609}\BibitemShut {NoStop}%
\bibitem [{\citenamefont {Holevo}(1973)}]{HOLEVO1973337}%
  \BibitemOpen
  \bibfield  {author} {\bibinfo {author} {\bibfnamefont {A.}~\bibnamefont
  {Holevo}},\ }\bibfield  {title} {\bibinfo {title} {Statistical decision
  theory for quantum systems},\ }\href
  {https://doi.org/https://doi.org/10.1016/0047-259X(73)90028-6} {\bibfield
  {journal} {\bibinfo  {journal} {Journal of Multivariate Analysis}\ }\textbf
  {\bibinfo {volume} {3}},\ \bibinfo {pages} {337} (\bibinfo {year}
  {1973})}\BibitemShut {NoStop}%
\bibitem [{\citenamefont {Helstrom}(1969)}]{Helstrom1969-ny}%
  \BibitemOpen
  \bibfield  {author} {\bibinfo {author} {\bibfnamefont {C.~W.}\ \bibnamefont
  {Helstrom}},\ }\bibfield  {title} {\bibinfo {title} {Quantum detection and
  estimation theory},\ }\href {https://doi.org/10.1007/BF01007479} {\bibfield
  {journal} {\bibinfo  {journal} {Journal of Statistical Physics}\ }\textbf
  {\bibinfo {volume} {1}},\ \bibinfo {pages} {231} (\bibinfo {year}
  {1969})}\BibitemShut {NoStop}%
\bibitem [{\citenamefont {Bae}\ and\ \citenamefont
  {Kwek}(2015)}]{bae2015quantum}%
  \BibitemOpen
  \bibfield  {author} {\bibinfo {author} {\bibfnamefont {J.}~\bibnamefont
  {Bae}}\ and\ \bibinfo {author} {\bibfnamefont {L.-C.}\ \bibnamefont {Kwek}},\
  }\bibfield  {title} {\bibinfo {title} {Quantum state discrimination and its
  applications},\ }\href {https://doi.org/10.1088/1751-8113/48/8/083001}
  {\bibfield  {journal} {\bibinfo  {journal} {Journal of Physics A:
  Mathematical and Theoretical}\ }\textbf {\bibinfo {volume} {48}},\ \bibinfo
  {pages} {083001} (\bibinfo {year} {2015})}\BibitemShut {NoStop}%
\bibitem [{\citenamefont {Puchała}\ and\ \citenamefont
  {Miszczak}(2017)}]{puchala2017symbolic}%
  \BibitemOpen
  \bibfield  {author} {\bibinfo {author} {\bibfnamefont {Z.}~\bibnamefont
  {Puchała}}\ and\ \bibinfo {author} {\bibfnamefont {J.}~\bibnamefont
  {Miszczak}},\ }\bibfield  {title} {\bibinfo {title} {Symbolic integration
  with respect to the {H}aar measure on the unitary groups},\ }\href
  {https://doi.org/10.1515/bpasts-2017-0003} {\bibfield  {journal} {\bibinfo
  {journal} {Bulletin of the Polish Academy of Sciences: Technical Sciences}\
  }\textbf {\bibinfo {volume} {65}},\ \bibinfo {pages} {21} (\bibinfo {year}
  {2017})}\BibitemShut {NoStop}%
\bibitem [{\citenamefont {Zhang}\ \emph {et~al.}(2017)\citenamefont {Zhang},
  \citenamefont {Singh},\ and\ \citenamefont {Pati}}]{Zhang_2017}%
  \BibitemOpen
  \bibfield  {author} {\bibinfo {author} {\bibfnamefont {L.}~\bibnamefont
  {Zhang}}, \bibinfo {author} {\bibfnamefont {U.}~\bibnamefont {Singh}},\ and\
  \bibinfo {author} {\bibfnamefont {A.~K.}\ \bibnamefont {Pati}},\ }\bibfield
  {title} {\bibinfo {title} {Average subentropy, coherence and entanglement of
  random mixed quantum states},\ }\href
  {https://doi.org/10.1016/j.aop.2016.12.024} {\bibfield  {journal} {\bibinfo
  {journal} {Annals of Physics}\ }\textbf {\bibinfo {volume} {377}},\ \bibinfo
  {pages} {125–146} (\bibinfo {year} {2017})}\BibitemShut {NoStop}%
\bibitem [{\citenamefont {Eisert}\ and\ \citenamefont
  {Plenio}(1999)}]{eisert1999comparison}%
  \BibitemOpen
  \bibfield  {author} {\bibinfo {author} {\bibfnamefont {J.}~\bibnamefont
  {Eisert}}\ and\ \bibinfo {author} {\bibfnamefont {M.~B.}\ \bibnamefont
  {Plenio}},\ }\bibfield  {title} {\bibinfo {title} {A comparison of
  entanglement measures},\ }\href {https://doi.org/10.1080/09500349908231260}
  {\bibfield  {journal} {\bibinfo  {journal} {Journal of Modern Optics}\
  }\textbf {\bibinfo {volume} {46}},\ \bibinfo {pages} {145} (\bibinfo {year}
  {1999})}\BibitemShut {NoStop}%
\bibitem [{\citenamefont {Vidal}\ and\ \citenamefont
  {Werner}(2002)}]{vidal2002computable}%
  \BibitemOpen
  \bibfield  {author} {\bibinfo {author} {\bibfnamefont {G.}~\bibnamefont
  {Vidal}}\ and\ \bibinfo {author} {\bibfnamefont {R.~F.}\ \bibnamefont
  {Werner}},\ }\bibfield  {title} {\bibinfo {title} {Computable measure of
  entanglement},\ }\href {https://doi.org/10.1103/PhysRevA.65.032314}
  {\bibfield  {journal} {\bibinfo  {journal} {Phys. Rev. A}\ }\textbf {\bibinfo
  {volume} {65}},\ \bibinfo {pages} {032314} (\bibinfo {year}
  {2002})}\BibitemShut {NoStop}%
\bibitem [{\citenamefont {Plenio}(2005)}]{plenio2005logarithmic}%
  \BibitemOpen
  \bibfield  {author} {\bibinfo {author} {\bibfnamefont {M.~B.}\ \bibnamefont
  {Plenio}},\ }\bibfield  {title} {\bibinfo {title} {Logarithmic negativity: A
  full entanglement monotone that is not convex},\ }\href
  {https://doi.org/10.1103/PhysRevLett.95.090503} {\bibfield  {journal}
  {\bibinfo  {journal} {Phys. Rev. Lett.}\ }\textbf {\bibinfo {volume} {95}},\
  \bibinfo {pages} {090503} (\bibinfo {year} {2005})}\BibitemShut {NoStop}%
\bibitem [{\citenamefont {Devetak}\ and\ \citenamefont
  {Winter}(2005)}]{devetak2005distillation}%
  \BibitemOpen
  \bibfield  {author} {\bibinfo {author} {\bibfnamefont {I.}~\bibnamefont
  {Devetak}}\ and\ \bibinfo {author} {\bibfnamefont {A.}~\bibnamefont
  {Winter}},\ }\bibfield  {title} {\bibinfo {title} {Distillation of secret key
  and entanglement from quantum states},\ }\href@noop {} {\bibfield  {journal}
  {\bibinfo  {journal} {Proceedings of the Royal Society A: Mathematical,
  Physical and engineering sciences}\ }\textbf {\bibinfo {volume} {461}},\
  \bibinfo {pages} {207} (\bibinfo {year} {2005})}\BibitemShut {NoStop}%
\bibitem [{\citenamefont {Zanardi}\ \emph {et~al.}(2000)\citenamefont
  {Zanardi}, \citenamefont {Zalka},\ and\ \citenamefont
  {Faoro}}]{zanardi2000entangling}%
  \BibitemOpen
  \bibfield  {author} {\bibinfo {author} {\bibfnamefont {P.}~\bibnamefont
  {Zanardi}}, \bibinfo {author} {\bibfnamefont {C.}~\bibnamefont {Zalka}},\
  and\ \bibinfo {author} {\bibfnamefont {L.}~\bibnamefont {Faoro}},\ }\bibfield
   {title} {\bibinfo {title} {Entangling power of quantum evolutions},\
  }\href@noop {} {\bibfield  {journal} {\bibinfo  {journal} {Physical Review
  A}\ }\textbf {\bibinfo {volume} {62}},\ \bibinfo {pages} {030301} (\bibinfo
  {year} {2000})}\BibitemShut {NoStop}%
\bibitem [{\citenamefont {Garcia-Escartin}\ and\ \citenamefont
  {Chamorro-Posada}(2013)}]{garcia2013swap}%
  \BibitemOpen
  \bibfield  {author} {\bibinfo {author} {\bibfnamefont {J.~C.}\ \bibnamefont
  {Garcia-Escartin}}\ and\ \bibinfo {author} {\bibfnamefont {P.}~\bibnamefont
  {Chamorro-Posada}},\ }\bibfield  {title} {\bibinfo {title} {{SWAP} test and
  {H}ong-{O}u-{M}andel effect are equivalent},\ }\href
  {https://doi.org/10.1103/PhysRevA.87.052330} {\bibfield  {journal} {\bibinfo
  {journal} {Phys. Rev. A}\ }\textbf {\bibinfo {volume} {87}},\ \bibinfo
  {pages} {052330} (\bibinfo {year} {2013})}\BibitemShut {NoStop}%
\bibitem [{\citenamefont {Beckey}\ \emph {et~al.}(2023)\citenamefont {Beckey},
  \citenamefont {Pelegr\'{\i}}, \citenamefont {Foulds},\ and\ \citenamefont
  {Pearson}}]{beckey2023multipartite}%
  \BibitemOpen
  \bibfield  {author} {\bibinfo {author} {\bibfnamefont {J.~L.}\ \bibnamefont
  {Beckey}}, \bibinfo {author} {\bibfnamefont {G.}~\bibnamefont
  {Pelegr\'{\i}}}, \bibinfo {author} {\bibfnamefont {S.}~\bibnamefont
  {Foulds}},\ and\ \bibinfo {author} {\bibfnamefont {N.~J.}\ \bibnamefont
  {Pearson}},\ }\bibfield  {title} {\bibinfo {title} {Multipartite entanglement
  measures via {B}ell-basis measurements},\ }\href
  {https://doi.org/10.1103/PhysRevA.107.062425} {\bibfield  {journal} {\bibinfo
   {journal} {Phys. Rev. A}\ }\textbf {\bibinfo {volume} {107}},\ \bibinfo
  {pages} {062425} (\bibinfo {year} {2023})}\BibitemShut {NoStop}%
\bibitem [{\citenamefont {Lowe}\ and\ \citenamefont
  {Nayak}(2022)}]{lowe2022lower}%
  \BibitemOpen
  \bibfield  {author} {\bibinfo {author} {\bibfnamefont {A.}~\bibnamefont
  {Lowe}}\ and\ \bibinfo {author} {\bibfnamefont {A.}~\bibnamefont {Nayak}},\
  }\bibfield  {title} {\bibinfo {title} {Lower bounds for learning quantum
  states with single-copy measurements},\ }\href
  {https://doi.org/10.48550/arXiv.2207.14438} {\bibfield  {journal} {\bibinfo
  {journal} {arXiv preprint arXiv:2207.14438}\ } (\bibinfo {year}
  {2022})}\BibitemShut {NoStop}%
\bibitem [{\citenamefont {B{\u{a}}descu}\ and\ \citenamefont
  {O'Donnell}(2021)}]{buadescu2021improved}%
  \BibitemOpen
  \bibfield  {author} {\bibinfo {author} {\bibfnamefont {C.}~\bibnamefont
  {B{\u{a}}descu}}\ and\ \bibinfo {author} {\bibfnamefont {R.}~\bibnamefont
  {O'Donnell}},\ }\bibfield  {title} {\bibinfo {title} {Improved quantum data
  analysis},\ }in\ \href {https://doi.org/10.1145/3406325.3451109} {\emph
  {\bibinfo {booktitle} {Proceedings of the 53rd Annual ACM SIGACT Symposium on
  Theory of Computing}}}\ (\bibinfo {year} {2021})\ pp.\ \bibinfo {pages}
  {1398--1411}\BibitemShut {NoStop}%
\bibitem [{\citenamefont {Leone}\ \emph {et~al.}(2022)\citenamefont {Leone},
  \citenamefont {Oliviero},\ and\ \citenamefont {Hamma}}]{leone2021renyi}%
  \BibitemOpen
  \bibfield  {author} {\bibinfo {author} {\bibfnamefont {L.}~\bibnamefont
  {Leone}}, \bibinfo {author} {\bibfnamefont {S.~F.~E.}\ \bibnamefont
  {Oliviero}},\ and\ \bibinfo {author} {\bibfnamefont {A.}~\bibnamefont
  {Hamma}},\ }\bibfield  {title} {\bibinfo {title} {Stabilizer {R}\'enyi
  entropy},\ }\href {https://doi.org/10.1103/PhysRevLett.128.050402} {\bibfield
   {journal} {\bibinfo  {journal} {Phys. Rev. Lett.}\ }\textbf {\bibinfo
  {volume} {128}},\ \bibinfo {pages} {050402} (\bibinfo {year}
  {2022})}\BibitemShut {NoStop}%
\end{thebibliography}%

\let\addcontentsline\oldaddcontentsline

\appendix

\onecolumngrid
\newpage 

\setcounter{secnumdepth}{2}
\setcounter{equation}{0}
\setcounter{figure}{0}
\setcounter{section}{0}

\renewcommand{\thesection}{\Alph{section}}
\renewcommand{\thesubsection}{\arabic{subsection}}
\renewcommand*{\theHsection}{\thesection}

\clearpage
\begin{center}

\textbf{\large \SMLong{}}
\end{center}
\setcounter{equation}{0}
\setcounter{figure}{0}
\setcounter{table}{0}

\makeatletter

\renewcommand{\thefigure}{S\arabic{figure}}

We provide proofs and additional details supporting the claims in the main text.

\makeatletter
\@starttoc{toc}

\makeatother

\section{Extended motivation and outlook}\label{sec:motivation}

Research at the intersection of cryptography, quantum theory, information theory, and complexity theory has the potential to reveal many fascinating insights. Over time, these fields have evolved, leading to development of new fields such as quantum information (from quantum theory and information theory), quantum complexity theory (from quantum theory and complexity theory), and modern cryptography (from cryptography and complexity theory).

Historically, cryptography was considered an art. In 1948, Shannon introduced the concept of perfect secrecy~\cite{shannon1949communication}, showing that a one-time pad could achieve it. In his groundbreaking paper, ``A Mathematical Theory of Communication''~\cite{shannon1948mathematical}, Shannon proved that a perfectly secure communication system requires a pre-agreed random key bit for every bit of the message to ensure privacy. For secret communication, parties needed to share a key that adversaries must not know.  By the 1980s, researchers realized that an encryption scheme could still be secure if it leaked only a negligible amount of information to an attacker with limited computational power. This led to the concept of computational security~\cite{diffiehellman_1976, goldwasser1982probabilistic}, which considers the practical limits on an attacker's computing power and allows for a small chance of failure. This approach has become the standard for defining cryptographic security. The shift from information-theoretic security to computational security enabled public-key cryptography~\cite{diffiehellman_1976}, allowing secure communication between parties who have never met. Modern cybersecurity relies heavily on computational security, with computational complexity being central to ensuring security.

In the 1990s, the combination of quantum theory and information theory gave birth to the field of quantum information theory. Recently, the integration of cryptography, complexity theory, information theory, and quantum theory has led to new insights, such as the generation of certifiable true randomness~\cite{brakerski2021certifying}, quantum pseudorandom states and unitaries~\cite{ji2018pseudorandom}, homogeneous space pseudorandomness~\cite{arvind2023quantumtugwarrandomness}, pseudorandom isometries~\cite{ananth2023pseudorandomisometry}, pseudorandom state scramblers~\cite{lu2023quantum}, computational entanglement theory~\cite{arnon2023computational} and new cryptographic principles under minimal assumptions~\cite{kretschmer2021quantum}. These advancements have introduced pseudoresources like pseudoentanglement, pseudomagic, and pseudoimaginarity~\cite{bouland2022quantum, gu2023little, haug2023pseudorandom}. At this point, there is hope for establishing a new sub-field called ``modern quantum information theory,'' where the computational aspects of objects from quantum information theory are rigorously explored.  There is a potential of similar breakthrough as cryptography went through in 1980s. This calls for the development of new primitives, such as pseudorandom density matrices and pseudorandom channels, which extend our understanding beyond pseudorandom states, unitaries, isometries~\cite{ananth2023pseudorandomisometry} and scramblers~\cite{lu2023quantum}. Our work can be seen as a step toward realizing the vision of shaping this new field, viz.\ modern quantum information theory.

Finally, our work on pseudorandom density matrices contributes to several new insights in quantum property testing. Property testing algorithms can serve as a preliminary step towards learning, allowing for the efficient selection of a hypothesis class for further learning. Cryptography and machine learning are often seen as opposites~\cite{rivest1991cryptography,kearns1994cryptographic}: cryptography hides patterns, while machine learning aims to reveal them. This connection between cryptography and machine learning leads to improved understanding of property testing for quantum properties such as magic, coherence, and entanglement. Previous results focused solely on pure states as in Refs.~\cite{bouland2022quantum, gu2023little, haug2023pseudorandom} which turn out to be special cases of our results on general mixed states. In particular, there is a fundamental difference whether one tests a pure or state, or a general mixed states. While for pure states one can efficiently test whether a state has $\Omega(1)$ or $\omega(\log n)$ of a given resource, for general mixed states testing is inefficient for any amount of the resource.

\section{Definitions}\label{sec:definitions}
Negligible functions $\text{negl}(n)$ can be defined as follows:
\begin{definition}[Negligible function]
    A \revA{positive real-valued function $\mu:\mathbb \mathbb{N} \to \mathbb R$} is negligible if and only if $\forall c \in \mathbb{N}$, $\exists n_0 \in \mathbb{N}$ such that $\forall n > n_0$, $\mu(n) < n^{-c}$.
\end{definition}

A property tester $\mathcal{A}_Q$ regarding a property $Q$ has to fulfill two conditions~\cite{rubinfeld1996robust,goldreich1998property,buhrman2008quantum,montanaro2013survey}: Completeness and soundness. Completeness demands that the tester accepts with high probability if the state has the property within a threshold $\beta$. Soundness implies that the tester rejects with high probability if the state exceeds a threshold value $\delta$ of the property.

\begin{definition}[Property tester]\label{def:prop-tester}
An algorithm $\mathcal{A}_Q$ is a  tester for property $Q$ using $t=t(n,\delta,\beta)$ copies if, given $t$ copies of $n$-qubit quantum state $\rho$, constants $\beta>0$ and $\delta>\beta$, when the following holds:
\begin{itemize}
\item (Completeness) If $Q(\rho)\le \beta$, then
\begin{align}
\Pr[\mathcal{A}_Q\text{ accepts given }\rho^{\otimes t}]\geq\frac{2}{3}.
\end{align}
\item (Soundness) If $Q(\rho)\ge \delta$, then 
\begin{align}
\Pr[\mathcal{A}_Q\text{ accepts given $\rho^{\otimes t}$}]\leq\frac{1}{3}.
\end{align}
\end{itemize}
\end{definition}

\section{Generalized Hilbert-Schmidt ensemble}\label{sec:HS}
For pure states, the Haar measure is the unique left and right-invariant measure. In contrast, for mixed states there  is no universally agreed-upon measure on the set of density matrices. Several candidate measures have been studied~\cite{Zyczkowski_2011, Zyczkowski_2001,sarkar2019bures}, based on different induced measures. 
A common choice is to generate random density matrices by tracing out qubits from Haar random states. We call this the generalized Hilbert-Schmidt ensemble (GHSE):  
\begin{definition}[Generalized Hilbert-Schmidt ensemble (GHSE)]
    Let $\mu_{n+m}$ be the Haar measure on pure states of $n+m$ qubits, then $\eta_{n,m}$ is the induced measure on random density matrices whose states can be generated as:
    \begin{equation}
        \eta_{n,m}=\{\operatorname{tr}_m(\ket{\psi}\bra{\psi})\}_{\ket{\psi}\in\mu_{n+m}}.
    \end{equation}
\end{definition}
The GHSE is invariant under arbitrary unitary transformations, i.e. $\eta_{n,m}=\{U\operatorname{tr}_m(\ket{\psi}\bra{\psi})U^\dagger\}_{\ket{\psi}\in\mu_{n+m}}$~\cite{Nechita_2007}.
We note that the GHSE can also be generated by drawing random matrices from the Ginibre ensemble:
\begin{fact}[\cite{Nechita_2007}]
    Let $X$ be a $2^n \times 2^m$ matrix with entries independently and identically distributed according to $\mathcal{N}_{\mathbb{C}}(0, 1)$ (complex Gaussian). Then, then the matrix
    \begin{equation}
        \rho = \frac{XX^{\dagger}}{\tr XX^{\dagger}}
    \end{equation}
    is distributed according to $\eta_{n, m}$.
\end{fact}

A special case of the GHSE is the Hilbert-Schmidt ensemble based on the Hilbert-Schmidt metric
\begin{equation}
    D_\text{HS}(\rho, \sigma) = [\tr(\rho - \sigma)^2]^{1/2}.
\end{equation}
This metric defines a product measure in the space of density matrices, consisting of two parts: the eigenvalues part and the eigenvectors part. The eigenvector distribution is just the Haar distribution.
The uniform random ensemble  $\mathcal{G}_\text{HS}$ induced by the Hilbert-Schmidt measure can be constructed by tracing out $n$ qubits from a $2n$ qubit Haar random state~\cite{Zyczkowski_2001,Zyczkowski_2011}
\begin{equation}
    \mathcal{G}_\text{HS}=\eta_{2n,n}\,.
\end{equation}

\section{Indistinguishability of maximally mixed state and generalized Hilbert-Schmidt ensemble} \label{sec:stat_indist}

Here, we show that the GHSE is statistically indistinguishable from the maximally mixed state.

\begin{theorem}[$t$-copy indistinguishability of GHSE from maximally mixed state] \label{def:stat_indisting_sup}
The trace distance between $t$ copies of maximally mixed state $I_n/2^n$ and $t$ copies of a state drawn from the GHSE $\eta_{n,m}$ is $O\left(t^2/2^m\right)$. For $t=\operatorname{poly}(n)$ and $m=\omega(\log n)$, the two ensembles are indistinguishable as the trace distance is $\operatorname{negl}(n)$.
\end{theorem}
\begin{proof}
Let us consider a bipartite pure state $\ket{\psi}_{AB}$ with subsystem dimensions $d_A$ and $d_B$ respectively. We construct the $t$-th moment of the GHSE as
\begin{equation}
    \rho_{\text{GHS}}^{(t)} = \mathbb{E}_{\psi \in \text{Haar}(d_A d_B)} \left[ (\text{tr}_B(\ket{\psi}\bra{\psi}))^{\otimes t} \right] = \frac{(d_A d_B - 1)!}{(d_A d_B + t - 1)!} \sum_{\pi \in S_t} d_B^{\text{cycles}(\pi)} \hat{\pi}_A.
\end{equation}
$S_t$ is the symmetric group of degree $t$, and $\hat{\pi}_A$ is the permutation unitary operator acting on the $t$ copies of subsystem $A$ associated with the element $\pi$ of the symmetric group. The function $\text{cycles}(\pi)$ counts the number of cycles in the permutation $\pi \in S_t$. In our context of the GHSE $\eta_{n,m}$, $d_A = 2^n$ and $d_B = 2^m$. The trace distance $\text{TD}(x,y)=\frac{1}{2}\Vert x-y\Vert_1$ between $\rho_{\text{GHS}}^{(t)}$ and the maximally mixed state is
\begin{equation}
\begin{split}
    \text{TD}&\left(\rho_{\text{GHS}}^{(t)}, \frac{I_A^{\otimes t}}{(d_A)^t}\right) = \text{TD}\left( \frac{1}{(d_A d_B)^t} \left(1 - \frac{t(t-1)}{2d_A d_B} + O\left(\frac{t^4}{d_A^2 d_B^2}\right) \right) \sum_{\pi \in S_t} d_B^{\text{cycles}(\pi)} \hat{\pi}_A , \frac{I_A^{\otimes t}}{(d_A)^t} \right) \\&= \frac{1}{2}\left\Vert\frac{1}{(d_A)^t}\left(-\frac{t(t-1)}{2d_Ad_B} + O\left(\frac{t^4}{d^2_Ad_B^2} \right)\right) I_A^{\otimes t} + \frac{1}{(d_A d_B)^t}\left(1-\frac{t(t-1)}{2d_A d_B} + O\left(\frac{t^4}{d_A^2 d_B^2} \right)\right) \sum_{\substack{\pi \in S_t \\ \pi \neq I}} d_B^{\text{cycles}(\pi)} \hat{\pi}_A \right\Vert_1  \\&\leq \frac{t(t-1)}{2d_Ad_B} + O\left(\frac{t^4}{d_A^2d_B^2}\right) + \frac{1}{(d_B)^t}\left(1-\frac{t(t-1)}{2d_A d_B} + O\left(\frac{t^4}{d_A^2 d_B^2} \right)\right) \sum_{\substack{\pi \in S_t \\ \pi \neq I}} d_B^{\text{cycles}(\pi)} \\&=\frac{t(t-1)}{2d_Ad_B} + O\left(\frac{t^4}{d_A^2d_B^2}\right) + \frac{1}{(d_B)^t}\left(1-\frac{t(t-1)}{2d_A d_B} + O\left(\frac{t^4}{d_A^2 d_B^2} \right)\right) \left( \frac{(d_B+t-1)!}{(d_B-1)!} - (d_B)^t \right) \\&= \frac{t(t-1)}{2d_Ad_B} + O\left(\frac{t^4}{d_A^2d_B^2}\right) + \left(1-\frac{t(t-1)}{2d_A d_B} + O\left(\frac{t^4}{d_A^2 d_B^2} \right)\right) \left(\frac{t(t-1)}{2 d_B} + O\left(\frac{t^4}{d_B^2}\right)\right) \\&= \frac{t(t-1)}{2d_B}\left(1 + \frac{1}{2d_A}\right) + O\left(\frac{t^4}{d_B^2}\right).
\end{split}
\end{equation}
In the second line, we split the sum over $S_t$ into the identity permutation and a sum over non-identity permutations. In the third line, we used the triangle inequality and the trace norm $\norm{I_A^{\otimes t}}_1 = \norm{\hat{\pi}_A}_1 = (d_A)^t$ since $\hat{\pi}_A$ is unitary. Substituting $d_A = 2^{n}$ and $d_B = 2^m$ gives the desired bound $O(t^2 / 2^m)$ for the trace distance.
\end{proof}
This result implies that if $t = \mathrm{poly}(n)$, then $m = \omega(\log n)$ suffices for the PRDM to be computationally indistinguishable from the maximally mixed state.

One can also consider a different notion of random mixed states which yields similar results: In particular, one can consider random states $\rho$ sampled from the space of $n$-qubit states with fixed rank $\text{rank}(\rho)=m$. This ensemble has been shown to be indistinguishable from the maximally mixed state for $m=\omega(\log n)$ in Ref.~\cite{childs2007weak} via the quantum collision problem, as well as in Ref.~\cite{wright2016learn} as the rank testing problem.

\section{Magic}\label{sec:magic}
Here, we show that the robustness of magic, a magic monotone, cannot be efficiently tested, as well as results regarding black-box magic state distillation. 

The robustness of magic $R(\rho)$ for a given state $\rho$ is defined via a linear optimization program~\cite{howard2017application,liu2022many}
\begin{equation}
R(\rho)=\min \vert c_\phi\vert \text{ s.t }  \rho=\sum_{\phi \in \text{STAB}}c_\phi \phi,
\end{equation}
where $\phi$ is from the set of $n$-qubit pure stabilizer states $\text{STAB}$.
Note that $R(\phi)=1$ for convex combination of stabilizer states, and else $R(\rho)>1$. Its upper bound is $R(\rho)\leq 2^n$. 

The sub-additive version of the robustness of magic is the log-free robustness of magic~\cite{liu2022many}
\begin{equation}
    \text{LR}(\rho)=\log\Big(\min \vert c_\phi\vert \text{ s.t }  \rho=\sum_{\phi \in \text{STAB}}c_\phi \phi\Big)
\end{equation}
with $0\leq \text{LR}(\rho)\leq n$.

First, we bound the magic of the GHSE. 
\begin{theorem}[Magic of GHSE]\label{thm:magic_tune}
Each state $\rho\in \eta_{n,m}$ of the GHSE $\eta_{n,m}$ has the log-free robustness $\text{LR}(\rho)\geq n-m-2\log(n) -1$ with overwhelming probability.
\end{theorem}
\begin{proof}

We can regard the well known dual of the linear program that computes the robustness magic~\cite{liu2022many}
\begin{equation}\label{eq:Robustnessdual}
R(\rho)=\max_A \text{tr}(\rho A) \text{ s.t. } \vert \text{tr}(A \phi) \vert\leq1\, \forall \phi \in\text{STAB}
\end{equation}
where $A$ are Hermitian matrices.
Instead of a maximization to get the exact $R(\rho)$, we now want to only find a lower bound. For this, it is sufficient to find one $A$ which satisfies the constraints of~\eqref{eq:Robustnessdual}.
Thus, we have
\begin{equation}
R(\rho)\geq \text{tr}(\rho A) \text{ s.t. } \vert \text{tr}(A \phi) \vert\leq1\, \forall \phi \in\text{STAB}\,.
\end{equation}
Next, we note that the constraint condition can also be written as bound of a maximization over all pure stabilizer states
\begin{equation}\label{eq:Robustnessdualopt}
R(\rho)\geq\text{tr}(\rho A) \text{ s.t. } \max_{\phi \in\text{STAB}}\vert \text{tr}(A \phi) \vert\leq1.
\end{equation}
Next, we note that the robustness of magic does not increase when taking tensor products with Clifford states.
We now tensor with $m$-qubit maximally mixed state $I_m/2^m$ to get
\begin{equation}\label{eq:Robustnessdualopt2}
R(\rho)\geq R(\rho\otimes I_m/2^m A) \text{ s.t. } \max_{\phi \in\text{STAB}}\vert \text{tr}(A \phi) \vert\leq1 ,
\end{equation}
where now $A$ is an $n+m$ qubit operator and the maximization is taken over all $n+m$ qubit Clifford states.

Now, we regard $n$-qubit states from the GHSE $\eta_{n,m}$ which are of the form
$\text{tr}_m(\ket{\psi}\bra{\psi})$ where $\ket{\psi}$ is an $(n+m)$-qubit state.
Further, we regard the $n+m$ qubit Hermitian operator $A=c\ket{\psi}\bra{\psi}$. Here, $A$ is a rank 1 projector with coefficient $c>0$. This simplifies the constraint considerably as we can drop the absolute value on the constraint as $A$ has only positive eigenvalues. 
This allows us to connect the maximization to another magic monotone
\begin{equation}
   \max_{\phi \in\text{STAB}} \text{tr}(c\ket{\psi}\bra{\psi} \phi)=c F_\text{STAB}(\ket{\psi})
\end{equation}
where $F_\text{STAB}(\ket{\psi})=\text{max}_{\ket{\phi}\in\text{STAB}}\vert \braket{\psi}{\phi}\vert^2$ is the stabilizer fidelity which is the fidelity with the closest stabilizer state.
We have
\begin{equation}
R(\text{tr}_m(\ket{\psi}\bra{\psi}))\geq c\,\text{tr}(\text{tr}_m(\ket{\psi}\bra{\psi}) \otimes I_m/2^m \ket{\psi}\bra{\psi}) \text{ s.t. } c F_\text{STAB}(\ket{\psi})\leq1\,.
\end{equation}
Now, we note that the constraint can be rewritten into $c\leq (F_\text{STAB}(\ket{\psi}))^{-1}$.
Thus, we can simplify our bound as follows:
\begin{equation}
R(\text{tr}_m(\ket{\psi}\bra{\psi}))\geq 
2^{-m} F_\text{STAB}(\ket{\psi})^{-1}\text{tr}(\text{tr}_m(\ket{\psi}\bra{\psi}) \otimes I_m \ket{\psi}\bra{\psi}).
\end{equation}
Now, we regard the term $\text{tr}(\text{tr}_m(\ket{\psi}\bra{\psi}) \otimes I_m \ket{\psi}\bra{\psi})$ closer. 
We can rewrite this as
\begin{gather}
   \text{tr}(\text{tr}_m(\ket{\psi}\bra{\psi}) \otimes I_m \ket{\psi}\bra{\psi})=\text{tr}(\text{tr}_m(\ket{\psi}\bra{\psi})\otimes  (\sum_{i} \ket{i}\bra{i}) \ket{\psi}\bra{\psi})=\\
   \text{tr}(\text{tr}_m(\ket{\psi}\bra{\psi})\otimes  \sum_{i}( \bra{i} \ket{\psi}\bra{\psi} \ket{i}))=\text{tr}(\text{tr}_m(\ket{\psi}\bra{\psi})^2)\geq 2^{-m}\,.
\end{gather}
where the last inequality follows from the fact that $\text{tr}_m(\ket{\psi}\bra{\psi})$ is a at most rank $m$ operator with trace $1$.

Thus, the  robustness of magic for $\rho$ is bounded as
\begin{equation}
R(\text{tr}_m(\ket{\psi}\bra{\psi}))\geq 2^{-2m}F_\text{STAB}(\ket{\psi})^{-1} \,.
\end{equation}
and similarly the log-free robustness of magic
\begin{equation}\label{eq:logfreebound}
\text{LR}(\text{tr}_m(\ket{\psi}\bra{\psi}))\geq -\log(F_\text{STAB}(\ket{\psi}))-2m \,,
\end{equation}
where we note that $D_\text{min}=-\log(F_\text{STAB})$ is the min-relative entropy of magic.
Thus, we transformed the problem of finding a lower bound on the log-free robustness for a mixed state into finding the stabilizer fidelity $F_\text{STAB}$ for a pure state, which is a much simpler problem.

Now, we regard the case where $\ket{\psi}$ is a $n+m$ qubit Haar random state.
Here, it is known that $F_\text{STAB}(\ket{\psi})$ concentrates when $\ket{\psi}$ is drawn from Haar measure $\mu$ over $n+m$-qubits~\cite{liu2022many}
\begin{equation}\label{eq:FSTABconcentrate}
    \Pr_{\ket{\psi}\leftarrow\mu}[ F_\text{STAB}(\ket{\psi}) >\epsilon]< \exp(0.54 (n+m)^2 -2^{n+m}\epsilon),
\end{equation}
which is valid for $n+m\geq 6$.
For concentration, we need the right hand side to scale as $\exp(-(n+m)^2)$ or faster to zero, thus we have
$\epsilon\geq 1.54 (n+m)^2 2^{-n-m}$, where we make the simplified choice of $\epsilon= 2(n+m)^2 2^{-n-m}$.
\eqref{eq:FSTABconcentrate} implies that we have 
\begin{equation}\label{eq:FSTABconcentrate2}
    \Pr_{\ket{\psi}\leftarrow\mu}[ F_\text{STAB}(\ket{\psi}) >2(n+m)^2 2^{-n-m}]< \exp(-(n+m)^2).
\end{equation}
In particular, with overwhelming probability for Haar random states $\ket{\psi}$ we have 
\begin{equation}
    F_\text{STAB}(\ket{\psi})\leq 2(n+m)^2 2^{-n-m}\,
\end{equation}
and
\begin{equation}
    D_\text{min}(\ket{\psi})\geq n+m -2\log(n+m)-1\,.
\end{equation}
We now consider the GHSE $\eta_{n,m}=\{\text{tr}_m(\ket{\psi}\bra{\psi})\}_{\ket{\psi}\in \mu}$ where $\ket{\psi}$ is drawn from the Haar measure $\mu$. 
The average log-free robustness of GHSE is given by
\begin{gather}
\Eset{\rho \in \eta_{n,m}}[\text{LR}(\rho)]=\int_{\psi\in\mu}\text{d}\psi \text{LR}(\text{tr}_m(\ket{\psi}\bra{\psi}))\geq -2m-\int_{\ket{\psi}\in\mu}\text{d}\psi\,\log(F_\text{STAB}(\ket{\psi})).
\end{gather}
Using Jensen's inequality $-\int_\psi\text{d}\psi\, \log(x)\geq  -\log(\int_\psi\text{d}\psi\, x)$ and~\eqref{eq:FSTABconcentrate2}, we now get
\begin{gather}
\Eset{\rho \in \eta_{n,m}}[\text{LR}(\rho)]\geq -2m-\log\left(\int_{\ket{\psi}\in\mu}\text{d}\psi\, F_\text{STAB}(\ket{\psi})\right)\geq n-m-2\log(n+m) -1\,.
\end{gather}
In particular, by choosing $m=\text{polylog}(n)$ , we have
$\Eset{\rho \in S}[\text{LR}(\rho)]=\Theta(n)$.
\end{proof}

\subsection{Magic testing}

\begin{theorem}[Magic cannot be efficiently tested for mixed states]\label{thm:no_test_magic_sup}
Any tester $\mathcal{A}_{\mathcal{M}}$ according to Def.~\ref{def:prop-tester} of log-free robustness of magic $\text{LR}$  requires superpolynomial number of copies for $\delta = \Theta(n)$, $\beta=0$ to test $n$-qubit states. 
\end{theorem}
\begin{proof}
The  proof idea is that there the GHSE with $m=\text{polylog}(n)$ has near-maximal magic, while being indistinguishable from maximally mixed states with zero magic. 

We now consider the GHSE ensemble. 
The average log-free robustness of ensemble $\eta_{n,m}$ is lower bounded by
\begin{gather}
\Eset{\rho \in \eta_{n,m}}[\text{LR}(\rho)]\geq n-m-2\log(n+m) -1\,,
\end{gather}
where due to concentration this bound is satisfied for nearly all states with exponentially high probability.

Now, we choose $m=\text{polylog}(n)$ and $\eta\equiv \eta_{n,\text{polylog}(n)}$, where we have
$\Eset{\rho \in \eta}[\text{LR}(\rho)]=\Theta(n)$ with overwhelming probability.  

As second ensemble, we choose the maximally mixed state $I_n/2^n$ with $\text{LR}(I_n/2^n)=0$.

Now, $\eta$ is indistinguishable from the maximally mixed state
\begin{equation}\label{eq:TDdistancemagic}
    \text{TD}\left(\underset{\rho\leftarrow \eta}{\mathbb{E}}\big[\rho^{\otimes t}\big],(I_n/2^n)^{\otimes t}\right) = O\left(\frac{t^2}{2^m}\right).
\end{equation}
The two ensembles can be distinguished with probability $P_\text{discr}$ only if they are sufficiently far in TD distance as given by the Helstrom bound~\cite{HOLEVO1973337, Helstrom1969-ny,bae2015quantum}
\begin{equation}\label{eq:helstrom2}
    P_\text{discr}(\rho,\sigma)=\frac{1}{2}+\text{TD}(\rho,\sigma)\,.
\end{equation}
In particular, any algorithm trying to distinguish two ensembles with $P_\text{discr}\ge2/3$ requires at least $\text{TD}(\rho ,\sigma)\ge1/6$.

Now, we want to distinguish $\eta$ and $I_n/2^n$, which have $\text{LR}=\Theta(n)$ and $\text{LR}=0$ log-free robustness of magic respectively with overwhelming probability, which implies $\delta=\Theta(n)$ and $\beta=0$. 
~\eqref{eq:TDdistancemagic} implies that any algorithm, which includes property testers for log-free robustness of magic, must use $t=\Omega(2^{\text{polylog}(n)/2})$ copies to distinguish the two ensembles with non-negligible probability. 

\end{proof}

\begin{corollary}[Lower bound on purity for efficient testing of magic]\label{thm:test_purity_magic_sup}
Using $t=\mathrm{poly}(n)$ copies of $n$-qubit state $\rho$, testing whether $\rho$ has $\text{LR}(\rho)=0$ or $\text{LR}(\rho)=\Theta(n)$  requires purity $\text{tr}(\rho^2)=\Omega(n^{-c})$ with $c>0$. 
\end{corollary}
The lower bound follows from the inefficiency of testing of magic for negligible purity states as shown in Thm.~\ref{thm:no_test_magic_sup}. 
We note that for pure states with $\text{tr}(\rho^2)=1$ explicit tests of magic with $t=O(1)$ are known~\cite{haug2022scalable,haug2023efficient}.

\subsection{Black-box magic state distillation}

\begin{theorem}[Inefficiency of black-box magic state distillation]\label{thm:magic_distill}
Any algorithm that uses stabilizer operations and $t=\mathrm{poly}(n)$ copies of
an arbitrary input state $\rho$ with $\text{LR}(\rho)=\Theta(n)$ can only synthesize trivial pure states $\ket{\psi}$ with $\text{LR}(\ket{\psi})=\operatorname{negl}(n)$.
\end{theorem}
\begin{proof}
\revA{The proof idea follows the proof for inefficiency of property testing of magic, where for completeness we again provide in detail:} 
Lets assume there exists a magic state distillation algorithm $\mathcal{A}_\text{M}$ which uses stabilizer operations to turn $t=\mathrm{poly}(n)$ copies of $\rho$ with $\text{LR}(\rho)=\Theta(n)$ into a single magical state $\ket{\psi}$ with $\text{LR}(\ket{\psi})=\Omega(n^{-c})$ with $c>0$. 
The same algorithm $\mathcal{A}_\text{M}$ applied to a state $\rho'$ with $\text{LR}(\rho')=0$ can only yield some stabilizer state $\sigma\in\text{STAB}$ as stabilizer operations cannot increase magic. 

We now draw $\rho$ from the GHSE with $m=\omega(\log n)$ which has $\text{LR}(\rho)=\Theta(n)$, while $\rho'=I/2^{n}$ is the maximally mixed state with $\text{LR}(\rho')=0$. 
GHSE and maximally mixed state are indistinguishable for any algorithm with $t=\mathrm{poly}(n)$ copies. 
However, the proposed magic state distillation algorithm $\mathcal{A}_\text{M}$ could be used to distinguish both ensembles, as the distilled magic state $\ket{\psi}$ can be efficiently distinguished from stabilizer states~\cite{haug2022scalable,haug2023efficient}. 
Thus, from contradiction the proposed magic state distillation algorithm $\mathcal{A}_\text{M}$ cannot exist.
\end{proof}

We note this result significantly improves the previously known bound by Ref.~\cite{gu2023little} which stated that the copies are bounded by $t=\Omega(\text{LR}(\ket{\psi})/\log^{1+c}(\text{LR}(\rho)))$.

\begin{corollary}[Lower bound on purity for black-box magic state distillation]\label{thm:magic_distill_purity}
For black-box magic state distillation of state $\ket{\psi}$ with $\text{LR}(\ket{\psi})=\Omega(n^{-c})$ with $c>0$ using $t=\mathrm{poly}(n)$ copies of input state $\rho$, one requires purity $\text{tr}(\rho^2)=\Omega(n^{-c})$.
\end{corollary}
This lower bound on purity follows directly from the proof of Theorem~\ref{thm:magic_distill} which was demonstrated using state ensembles with purity $\text{tr}(\rho^2)=\operatorname{negl}(n)$. 

\section{Coherence}\label{sec:coherence}

Here, we show that coherence cannot be efficiently tested, as well as results regarding black-box distillation of coherence.
We regard the relative entropy of coherence which is a coherence monotone~\cite{baumgratz2014quantifying, streltsov2017colloquium}
\begin{equation}
    C(\rho) = S(\Delta[\rho]) - S(\rho)
\end{equation}
where $S(\rho)=-\text{tr}(\rho \log \rho)$ denotes the von Neumann entropy of the state $\rho$.  $\Delta[\rho]=\sum_i \ket{i}\bra{i}\rho\ket{i}\bra{i}$ corresponds to the state $\rho$ with only diagonal elements left, which is equivalent to applying the completely dephasing channel to $\rho$.

Now, we calculate the lower bound on the relative entropy of coherence for the GHSE:

\begin{lemma}[Expectation and variance of relative entropy of coherence of GHSE]\label{lem:coherencePRDM}
    The expected relative entropy of coherence of GHSE $\eta_{n,m}=\{\operatorname{tr}_m(\ket{\psi}\bra{\psi})\}_{\ket{\psi}\in\mu_{n+m}}$ is lower bounded as
    \begin{equation}
        \Eset{\rho \leftarrow \eta_{n,m}} [C(\rho)]\geq n -m-1\,.
    \end{equation}
    
\end{lemma}
\begin{proof}
    The average relative entropy of coherence is given by $\Eset{\rho \leftarrow \eta_{n,m}} [C(\rho)] = \Eset{\rho  \leftarrow \eta_{n,m}}[S(\Delta[\rho])] - \Eset{\rho \leftarrow \eta_{n,m}}[S(\rho)]$.

    To find a lower bound on $\Eset{\rho \leftarrow \eta_{n,m}} [C(\rho)]$, we require a lower bound on $\Eset{\rho  \leftarrow \eta_{n,m}}[S(\Delta[\rho])]$ and an upper bound on $\Eset{\rho \leftarrow \eta_{n,m}}[S(\rho)]$.

    First, as upper bound we find 
    \begin{equation}\label{eq:coherenceUpper}
        \Eset{\rho \leftarrow \eta_{n,m}}[S(\rho)]\leq m.
    \end{equation}
    This is easy to see by regarding $\rho=\operatorname{tr}_m(\ket{\psi}\bra{\psi})$ with  arbitrary $n+m$ qubit state $\ket{\psi}$.
    We recall that the von Neumann entropy is upper bounded by the rank 
    \begin{equation}\label{eq:neumannrank}
        \log(\text{rank}(\rho))\geq S(\rho)\,.
    \end{equation}
    After the partial trace over $m$ qubits,  we have $m\geq\log(\text{rank}(\rho))$, which immediately gives us~\eqref{eq:coherenceUpper}.
    
    Now, we find a lower bound on $\Eset{\rho  \leftarrow \eta_{n,m}}[S(\Delta[\rho])]$. 
    The lower bound on average von Neumann entropy of the diagonal $\Delta[\rho]$ for $\rho \in\eta_{n,m}$ can be bounded using monotonicity of $\alpha$-R\'enyi entropy: 
    \begin{equation}
    S(\Delta[\rho]) = -\tr(\Delta[\rho] \log (\Delta[\rho])) \geq \frac{1}{1-\alpha} \log( \tr (\Delta[\rho]^{\alpha} )), \hspace{3ex} \alpha \geq 2
    \end{equation}
    For $\alpha = 2$ and using Jensen's inequality for the function $g(x) = -\log x$, we get,
    \begin{equation}
        \Eset{\rho \leftarrow \eta_{n,m}} [S(\Delta[\rho])] \geq - \log( \Eset{\rho \leftarrow \eta_{n,m}} \tr (\Delta[\rho]^2)).
    \end{equation}
    Now, we can expand $U = \sum_{i, j}^{2^{n+m} - 1} U_{i,j} |i\rangle \langle j|$ to get,
    \begin{equation}
        \begin{split}
            \rho &= \tr_m[U|0\rangle \langle 0| U^{\dagger}] , \hspace{3ex} U \leftarrow \mu_{n+m} \\
            &= \tr_m \Big[\sum_{i_1, k_1}^{2^n - 1}\sum_{i_2, k_2}^{2^m - 1} U_{i_1 i_2, 0} U^{*}_{k_1 k_2, 0} |i_1 i_2\rangle \langle k_1 k_2|\Big].
        \end{split}
    \end{equation}
    Then, it is easy to observe that the diagonal part can be obtained as:
    \begin{equation}
        \Delta[\rho] = \sum_{i}^{2^n - 1}\sum_{j}^{2^m - 1}U_{ij, 0}U^{*}_{ij, 0} |i \rangle \langle i|.
    \end{equation}
    Now, we can calculate the expectation of purity of the above state,
    \begin{equation}
        \Eset{U \leftarrow \eta_{n+m}}[ \tr \Delta[\rho]^2] = \sum_{i}^{2^n - 1}\sum_{j, l}^{2^m - 1} \int_U \d\mu_{n+m} U_{ij, 0}U^{*}_{ij, 0}U_{il, 0}U^{*}_{il, 0}.
    \end{equation}
    Now, using the identity~\cite{puchala2017symbolic},
    \begin{equation}
        \begin{split}
            \int \d\mu_{n+m} U_{i_1 j_1} U_{i_2 j_2} U_{i_1^{\prime} j_1^{\prime}}^* U_{i_2^{\prime} j_2^{\prime}}^*= & \frac{1}{d^2-1}\left(\delta_{i_1 i_1^{\prime}} \delta_{i_2 i_2^{\prime}} \delta_{j_1 j_1^{\prime}} \delta_{j_2 j_2^{\prime}}+\delta_{i_1 i_2^{\prime}} \delta_{i_2 i_1^{\prime}} \delta_{j_1 j_2^{\prime}} \delta_{j_2 j_1^{\prime}}\right)\\
            & -\frac{1}{d\left(d^2-1\right)}\left(\delta_{i_1 i_1^{\prime}} \delta_{i_2 i_2^{\prime}} \delta_{j_1 j_2^{\prime}} \delta_{j_2 j_1^{\prime}}+\delta_{i_1 i_2^{\prime}} \delta_{i_2 i_1^{\prime}} \delta_{j_1 j_1^{\prime}} \delta_{j_2 j_2^{\prime}}\right),
        \end{split}
    \end{equation}
    with $d = 2^{n+m}$, we get:
    \begin{equation}
        \Eset{U \leftarrow \eta_{n+m}}[ \tr \Delta[\rho]^2] = \sum_{i}^{2^n - 1}\sum_{j, l}^{2^m - 1} \Big[\frac{1}{d^2 - 1}(1 + \delta_{jl}) - \frac{1}{d(d^2 - 1)}(1 + \delta_{jl})\Big] = \frac{2^m+1}{2^{n+m}+1}.
    \end{equation}
    Thus, from Jensen's inequality we get $\Eset{U \leftarrow \eta_{n+m}} [S(\Delta[\rho])] \geq -\log \frac{2^m+1}{2^{n+m}+1}$. Using this and \eqref{eq:coherenceUpper}, we get:
    \begin{equation}\label{lower bound of expected relative entropy of coherence}
        \Eset{U \leftarrow \eta_{n+m}} [C(\rho)] \geq \log (2^{n+m} + 1) - \log (2^m +1) - m \geq n-m -1.
    \end{equation}
    for $m = \text{poly}\log (n)$.
    
\end{proof}

Further, the relative entropy of coherence concentrates around its mean for GHSE~\cite{Zhang_2017}. This follows directly from L\'evy's lemma and the fact that the relative entropy of coherence is Lipschitz continuous. In particular, we have the Lipschitz function $f: \mathbb{S}^a \rightarrow \mathbb{R}$, defined as $f(\psi_{AB}) = S(\Delta[\rho_A]) - S(\rho_A)$ with $\rho_A = \tr_B [|\psi_{AB}\rangle\langle\psi_{AB}|]$ and $a = 2d_Ad_B - 1$ where in our case $d_A=2^n$, $d_B=2^m$. The Lipschitz constant for relative entropy of coherence for GHSE with $d_A \geq 3$ is shown to be $2\sqrt{8}\log d_A$ ~\cite{Zhang_2017}. Thus, we get, from L\'evy's lemma,
\begin{equation}\label{eq:levy lemma for GHSE coherence}
    \Pr_{\rho\leftarrow \eta_{n,m}} \left[\left | C(\rho) - \Eset{\rho' \leftarrow \eta_{n,m}} [C(\rho')]\right | > \epsilon \right] \leq 2 \exp\left(-\frac{2^{n+m} \epsilon^2}{144\pi^3n^2\log 2 }\right)\,.
\end{equation}
Now, using Lemma~\ref{lem:coherencePRDM}, we get 
\begin{equation}\label{eq:concentrationCoherence}
    \Pr_{\rho\leftarrow \eta_{n,m}} \left[ C(\rho)< n-m-1- \epsilon \right] \leq 2 \exp\left(-\frac{2^{n+m} \epsilon^2}{144\pi^3n^2\log 2 }\right)\,.
\end{equation}

\subsection{Coherence testing}

\begin{theorem}[Coherence testing is inefficient for mixed states]\label{thm:coherence testing} Any tester $\mathcal{A}_{\mathcal{C}}$ of relative entropy of coherence $C$ according to Def. \ref{def:prop-tester} requires superpolynomially many copies to test coherence with $\delta =\Theta(n)$ and $\beta =0$ for $n$-qubit states. 
\end{theorem}
\begin{proof}
    We now set $m = \text{polylog}(n)$ and use~\eqref{eq:concentrationCoherence} with $\epsilon=1$. Then, it follows that the coherence of the GHSE ensemble $\eta_{n,\text{polylog}(n)}$ is $\Theta(n)$ with overwhelming probability.
    The trace distance between GHSE with $m = \text{polylog}(n)$, $t=\mathrm{poly}(n)$ and maximally mixed state with zero coherence is negligible as shown by Theorem~\ref{def:stat_indisting_sup}.
    The two ensembles can be distinguished with probability $P_\text{discr}$ only if they are sufficiently far in TD distance as given by the Helstrom-Holevo bound~\cite{HOLEVO1973337, Helstrom1969-ny,bae2015quantum}
\begin{equation}\label{eq:helstrom}
    P_\text{discr}(\rho,\sigma)=\frac{1}{2}+\text{TD}(\rho,\sigma)\,.
\end{equation}
In particular, any algorithm trying to distinguish two ensembles with $P_\text{discr}\ge2/3$ requires at least $\text{TD}(\rho ,\sigma)\ge1/6$. Now, distinguishing GHSE for $m=\text{polylog}(n)$ and maximally mixed state with probability at least 2/3, requires $t = \Omega(2^{\text{polylog}(n)})$ copies for any possible algorithm. This implies that any coherence tester must also use at least a superpolynomial number of copies.
\end{proof}

\begin{corollary}[Lower bound on purity for efficient testing of coherence]\label{thm:test_purity_coherence_sup}
Using $t=\mathrm{poly}(n)$ copies of $n$-qubit state $\rho$, testing whether $\rho$ has $C(\rho)=0$ or $C(\rho)=\Theta(n)$  requires purity $\text{tr}(\rho^2)=\Omega(n^{-c})$ with $c>0$. 
\end{corollary}
The lower bound follows from the inefficiency of testing of coherence for negligible purity states as shown in Thm.~\ref{thm:coherence testing}. 
We note that for pure states with $\text{tr}(\rho^2)=1$ explicit tests of coherence with $t=O(1)$ are known~\cite{haug2023pseudorandom}.

\subsection{Black-box coherence distillation}

\begin{theorem}[Inefficiency of black-box coherence state distillation]\label{thm:coherence_distill}
Any algorithm that uses incoherent operations and $t=\mathrm{poly}(n)$ copies of
an arbitrary input state $\rho$ with $C(\rho)=\Theta(n)$ can only synthesize trivial pure states $\ket{\psi}$  with $C(\ket{\psi})=\operatorname{negl}(n)$.
\end{theorem}
\begin{proof}
\revA{Our proof follows the proof for inefficiency of property testing}: 
Let us assume there exists a coherence distillation algorithm $\mathcal{A}_\text{C}$ which uses incoherent operations to turn $t=\mathrm{poly}(n)$ copies of $\rho$ with $C(\rho)=\Theta(n)$ into a single coherent state $\ket{\psi}$ with $C(\ket{\psi})=\Omega(n^{-c})$ with $c>0$. 
The same algorithm $\mathcal{A}_\text{C}$ applied to a state $\rho'$ with $C(\rho')=0$ can only yield some incoherent state $\sigma$ with $C(\sigma)=0$ as incoherent operations cannot increase the relative entropy of coherence as it is a monotone. 

We now either draw $\rho$ from a GHSE with $m=\omega(\log n)$ which has $C(\rho)=\Theta(n)$, or the maximally mixed state $\rho'=I_n/2^n$ with $C(\rho')=0$. 
GHSE and maximally mixed state are statistically indistinguishable for $t=\mathrm{poly}(n)$ copies. 
However, the proposed coherence state distillation algorithm $\mathcal{A}_\text{C}$ could be used to distinguish both ensembles, as the distilled coherent state $\ket{\psi}$ from $\rho$ can be efficiently distinguished from the incoherent output $\sigma$ generated from $\rho'$~\cite{haug2023pseudorandom}. In particular, if $\sigma$ is pure, one can test that it is not coherent efficiently~\cite{haug2023pseudorandom}. 
If $\sigma$ is a mixed incoherent state, one can use the SWAP test to distinguish it from $\ket{\psi}$.
Thus, from contradiction the proposed coherence distillation algorithm $\mathcal{A}_\text{C}$ cannot exist.
\end{proof}

\begin{corollary}[Lower bound on purity for black-box coherence distillation]\label{thm:coherence_distill_purity}
Black-box coherence distillation of state $\ket{\psi}$ with $C(\ket{\psi})=\Omega(n^{-c})$ with $c>0$ and $t=\mathrm{poly}(n)$ copies of input state $\rho$ requires purity $\text{tr}(\rho^2)=\Omega(n^{-c})$.
\end{corollary}
This lower bound on purity follows directly from the proof of Theorem~\ref{thm:coherence_distill} which was proven using state ensembles with purity $\text{tr}(\rho^2)=\operatorname{negl}(n)$.

\section{Entanglement}\label{sec:entanglement}
Here, we show that entanglement as measured by the distillable entanglement $E_\text{D}$ cannot be efficiently tested, as well as results regarding black-box entanglement distillation. 
While we concentrate on $E_\text{D}$, note that the same results also apply for the logarithmic negativity and entanglement of formation~\cite{eisert1999comparison,vidal2002computable,plenio2005logarithmic}.

\revA{
First, we compute the distillable entanglement of the GHSE
\begin{theorem}[Distillable entanglement of GHSE]\label{thm:entanglement_tune}
Each state $\rho\in \eta_{n,m}$ of the GHSE $\eta_{n,m}$ has the distillable entanglement between bipartition $n_A=\vert A\vert$ and $n_B=\vert B\vert$ with $n_A \leq n_B$, $n\geq m$
        \begin{align*}
         \mathbb{E}_{\rho\in\eta_{n,m}}[E_D(\rho)] \geq n_A - m - 1
        \end{align*}
    For $m=\mathrm{polylog}(n)$, we have $\mathbb{E}_{\rho\in\eta_{n,m}}[E_D(\rho)]=\Theta(n)$. 
\end{theorem}
\begin{proof}
    To show this result, we use the hashing inequality~\cite{devetak2005distillation,plenio2005introduction} which states
    \begin{equation}
        E_\text{D}(\rho)\geq S(\rho_A)-S(\rho)\,.
    \end{equation}
    A further relaxation of the bound gives us
        \begin{equation}\label{eq:entSimple}
        E_\text{D}(\rho)\geq -\log(\text{tr}(\rho_A^2))-\log(\text{rank}(\rho))\,.
    \end{equation}
    Now, we one can easily see that $\log(\text{rank}(\rho))=m$. Further, since we can treat the subsystem $\rho_A$ as obtained from $n+m$ Haar random state by partial trace on $n_B + m$ qubits, one finds that 
    \begin{equation}
        \mathbb{E}_{\rho\in\eta_{n,m}}[\text{tr}(\rho_A^2)]=\frac{2^{n_A}+2^{n+m-n_A}}{2^{n+m}+1}\geq \frac{1}{2}(2^{n_A-n-m}+2^{-n_A})
    \end{equation}
    via Haar integration~\cite{zanardi2000entangling}. As the Lipschitz constant of the purity is constant, this bound is fulfilled with overwhelming probability according to Levy's lemma.
    Using convexity, we have
\begin{equation}
    \mathbb{E}_{\rho\in\eta_{n,m}}[-\log(\text{tr}(\rho_A^2))]\leq -\log(\text{tr}(\mathbb{E}_{\rho\in\eta_{n,m}}[\rho_A^2]))\leq n_A-1\,.
\end{equation}
    Thus, we get on average
    \begin{equation}
    \mathbb{E}_{\rho\in\eta_{n,m}}[E_D(\rho)] \geq n_A - m - 1\,.
    \end{equation}
    When $m = \text{polylog}(n)$, we have that $E_D(\rho) = \Theta(n_A)$.
    
\end{proof}
}

\subsection{Entanglement testing}

\begin{theorem}[Entanglement cannot be efficiently tested for mixed states]\label{thm:no_test_entanglement_sup}
Any tester $\mathcal{A}_{\mathcal{E}}$ according to Def.~\ref{def:prop-tester} of distillable entanglement $E_\text{D}(\rho)$ requires a superpolynomial number of copies for $\delta = \Theta(n)$, $\beta=0$ to test $n$-qubit states.
\end{theorem}
\begin{proof}
This follows directly by regarding two ensembles which are statistically indistinguishable, yet have widely different entanglement.
In particular, the maximally mixed state has $E_\text{D}(I_n/2^n)=0$, while the GHSE $\rho\in\eta_{n,m}$ for $m=\text{polylog}(n)$ has $E_\text{D}(\rho)=\Theta(n)$. 

Any algorithm, including property testers for entanglement, require $\Omega(2^{\mathrm{polylog}(n)/2})$ copies to distinguish those two ensembles due to Thm.~\ref{def:stat_indisting_sup}.
\end{proof}

\begin{corollary}[Lower bound on purity for efficient testing of entanglement]\label{thm:test_purity_entanglement_sup}
Using $t=\mathrm{poly}(n)$ copies of $n$-qubit state $\rho$, testing whether $\rho$ has $E_\text{D}(\rho)=0$ or $E_\text{D}(\rho)=\Theta(n)$  requires purity $\text{tr}(\rho^2)=\Omega(n^{-c})$ with $c>0$. 
\end{corollary}
The lower bound follows from the inefficiency of testing of entanglement for negligible purity states as shown in Thm.~\ref{thm:no_test_entanglement_sup}. 
We note that for pure states with $\text{tr}(\rho^2)=1$ explicit tests of entanglement with $t=O(1)$ are known~\cite{ekert2002direct,bouland2022quantum}.

\subsection{Black-box entanglement distillation}

\begin{theorem}[Inefficiency of black-box entanglement distillation]\label{thm:entanglement_distill}
Any algorithm that uses LOCC operations and $t=\mathrm{poly}(n)$ copies of
an arbitrary input state $\rho$ with $E_\text{D}(\rho)=\Theta(n)$ can only synthesize trivial pure states $\ket{\psi}$ with $E_\text{D}(\ket{\psi})=\operatorname{negl}(n)$.
\end{theorem}
\begin{proof}
\revA{The proof idea follows the proof for inefficiency of property testing}:
Lets assume there exists a entanglement distillation algorithm $\mathcal{A}_\text{E}$ which uses LOCC operations to turn $t=\mathrm{poly}(n)$ copies of state $\rho$ with $E_\text{D}(\rho)=\Theta(n)$ into a single entangled state $\ket{\psi}$ with $E_\text{D}(\ket{\psi})=\Omega(n^{-c})$ with $c>0$. 
The same algorithm $\mathcal{A}_\text{E}$ applied to a state $\rho'$ with $E_\text{D}(\rho')=0$ can only yield some separable state $\sigma\in\text{SEP}$ as LOCC operations cannot increase entanglement. 

We now draw $\rho$ from the GHSE with $m=\omega(\log n)$ which has $E_\text{D}(\rho)=\Theta(n)$, while $\rho'=I/2^{n}$ is the maximally mixed state with $E_\text{D}(\rho')=0$. 
GHSE and maximally mixed state are indistinguishable for any algorithm with $t=\mathrm{poly}(n)$ copies. 
However, the proposed entanglement distillation algorithm $\mathcal{A}_\text{E}$ could be used to distinguish both ensembles, as the distilled  state $\ket{\psi}$ can be efficiently distinguished from non-entangled states~\cite{ekert2002direct}. 

Thus, from contradiction the proposed entanglement distillation algorithm $\mathcal{A}_\text{E}$ cannot exist.
\end{proof}

\begin{corollary}[Lower bound on purity for black-box entanglement distillation]\label{thm:entanglement_distill_purity}
For black-box entanglement distillation of state $\ket{\psi}$ with $E_\text{D}(\ket{\psi})=\Omega(n^{-c})$ with $c>0$ using $t=\mathrm{poly}(n)$ copies of input state $\rho$, one requires purity $\text{tr}(\rho^2)=\Omega(n^{-c})$.
\end{corollary}
This lower bound on purity follows directly from the proof of Theorem~\ref{thm:entanglement_distill} which was demonstrated using a state ensembles with purity $\text{tr}(\rho^2)=\operatorname{negl}(n)$.

\section{Noise robustness of PRDM}\label{sec:noise}
A PRDM is noise robust to a noise channel if it remains PRDM after application of the channel. Here, we restrict ourselves to unital noise channels, i.e.\ noise channels where the maximally mixed state is a fixed point. Such noise channels encompass many noise models in quantum information, including depolarizing noise, Pauli channels,  and dephasing noise. 

First, we show that PRDMs with $m=O(\log n)$ are not robust to unital noise channels. In particular, any PRDM with non-neglible purity $\text{tr}(\rho^2)$ cannot be noise robust. Then, we show that PRDMs with $m=\omega(\log n)$ are robust to unital noise.

\begin{lemma}\label{general PRDM noise robustness}
    PRDMs with non-negligible purity $\text{tr}(\rho^2)=\Omega(n^{-c})$ with $c>0$ are not noise robust to unital channels. This includes all PRDMs with $m=O(\log n)$. 
\end{lemma}
\begin{proof}
    We use the global depolarisation channel as example of a unital noise model
    \begin{equation}
        \Lambda_p(\rho) = (1-p)\rho + \frac{p}{2^n}I .
    \end{equation}
    Now, it is easy to see that the expected purity of the globally depolarised state is:
    \begin{equation}
        \tr (\Lambda_p(\rho)^2) = (1-p)^2 \tr \rho^2 + \operatorname{negl}(n).
    \end{equation}
    Now, for noise probability $p = \Omega(n^{-c})$, we can compare the purity of ensembles $\eta_{n, m}$ and the ensemble $\mathcal{E}_{n, m}$ of globally depolarised $\eta_{n, m}$ states using SWAP test~\cite{barenco1997stabilization, garcia2013swap, beckey2023multipartite} to get,
    \begin{equation}
        \left |\Eset{\rho \leftarrow \eta_{n, m}}[\Pr_{\text{SWAP}}(\rho)] - \Eset{\rho \leftarrow \mathcal{E}_{n, m}}[\Pr_{\text{SWAP}}(\Lambda_p(\rho_k)] \right | = \left |(p - p^2/2) \Eset{\rho \leftarrow \eta_{n, m}} [\tr \rho^2 ]\right | = \Omega(n^{-c}\Eset{\rho \leftarrow \eta_{n, m}} [\tr \rho^2])\,.
    \end{equation}
    Thus, the SWAP test only fails to distinguish ensembles $\mathcal{E}_{n, m}$ and $\eta_{n, m}$ efficiently iff $\Eset{\rho \leftarrow \eta_{n, m}} [\tr \rho^2] = \operatorname{negl}(n)$. Thus negligible expected purity is a necessary condition for PRDMs to be noise-robust. For PRDMs, it is easy to see that $\text{tr}(\rho_k^2)\geq 2^{-m}$. Thus, for $m=O(\log n)$ PRDMs have non-negligible purity, and thus are not robust to noise. 
\end{proof}

Next, we show that PRDMs with $m=\omega(\log n)$ are robust to unital noise:
\begin{theorem}[PRDMs are robust to unital noise]
    PRDMs with $m=\omega(\log n)$ are robust under efficiently implementable unital noise channels, i.e. channels where the identity is the fixed point $\Gamma(I)=I$. In particular, if $\{\rho_k\}_k$ is PRDM, then $\{\Gamma(\rho_k)\}_k$ is also PRDM. 
\end{theorem}
\begin{proof}
PRDMs are indistinguishable from the maximally mixed state for any efficient algorithm for $m=\omega(\log n)$
  \begin{equation}
        \Big{|}\Pr_{k \leftarrow \mathcal{K}}[\mathcal{A}(\rho_k^{\otimes t}) = 1] - \Pr [\mathcal{A}((I_{n}/2^{n})^{\otimes t}) = 1]\Big{|} = \operatorname{negl}(n)\,.
    \end{equation}
In particular, let us define the efficient algorithm $\mathcal{A}(\Gamma(\rho))$ to which the indistinguishability must also apply. 
Now, we use $\mathcal{A}$ and apply the the fixed point condition $\Gamma(I)=I$ for the maximally mixed state
  \begin{equation}
        \Big{|}\Pr_{k \leftarrow \mathcal{K}}[\mathcal{A}(\Gamma(\rho_k)^{\otimes t}) = 1] - \Pr [\mathcal{A}((I_{n}/2^{n})^{\otimes t}) = 1]\Big{|} = \operatorname{negl}(n)\,, 
    \end{equation}
which is also indistinguishable. 
\end{proof}

\revA{
\section{Separation in security between PRDM and PRS}\label{sec:kretschmer}

Ref.~\cite{kretschmer2021quantum} showed that PRS and Haar random states can be sample-efficiently distinguished using classical shadows, where the algorithm requires a $\mathsf{PostBQP}$ oracle and post-selection. Here, sample efficient means that the algorithm uses $t=\text{poly}(n)$ copies of the state.
The idea of the attack is to use classical shadows to check whether the given state has high overlap with keyed states $\rho_k$ from the PRDM.
We now show that this attack does not work for PRDMs with $m=\omega(\log n)$: 

\begin{proposition}
\label{prop:Kretschmer}[Security of PRDMs against Kretschmer's attack]
    To distinguish PRDM $\rho_k$ with $m=\omega(\log n)$ from the GHSE using observable $O_k = \rho_k$ as shown in Ref.~\cite{kretschmer2021quantum}  requires super-polynomially many copies.
\end{proposition}
\begin{proof}
    The attack computes overlaps $\text{tr}(\rho_k O_i)$ where observable $O_i=\rho_i$ is a state from the PRDM. 
    For $\omega(\log n)$, we have $\text{tr}(\rho_k \rho_i)=\operatorname{negl}(n)$ for any $i$, including the correct state $i=k$. This is because $\rho_k$ has negligible purity, i.e. $\text{tr}(\rho_k^2)=\operatorname{negl}(n)$.
    Thus to measure the overlap, the additive precision required will be $\epsilon = \operatorname{negl}(n)$. Now,  classical shadows~\cite{lowe2022lower} have a lower bound on sample complexity $T = \Omega(n \log(M)^2/\epsilon^2)$~\cite{buadescu2021improved} where $M$ is the number of observables to be measured. As $\epsilon=\operatorname{negl}(n)$, we get the number of samples as $T = \Omega(1/\operatorname{negl}(n))$.
\end{proof}

Next, any sample-efficient algorithm that distinguishes PRDMs and GHSE also can distinguish PRS and Haar random states. Note that we do not make any assumptions about the computational complexity of the algorithm, and only demand sample efficiency, i.e. needing at most $t=\text{poly}(n)$ copies of the state:

\begin{proposition}[PRDM attack implies PRS attack]\label{prop:PRDMPRS}
    Any sample-efficient algorithm that distinguishes PRDM and GHSE also sample-efficiently distinguishes PRS from Haar random states.
\end{proposition}
\begin{proof}
    PRDMs can be efficiently constructed from PRS by partial trace as shown in~\eqref{eq:tracePRDM}. Further, GHSE can be constructed the same way via partial tracing of Haar random states. Assuming a sample-efficient algorithm $\mathcal{A}(\rho^{\otimes t})$ with $t=\text{poly}(n)$ that distinguishes PRDMs for GHSE, this algorithm also distinguishes PRS from Haar random states. In particular, one first traces out $m$ qubits, then gives it to $\mathcal{A}$ to distinguish.
\end{proof}
By combining Proposition~\ref{prop:Kretschmer} and Proposition~\ref{prop:PRDMPRS}, we immediately get a separation in security between PRDM and PRS:
\begin{theorem}[Security separation between PRDM and PRS (see \SM{}~\ref{sec:kretschmer})]\label{thm:separation_sup}
    Every sample-efficient algorithm that distinguishes PRDMs from GHSE also distinguishes PRSs from Haar random states, while there exist $\mathsf{PostBQP}$ oracles that sample-efficiently distinguish PRSs from Haar random states yet fail for PRDM and GHSE.
\end{theorem}
}

\section{Pseudomagic}\label{sec:pseudomagic}

Here, we explicitly construct a PRDM ensemble which we call the pseudomagic ensemble: This pseudomagic ensemble consists of two computationally indistinguishable ensembles with different log-free robustness of magic $\text{LR}(\rho)$. In particular, the first ensemble has high magic, while the second ensemble has zero magic. 
The gap in terms of magic between the two ensembles is maximal  $f(n)=\Theta(n)$ vs $g(n)=0$.
First, we define pseudomagic state ensemble:

\begin{definition}[Pseudomagic]
A \emph{pseudomagic state ensemble} with gap $f(n)$ vs. $g(n)$ consists of two ensembles of $n$-qubit states $\rho_k$ and $\sigma_k$, indexed by a secret key $k \in \mathcal{K}$, $k\in\{0,1\}^{\mathrm{poly}(n)}$ with the following properties:

\begin{enumerate}
\item \emph{Efficient Preparation}: Given $k$, $\rho_k$ (or $\sigma_k$, respectively) is efficiently preparable by a uniform, poly-sized quantum circuit.

\item \emph{Pseudomagic}: With probability $\geq 1 - \frac{1}{\mathrm{poly}(n)}$ over the choice of $k$, the log-free robustness of magic $\text{LR}(\rho)=\log(\min \vert c_\phi\vert \text{ s.t }  \rho=\sum_{\phi \in \text{STAB}}c_\phi \phi)$ for $\rho_k$ (or $\sigma_k$, respectively) is $\Theta(f(n))$ (or $\Theta(g(n))$, respectively).

\item \emph{Indistinguishability}: For any polynomial $p(n)$, no poly-time quantum algorithm can distinguish between the ensembles of $\mathrm{poly}(n)$ copies 
with more than negligible probability. That is, for any poly-time quantum algorithm $A$, we have that
\[\left| \Pr_{k \gets \mathcal{K}} [A(\rho_k^{\otimes \mathrm{poly}(n)}) = 1] - \Pr_{k \gets \mathcal{K}} [A(\sigma_k^{\otimes \mathrm{poly}(n)}) = 1] \right| = \operatorname{negl}(n)\,.\]
\end{enumerate}
\end{definition}

\begin{theorem}[Maximal pseudomagic]
There are pseudomagic state ensembles with pseudomagic gap $f(n)=\Theta(n)$ vs $g(n)=0$. 
\end{theorem}
\begin{proof}
Our second ensemble is the maximally mixed state $\{I_n/2^n\}$ with zero magic $\text{LR}=g(n)=0$. We note that the maximally mixed state can be efficiently prepared via $n$ Bell pairs $\ket{\Phi}=[\frac{1}{\sqrt{2}}(\ket{00}+\ket{11})]^{\otimes n}$, and tracing out the last qubit of each Bell pair, i.e. $I_n/2^n=\text{tr}_n(\ket{\Phi}\bra{\Phi})$. 

As our first ensemble, we construct the PRDMs $\chi_{n,m}=\{\text{tr}_m(\ket{\psi_k} \bra{\psi_k})\}_{k \leftarrow \mathcal{K}}$ with $m=\text{polylog}(n)$ which we define as $\chi_n\equiv \chi_{n,\text{polylog}(n)}$.
Here, we have the $n+m$ qubit PRS construction via the binary phase state $\ket{\psi_k}=2^{-(n+m)/2}\sum_x (-1)^{f_k(x)}\ket{x}$ with pseudorandom function $f_k(x): \{0,1\}^{n+m}\rightarrow\{0,1\}$ by Ref.~\cite{ji2018pseudorandom,brakerski2019pseudo,bouland2022quantum, ananth2022pseudorandom}. 
Now, we need to compute the magic of this PRDM. 

We can use the same formalism as for computing the log-free robustness for the GHSE, i.e.~\eqref{eq:logfreebound}. 
Further, we now use the stabilizer R\'enyi entropy instead of the stabilizer fidelity. In particular, we have~\cite{haug2023stabilizer} 
\begin{equation}
    F_\text{STAB}(\ket{\psi})\leq 2^{-\frac{1}{4}M_2(\ket{\psi})}\,,
\end{equation}
where we have the $2$-R\'enyi stabilizer entropy~\cite{leone2021renyi}
\begin{equation}
    M_2(\ket{\psi})=-\log\left(2^{-n}\sum_{P\in\mathcal{P}_{n+m}} \bra{\psi}P\ket{\psi}^4\right)
\end{equation}
where $\mathcal{P}_{n+m}$ is the set of unsigned tensor products of Pauli operators with $\vert \mathcal{P}_n\vert=4^{n+m}$.
Thus, we can write
\begin{equation}\label{eq:logfreebound2}
\text{LR}(\text{tr}_m(\ket{\psi}\bra{\psi}))\geq -\log(F_\text{STAB}(\ket{\psi}))-2m \geq 
\frac{1}{4}M_2(\ket{\psi})-2m\,.
\end{equation}
For the $2$-R\'enyi stabilizer entropy, a lower bound was found for the binary phase state in Ref.~\cite{gu2023little}. 
In particular, when the pseudorandom function is sampled randomly from the ensemble of 8-wise independent pseudorandom functions, we have $M_2(\ket{\psi_k})=\Theta(n+m)$ overwhelming probability~\cite{gu2023little}. 
Thus, the average log-free robustness of magic for the PRDM is given by
\begin{equation}
\Eset{k}[\text{tr}_m(\ket{\psi_k}\bra{\psi_k})]=\Theta(n)-2m\,.
\end{equation}
In particular, for $m=\text{polylog}(n)$ we have
\begin{equation}
\Eset{k}[\text{tr}_m(\ket{\psi_k}\bra{\psi_k})]=f(n)=\Theta(n)\,.
\end{equation}

\end{proof}

\section{Pseudocoherence}\label{sec:pseudocoherence}

Here we explicitly construct a special PRDM ensemble which we call the pseudocoherence ensemble: This pseudocoherence ensemble consists of two computationally indistinguishable ensembles with different relative entropy of coherence $C(\rho)$. In particular, the first ensemble has high coherence, while the second ensemble has zero coherence. 
The gap in terms of coherence between the two ensembles is maximal  $f(n)=\Theta(n)$ vs $g(n)=0$.
First, we define pseudocoherent state ensemble:

\begin{definition}[Pseudocoherent State Ensemble]\label{def:pseudocoherent ensemble}
A \emph{pseudocoherent state ensemble} with gap $f(n)$ vs. $g(n)$ consists of two ensembles of $n$-qubit states $\rho_k$ and $\sigma_k$, indexed by a secret key $k \in \mathcal{K}$, $k\in\{0,1\}^{\mathrm{poly}(n)}$ with the following properties:

\begin{enumerate}
\item \emph{Efficient Preparation}: Given $k$, $\rho_k$ (or $\sigma_k$, respectively) is efficiently preparable by a uniform, poly-sized quantum circuit.

\item \emph{Pseudocoherence}: With probability $\geq 1 - \frac{1}{\mathrm{poly}(n)}$ over the choice of $k$, the relative entropy of coherence  of $\rho_k$ (or $\sigma_k$, respectively) is $\Theta(f(n))$ (or $\Theta(g(n))$, respectively).

\item \emph{Indistinguishability}: For any polynomial $p(n)$, no poly-time quantum algorithm can distinguish between the ensembles of $\mathrm{poly}(n)$ copies 
with more than negligible probability. That is, for any poly-time quantum algorithm $\mathcal{A}$, we have that
\[\left| \Pr_{k \gets \mathcal{K}} [A(\rho_k^{\otimes \mathrm{poly}(n)}) = 1] - \Pr_{k \gets \mathcal{K}} [A(\sigma_k^{\otimes \mathrm{poly}(n)}) = 1] \right| = \operatorname{negl}(n)\,.\]
\end{enumerate}
\end{definition}

\begin{theorem}[Maximal pseudocoherence]
There exists pseudocoherent state ensembles with pseudocoherence gap $f(n)=\Theta(n)$ vs $g(n)=0$. 
\end{theorem}
\begin{proof}
    The second ensemble is simply the maximally mixed state $I_n/2^n$. 
    The maximally mixed state has coherence $C(I_n/2^n)=g(n)=0$.

    We construct the first ensemble as $\chi_{n,m} = \{\operatorname{tr}_m(\ket{\psi_k})\}_{k\in\mathcal{K}}$
    where we choose $m=\text{polylog}(n)$ and define $\chi\equiv \chi_{n,\text{polylog}(n)}$.
    Here, $\mathcal{K}$ is the key space, $\ket{\psi_k}=2^{-(n+m)/2}\sum_x (-1)^{f_k(x)}\ket{x}$ the $n+m$ qubit binary phase state with pseudorandom function $f_k(x): \{0,1\}^{n+m}\rightarrow\{0,1\}$.
    The PRDM can be efficiently constructed using the PRS construction of Ref.~\cite{brakerski2019pseudo}.

    The coherence of $\chi_{n,m}$ is straightforward to compute. 
    In particular, 
    \begin{equation}\label{eq:coherencePRDM_sup}
        \Eset{k}[C(\operatorname{tr}_m(\ket{\psi_k}))] = \Eset{k}[S(\Delta[\operatorname{tr}_m(\ket{\psi_k})])] - \Eset{k}[S(\operatorname{tr}_m(\ket{\psi_k}))]\geq n-m .
    \end{equation}
    The last term can be upper bounded as $\Eset{k}[S(\operatorname{tr}_m(\ket{\psi_k}))]\leq m$ due to it being at most rank $m$ as shown in~\eqref{eq:neumannrank}. 
    The first term in~\eqref{eq:coherencePRDM_sup} can be directly computed. It is easy to see that 
    $S(\Delta[\ket{\psi_k}])=n+m$ as $\Delta[\ket{\psi_k}]=I_{n+m}/2^{n+m}$. 
    Thus, we have 
    $S(\text{tr}_m(\Delta[\ket{\psi_k}]))=n$. 
    Next, we note that one can interchange partial trace and diagonal operator $\Delta[\rho]$
    \begin{equation}
   \text{tr}_m(\Delta[\rho])=\Delta[\text{tr}_m(\rho)].
    \end{equation}
    We have
        \begin{equation}
        \Delta[\text{tr}_m(\rho)]=  \sum_{j}\ket{j} \bra{j}\left(\sum_{\ell}\bra{\ell}\rho\ket{\ell}\right) \ket{j}\bra{j}
    \end{equation}
    and
        \begin{equation}
        \text{tr}_m(\Delta[\rho])=\sum_{\ell} \bra{\ell} \left(\sum_{ij}\ket{ij}\bra{ij}\rho\ket{ij}\bra{ij}\right)\ket{\ell}
    \end{equation}
    where $\ell$ and $i$ is the summation over the $m$ qubits, while $j$ the summation over the $n$ qubits.

    Now, one can easily show
    \begin{gather}
        \text{tr}_m(\Delta[\rho])=\sum_{\ell} \bra{\ell} \left(\sum_{ij}\ket{ij}\bra{ij}\rho\ket{ij}\bra{ij}\right)\ket{\ell}=\sum_{\ell}  \left(\sum_{ij}\delta_{\ell i}\ket{j}\bra{ij}\rho\ket{ij}\bra{j}\delta_{\ell i}\right)=\\
        \sum_{\ell}  \left(\sum_{j}\ket{j}\bra{\ell j}\rho\ket{\ell j}\bra{j}\right)=\sum_{j}\ket{j} \bra{j}\left(\sum_{\ell}\bra{\ell}\rho\ket{\ell}\right) \ket{j}\bra{j}=\text{tr}_m(\Delta[\rho]).
    \end{gather}

    With this result and $\Delta[\ket{\psi_k}]=I_{n+m}/2^{n+m}$ we immediately get
        \begin{equation}
        S(\Delta[\text{tr}_m(\ket{\psi_k})])=n\,.
    \end{equation}
    Now, finally we choose $m=\text{polylog}(n)$ to get
    \begin{equation}
        \Eset{k}[C(\operatorname{tr}_m(\ket{\psi_k}))] = f(n)=\Theta(n)\,.
    \end{equation}
\end{proof}

\revA{
\section{Pseudoentanglement}\label{sec:pseudoentanglement}

Pseudoentanglement consists of two computationally indistinguishable ensembles of states, one with high entanglement while the other has low entanglement~\cite{bouland2022quantum}. Here, we use distillable entanglement $E_\text{D}$ as well as entanglement cost $E_\text{C}$.

We now construct an ensemble of mixed states that has a maximal gap in terms of entanglement between the two ensembles, i.e. $f(n)=\Theta(n)$ vs $g(n)=0$. Note that for pure states this gap is only $f(n)=\Theta(n)$ vs $g(n)=\text{polyog}(n)$~\cite{bouland2022quantum}.

First, we define pseudoentanglement:

\begin{definition}[Pseudoentangled state ensemble]\label{def:pseudoentangled ensemble}
A \emph{pseudoentangled state ensemble} with gap $f(n)$ vs $g(n)$ consists of two ensembles of $n$-qubit states $\rho_k$ and $\sigma_k$, indexed by a secret key $k \in \mathcal{K}$, $k\in\{0,1\}^{\mathrm{poly}(n)}$ with the following properties:

\begin{enumerate}
\item \emph{Efficient Preparation}: Given $k$, $\rho_k$ (or $\sigma_k$, respectively) is efficiently preparable by a uniform, poly-sized quantum circuit.

\item \emph{Pseudoentanglement}: With probability $\geq 1 - \frac{1}{\mathrm{poly}(n)}$ over the choice of $k$, the distillable entanglement $E_\text{D}$  of $\rho_k$ (or $\sigma_k$, respectively) is $\Theta(f(n))$ (or $\Theta(g(n))$, respectively). 

\item \emph{Indistinguishability}: For any polynomial $p(n)$, no poly-time quantum algorithm can distinguish between the ensembles of $\mathrm{poly}(n)$ copies 
with more than negligible probability. That is, for any poly-time quantum algorithm $\mathcal{A}$, we have that
\[\left| \Pr_{k \gets \mathcal{K}} [A(\rho_k^{\otimes \mathrm{poly}(n)}) = 1] - \Pr_{k \gets \mathcal{K}} [A(\sigma_k^{\otimes \mathrm{poly}(n)}) = 1] \right| = \operatorname{negl}(n)\,.\]
\end{enumerate}
\end{definition}

For entanglement, we have to choose a bipartition relative to which we measure the entanglement. Here, we assume that any bipartition $A$, $B$ with $n_A\leq n_B$ can be chosen, where the distinguisher may make the choice. Note that to get a large pseudoentanglement gap, trivially one must choose $n_A=\Theta(n)$, else the entanglement cannot be large trivially.

\begin{theorem}[Maximal pseudoentanglement]
For any chosen bipartition $A$, $B$ with $n_A\leq n_B$ with $n_A=\Theta(n)$, there exists pseudoentangled state ensembles with pseudoentanglement gap $f(n)=\Theta(n)$ vs $g(n)=0$ with high probability. 
\end{theorem}
\begin{proof}
    The second ensemble is the maximally mixed state $I_n/2^n$. 
    The maximally mixed state has distillable entanglement $E_\text{D}(I_n/2^n)=g(n)=0$, as well as entanglement cost $E_\text{C}(I_n/2^n)=0$.

    We construct the first ensemble as $\chi_{n,m} = \{\operatorname{tr}_m(\ket{\psi_k})\}_{k\in\mathcal{K}}$
    where we choose $m=\text{polylog}(n)$ and define $\chi\equiv \chi_{n,\text{polylog}(n)}$.
    Here, $\mathcal{K}$ is the key space, $\ket{\psi_k}=2^{-(n+m)/2}\sum_x (-1)^{f_k(x)}\ket{x}$ the $n+m$ qubit binary phase state with pseudorandom function $f_k(x): \{0,1\}^{n+m}\rightarrow\{0,1\}$.
    We have $\rho=\text{tr}_m(\chi_{n,m})$ and $\rho_A=\text{tr}_B(\rho)$.
    Now, we use the hashing inequality as given in~\eqref{eq:entSimple} to bound the entanglement
    \begin{equation}
        \Eset{k}[E_\text{D}(\rho)]\geq  \Eset{k}[-\log(\text{tr}(\rho_A^2))-\log(\text{rank}(\rho))]\,.
    \end{equation}
    We have $\log(\text{rank}(\chi_{n,m}))\leq m$ by construction. 
    Next, depending on the choice of $f_k(x)$, the entanglement can be tuned~\cite{bouland2022quantum}. In particular, one can choose $f_k(x)$ such that 
    \begin{equation}\label{eq:entPRDM}
        \Eset{k}[\log(\text{tr}(\rho_A^2))] =\Theta(n)
    \end{equation}
    with high probability for any choice of $A$. 
    Thus, we get for $m=\text{polylog}(n)$
    \begin{equation}
        \Eset{k}[E_\text{D}(\rho)]= \Theta(n)\,.
    \end{equation}
    As $E_\text{C}\geq E_\text{D}$, we also have $\Eset{k}[E_\text{C}(\rho)]\geq \Theta(n)$ for the entanglement cost.

\end{proof}

We note that Ref.~\cite{goulao2024pseudo} found similar pseudoentanglement gaps, but for a more restrictive definition of pseudoentanglement: In Ref.~\cite{goulao2024pseudo}, the bipartition was fixed beforehand, whereas in our definition, the pseudoentanglement gap applies to arbitrary bipartitions.
}

\section{Noise-robust EFI pairs}\label{sec:EFI}
Here, we show that PRDMs can be used to construct EFI pairs.
An EFI pair is a pair of efficient quantum algorithms which prepare states that are statistically far but computationally indistinguishable:
\begin{definition}[EFI pairs~\cite{brakerski2022computational}]
    We call $\nu= \left(\nu_{b, \kappa}\right)$ a pair of EFI states if it satisfies the following criteria:
    \begin{enumerate}
        \item \emph{Efficient generation}: There exists efficient quantum algorithm $\mathcal{G}$ that on input $\left(1^\kappa, b\right)$ for some integer security parameter $\kappa$ and $b \in\{0,1\}$, outputs the mixed state $\nu_{b, \kappa}$.
        \item \emph{Statistically far}: $\Vert\nu_{0, \kappa}- \nu_{1, \kappa}\Vert_1=\Omega(1/\text{poly}(\kappa))$.
        \item \emph{Computational indistinguishability}: $\left(\nu_{0, \kappa}\right)_\kappa$ is computationally indistinguishable to $\left(\nu_{1, \kappa}\right)_\kappa$.
    \end{enumerate}
\end{definition}

\subsection{EFI pairs from PRDM}

We now show that PRDM  and maximally mixed states form EFI pairs for a security parameter $\kappa$ that scales linearly in $n$:
\begin{theorem}[EFI pairs with PRDM]
    PRDMs constructed by tracing out $\omega(\log n)< m<\frac{n}{2}(1-c)-\frac{1}{2}$ qubits from an $n+m$ qubit PRS with security parameter $\kappa=c(n+m)$ and constant $0<c<1$, and the maximally mixed state $I_n/2^n$ form EFI pairs.
\end{theorem}
\begin{proof}
We choose $\nu_0$ to be a PRDM, and $\nu_1$ the maximally mixed state.
Here, we note that the distinguishing algorithm (both for the statistical and computational case) only knows the integer security parameter $\kappa$, however does not know the specific key $k\in\{0,1\}^{\kappa}$.
Thus, the distinguisher only sees the ensemble average over all keys $k$. The corresponding ensemble density matrix for the distinguisher is given by  $\nu_1=I_n/2^n$, and $\nu_0=2^{-\kappa}\sum_{k\in\{0,1\}^{\kappa}} \rho_k$ with $\rho_k\in\zeta_{n,m}$, where we have the PRDM construction $\zeta_{n,m}=\{\text{tr}_m(\ket{\psi_k} \bra{\psi_k})\}_{k\in\{0,1\}^{\kappa}}$ from PRS $\ket{\psi_k}$.

Efficient generation and computationally indistinguishability has already been shown in the main text, which gives us the lower bound $m=\omega(\log n)$. 

Now, statistical distinguishability follows from the Fannes-Audenaert inequality~\cite{audenaert2007sharp}:
\begin{equation}
    |S(\nu_1) - S(\nu_0)| \leq T \log(2^n - 1) + H(\{T,1-T\}),
\label{eq:fannes}
\end{equation}
where $S(\rho)$ is the von Neumann entropy of $\rho$,
\begin{equation}
    T = \text{TD}(\nu_1, \nu_0)=\frac{1}{2}\Vert \nu_1-\nu_0\Vert_1
\end{equation}
is the trace distance between $\nu_1$ and $\nu_0$, and $H(\{T,1-T\})$ is the binary entropy of the probability distribution $\{T, 1-T\}$. 
Since $S(\nu_1) = n$ and $H(\{T,1-T\}) \leq 1$, rearranging~\eqref{eq:fannes} yields
\begin{equation}\label{eq:td_lowerbound_fannes}
    T \geq 1 - \frac{S(\nu_0)+1}{n}.
\end{equation}
To achieve distinguishability we have to achieve a lower bound of $T \geq 1/\text{poly}(n)$. 
Now, we show how to choose $S(\nu_0)$ and $m$ to achieve this.

Now, let us compute $S(\nu_0)$. 
Remember we have
$\nu_0=2^{-\kappa}\sum_{k\in\{0,1\}^{\kappa}}\text{tr}_m(\ket{\psi_k} \bra{\psi_k})$ with PRS $\ket{\psi_k}$. 
First, we bound the von-Neumann entropy of each state of the ensemble 
\begin{equation}
    S(\text{tr}_m(\ket{\psi_k} \bra{\psi_k}))\leq m
\end{equation}
due to the rank being maximally $m$ after partial trace. Now, the ensemble average over the key space can increase the entropy by at most $\kappa$~\cite{morimae2022quantum}, and we finally get
\begin{equation}\label{eq:constraint_Svn_EFI}
S(\nu_0)\leq \kappa +S(\text{tr}_m(\ket{\psi_k} \bra{\psi_k})) \leq \kappa +m \,.
\end{equation}
Inserting into~\eqref{eq:td_lowerbound_fannes} we get
\begin{equation}
    T \geq 1 - \frac{\kappa +m+1}{n}.
\end{equation}
Let us now assume a security parameter scaling as $\kappa=O(m+n)$. Then, we make the ansatz $\kappa=c(n+m)$ with some constant $0<c<1$~\cite{morimae2022quantum}
\begin{equation}
    T \geq 1 - \frac{cn+m(c+1)+1}{n}\,.
\end{equation}
We now demand $T\geq 1/\text{poly}(n)$, which is fulfilled when
\begin{equation}
     1 - \frac{cn+m(c+1)+1}{n} \geq 1/\text{poly}(n)\,,
\end{equation}
which can be satisfied for
\begin{equation}
    m < \frac{n}{2}(1-c)-\frac{1}{2}\,.
\end{equation}

\end{proof}

For general PRDMs we require $m=\omega(\log n)$ and $m< \frac{n}{2}(1-c)-\frac{1}{2}$. However, the upper bound comes from the upper bound of the entropy of PRDM states due to the partial trace over $m$ qubits. However, there are PRDM constructions with much lower entropy and thus wider upper bounds in terms of $m$. In particular, particular, we show how to achieve scaling $m<a n$, with arbitrary constant $a>0$ by using PRDM construction via the PRS proposed for pseudoentanglement~\cite{bouland2022quantum}. In Ref.~\cite{bouland2022quantum}, a PRS with 1D pseudo-area law was given, which has entanglement entropy of only $\text{polylog}(n+m)$ across every bipartition:
\begin{fact}[EFI pair with pseudoentanglement]
    PRDMs constructed by tracing out $m$ qubits from the $(n+m)$-qubit PRS with 1D pseudo-area law entanglement entropy  $\text{polylog}(n+m)$ across every bipartite cut, and the maximally mixed state $I_n/2^n$ form an EFI pair for $m=\frac{1-c}{c}n-\text{polylog}(n)$, security parameter $\kappa=c(n+m)$ and constant $0<c<1$.
\end{fact}
\begin{proof}
It is known that pseudoentangled PRS can be constructed from $n+m$-qubit binary phase PRS. They are guaranteed to have$S(\text{tr}_m(\ket{\psi}\bra{\psi}))=\Theta(\text{poly} \text{log} (n+m))$ entanglement across every bipartition , as measured by the von Neumann entropy~\cite{bouland2022quantum}. 
This gives us now a lower bound on $S(\nu_0)$ than in general case, where we have according to~\eqref{eq:constraint_Svn_EFI}
\begin{equation}
S(\nu_0)\leq \kappa +S(\text{tr}_m(\ket{\psi_k} \bra{\psi_k})) \leq \kappa + \text{poly} \text{log} (n+m) \,.
\end{equation}
Inserting into~\eqref{eq:td_lowerbound_fannes} we get
\begin{equation}
    T \geq 1 - \frac{\kappa + \text{poly} \text{log} (n+m)+1}{n}.
\end{equation}
We again make the ansatz $\kappa=c(n+m)$ for the security parameter with constant $0<c<1$
\begin{equation}
    T \geq 1 - \frac{cn+mc+\text{poly} \text{log}(n+m)+1}{n}\,.
\end{equation}
We now demand $T\geq 1/\text{poly}(n)$, which is fulfilled when
\begin{equation}
     1 - \frac{cn+mc+\text{poly} \text{log}(n+m)+1}{n} > 1/\text{poly}(n)\,.
\end{equation}
This is satisfied when
\begin{equation}
    m < \frac{1-c}{c} n -\text{polylog}(n)\,.
\end{equation}
Thus, by choosing $c$ arbitrarily close to $0$, i.e. $c\rightarrow 0$, one can have $m$ scale linear in $n$ with arbitrary prefactor.
\end{proof}

\subsection{Proof of noise-robustness}
We can also show that EFI pairs constructed from PRDMs and maximally mixed state are noise-robust. To this end, let us consider a general mixed unitary noise channel
\begin{equation}
    \Phi(\rho) = \sum_{i=1}^{r} p_i U_i \rho U_i^\dag,
\label{eq:mixed_unitary_channel}
\end{equation}
where $p_i$ are probabilities that sum to unity, $U_i$ are unitary operators and $r$ is the mixed-unitary rank of the channel $\Phi$.
\begin{theorem}[Noise-robust EFI pairs]
    PRDMs constructed by tracing out $m=\omega(\log n)$ qubits from $(n+m)$ qubit PRS  with security parameter $\kappa=cn$ with $0<c<1$, and the maximally mixed state $I_n/2^n$ remain EFI pairs under application of efficient mixed unitary channel $\Phi(\rho)= \sum_{i=1}^{r} p_i U_i \rho U_i^\dag$ whenever the Shannon entropy of its probabilities is bounded as
    \begin{equation}
        H(\{p_i\}_i)\leq n(1-c)-m-2\,.
    \end{equation}
\label{thm:noise_robust_EFI_theorem}
\end{theorem}
\begin{proof}
First, we note that the maximally mixed state is invariant under mixed unitary noise channels as it is its fixed point, i.e. $\Phi(\nu_0)=\nu_0$.
Next, we consider the PRDM under noise.
First, we note that PRDMs remain PRDMs after application of efficiently implementable unital channels. 
Now, we need to show that the PRDM remains statistically far.
Denoting $\nu_0^\prime = \Phi(\nu_0)$, we can bound the von Neumann entropy of the noisy state $\nu_0^\prime$ by
\begin{equation}
    S(\nu_0^\prime) = S\left(\sum_{i=1}^{r} p_i U_i \nu_0 U_i^\dag \right) \leq \sum_{i=1}^{r} p_i S(U_i \nu_0 U_i^\dag) + H(\{p_i\}_i) = S(\nu_0) + H(\{p_i\}_i),
\end{equation}
where $H(\{p_i\}_i)$ is the Shannon entropy for the probability distribution $\{p_i\}_i$. The analogous equation to~\eqref{eq:td_lowerbound_fannes} for the new trace distance $T^\prime = \text{TD}(\nu_0^\prime ,\nu_1)$ is
\begin{equation}
    T^\prime \geq 1 - \frac{S(\nu_0) + H(\{p_i\}_i) + 1}{n}.
\end{equation}
Noise robustness is achieved when $T^\prime=\Omega(1/\text{poly}(n))$. Thus, it is sufficient to have
\begin{equation}
    S(\nu_0) + H(\{p_i\}_i) \leq n - 1 - \frac{1}{\text{poly}(n)}.
\end{equation}
This sets a bound on the probability distribution for the Kraus operators,
\begin{equation}
    H(\{p_i\}_i) \leq n - S(\nu_0) - 1 - \frac{1}{\text{poly}(n)}.
\end{equation}
Since $S(\nu_0) \leq m+\kappa$ and we consider $m < n-1$, the above condition can be achieved when
\begin{equation}
    H(\{p_i\}_i) < n -m-\kappa- 2\,.
\end{equation}
We now choose $\kappa=c n$ with $0<c<1$, where we find
\begin{equation}
    H(\{p_i\}_i) < n(1-c) -m- 2\,.
\end{equation}

\end{proof}
For $m = \text{polylog}(n)$, the bound on $H(\{p_i\}_i)$ is asymptotically linear with the number of qubits $n$.

A noise model highly relevant in physical systems and commonly used to model noise in near-term and quantum error correction models is the local depolarizing noise $\Lambda_p^{\otimes n}(\rho)^{\otimes n}$ acting on all $n$ qubits. It consists of the local depolarizing channel $\Lambda_p(\rho)=p/4\sum_{\alpha\in\{x,y,z\}}\sigma^\alpha \rho \sigma^\alpha+(1-3p/4)\rho$, Pauli operators $\sigma^\alpha$ with $\alpha\in\{x,y,z\}$, and the depolarizing probability $p$.
This noise model has maximal rank of $H(\{p_i\}_i)$ with $r = 4^n$. Nonetheless, we show that EFI pairs remain robust against such noise even for relatively high constant noise rate $p$. 

\begin{corollary}[Noise-robust EFI pairs against local depolarizing noise]
PRDMs constructed by tracing out $m=\omega(\log n)$ qubits from $(n+m)$ qubit PRS  with security parameter $\kappa=c(n+m)$ with $0<c<1$ remain EFI pairs after application of local depolarizing channel on all $n$ qubits $\Lambda_p^{\otimes n}(\rho)$, when $H(\{1-3p/4,p/4,p/4,p/4\}) \leq (1-c) -m/n- 2/n$. In particular, when $m 
= \operatorname{polylog}(n)$ and $c=10^{-4}$, we have $p < \frac{1}{4} -O(\operatorname{polylog}(n)/n)$.
\end{corollary}
\begin{proof}
    Since the local depolarizing channel acts independently on the $n$ qubits, the Shannon entropy of the probability distribution in~\eqref{eq:mixed_unitary_channel} is
    \begin{equation}
        H(\{p_i\}_{i=1}^{4^n}) = n H(\{1-3p/4,p/4,p/4,p/4\}).
    \end{equation}
Using Theorem~\ref{thm:noise_robust_EFI_theorem}, noise robustness is guaranteed when
\begin{equation}
    H(\{1-3p/4,p/4,p/4,p/4\}) \leq (1-c) -\frac{m-2}{n}\,.
\end{equation}
Now, we choose $m = \text{polylog}(n)$ and choose a security parameter $\kappa=cn$ that scales linearly in $n$, but with constant $c$.
If we choose arbitrarily small $c\rightarrow0$, we find that noise-robustness is fulfilled for any $p < 0.252 -O(\text{polylog}(n)/n)$. 
Choosing $c=10^{-4}$, we have $p < 0.25 -O(\text{polylog}(n)/n)$, while for $c=1/12$ we have $p < 0.22 -O(\text{polylog}(n)/n)$.

\end{proof}

\section{PRS without quantum memory}\label{Pseudorandomness without quantum memory}
We define the notion of memoryless PRS which are indistinguishable to Haar random states for any efficient algorithm without access to quantum memory. 
In particular, the observer has access to polynomial many copies of the state, however can only perform efficient measurements on one copy at a time, where the measurements can be chosen adaptively on previous measurements outcomes.

We note that the previously introduced single-copy PRS is a special case of PRS without access to quantum memory~\cite{morimae2022quantum,chen2024power}. Single-copy PRS is indistinguishable for any observer with access to only a single copy of the state. Any PRS without quantum memory is a single-copy PRS, while there are single-copy PRS which are not PRS without quantum memory. A simple example is the ensemble of all computational basis states $\{\ket{j}\}_{j=1}^{2^n}$: It is a single-copy PRS as it is indistinguishable from Haar random states for a single copy, but not a PRS without quantum memory as one can distinguish it from multiple copies of Haar random states by testing the coherence, which can be done without the need of quantum memory. \revA{In particular, one measures the states in the computational basis and check whether there are collisions of bitstrings}~\cite{haug2023pseudorandom}.

First, let us define algorithms without quantum memory:
\begin{definition}[Learning without quantum memory~\cite{chen2022exponential}]\label{def:memoryless}
An algorithm $\mathcal{W}$ without quantum memory obtains classical
data from an oracle that prepares $\rho$ by performing arbitrary POVM measurements on $\rho$. For each access to the oracle, $\mathcal{W}$ can select a POVM $\{F_s\}_s$ that can depend on previous outcomes, and obtain the classical outcome $s$ with probability
$\operatorname{tr}(F_s\rho)$. After $T$
oracle accesses, $\mathcal{W}$ predicts the properties of $\rho$.
\end{definition}

As we will see, memoryless PRS are robust to noise for algorithms without quantum memory:
\begin{definition}[Memoryless PRS]
    Let $\lambda$ be the polynomial sized security parameter and $\mathcal{K}$ be the key space respectively dependent on the security parameter. 
    Then, a keyed family of pure quantum states $\{|\phi_k\rangle\}_{k \leftarrow \mathcal{K}}$ is called memoryless PRS, which is secure to algorithms $\mathcal{W}$ without quantum memory if:
    \begin{enumerate}
        \item Efficient Generation: There exist an efficient quantum algorithm $S$ such that $S(k, 1^{\lambda}) = |\phi_k\rangle$.
        \item Computational Indistinguishability without quantum memory: For a random key $k \in K$, given $t = poly(\lambda)$ copies of $|\phi_k\rangle$ are computationally indistinguishable from $t$ copies of Haar random states for any quantum polynomial time algorithm $\mathcal{W}$ with no quantum memory:
        \begin{equation}
            \left|\Pr_{k \leftarrow K}[D(|\phi_k\rangle^{\otimes m}) = 1] - \Pr_{|\psi\rangle \leftarrow \eta_H}[D(|\psi\rangle^{\otimes m}) = 1]\right| = \operatorname{negl}(\lambda)
        \end{equation}
    \end{enumerate}
\end{definition}
Any PRS is also a memoryless PRS, making memoryless PRS a weaker notion of pseudorandomness than PRS.

We describe the computational model of the memoryless adversary to be such that the adversary has black-box access to oracles that either prepare PRS or Haar random state. The main idea is to use the tree method~\cite{chen2022exponential} to perform many vs one distinguishability task by comparing the probability distributions over the leaf nodes. We will denote $p^{\rho}(l)$ as the probability of getting to leaf $l$ when the given state is $\rho$. For the given many vs one distinguishability task, we are then interested in the comparison of the expectation of this probability of the ensemble and the given state. Note that the POVM measurements that we are going to consider will be Rank-1 POVM measurements which with post-processing are equivalent to general POVM for the case when we are not interested in post-measurement quantum state as shown in~\cite{chen2022exponential}. Formally, we put it below.

\begin{definition}[Tree method for learning without quantum memory~\cite{chen2022exponential}] The tree method consists of a rooted tree $\mathcal{T}$ with the following features:
\begin{enumerate}
    \item At every depth of tree, the learner chooses a positive operator valued measurement(POVM) $\{2^n c^{u}_{s} |\psi^{u}_{s}\rangle \langle \psi^{u}_{s}|\}_s$ where $u$ represents the depth of tree to measure the given copy and record its output in the classical memory.
    \item Each node in the tree represents the state of classical memory and is assigned a probability which represents the probability of reaching that node from the root by successive adaptive POVM measurements. As an example, for a child node $v$ of $w$ connected through an edge corresponding to POVM element $2^n c^{u}_{s} |\psi^{u}_{s}\rangle \langle \psi^{u}_{s}|$
    \begin{equation}
        p(v) = p(w)2^n c^u_s \tr [\rho  |\psi^{u}_{s}\rangle \langle \psi^{u}_{s}|]
    \end{equation}
    where $p(v)$ and $p(w)$ are probabilities associated with nodes $v$ and $w$ respectively.
\end{enumerate}
\end{definition}

Note that, by definition, the root node will have probability $p^{\rho}(r) = 1$, and the probability distribution on leaf nodes is given by $\{p^{\rho}(l)\}_l$ which operationally represents the state of probability distribution over the state of classical memory. Now, we use this distribution over states of classical memory to compare the given ensembles. It is easy to see that the expected probability of reaching a leaf node $l$ on a given ensemble $\rho_H$ is given as the product of probabilities of nodes that are in the path from the root to the leaf $l$ i.e.
\begin{equation}\label{LeafDistribution}
    p^{\rho_H}(l) = \mathbb{E}_{\rho_H} \prod_{t = 1}^T 2^n c^{u_t}_{s_t} \tr [\rho_H |\psi^{u_t}_{s_t}\rangle \langle \psi^{u_t}_{s_t}|]
\end{equation}
where the path from root to leaf $l$ is given by nodes corresponding to POVM elements $\{2^n c^{u}_{s} |\psi^{u_t}_{s_t}\rangle \langle \psi^{u_t}_{s_t}|\}_{t =1}^T$ as $T$ is depth of the tree.

\begin{fact}[Le Cam one sided bound~\cite{chen2022exponential}]\label{LeCamOneside} For learning without quantum memory described using a rooted tree $\mathcal{T}$, if for all leaves of the tree $l$,
\begin{equation}
    \frac{\mathbb{E_{\rho}}p^{\rho}(l)}{p^{I/2^n}(l)} \geq 1 - \delta
\end{equation}
then the probability of distinguishing the ensemble $\{\rho\}$ from the maximally mixed state is upper bounded by $\delta$.
\end{fact}

\begin{fact}[~\cite{chen2022exponential}]\label{Permutation product}
    For any collection of pure states $\{|\psi_i\rangle \}_i$, we have:
    \begin{equation}
        \sum_{\pi \in S_T} \tr \Big[\pi (\otimes_{i = 1}^T |\psi_i\rangle \langle \psi_i|)\Big] \geq 1
    \end{equation}
    where $S_T$ represents all permutation over $T$-copies.
\end{fact}

\begin{theorem}\label{memory less PRS proof 1}
A memoryless adversary requires exponential copies to distinguish noisy Haar and maximally mixed state when noise is represented by a general unital CPTP map $\mathcal{E}[\rho] = \sum_i E_i \rho E_i^{\dagger}$.
\end{theorem}
\begin{proof}
    From Fact \ref{LeCamOneside}, we need to look at the expression $\frac{\mathbb{E}_{\rho \leftarrow \mu_n}p^{\rho}(l)}{p^{I/2^n}(l)}$ for every $l$ where $\mu_n$ represents Haar measure on $n$ qubit pure states. Using \eqref{LeafDistribution}, we rewrite the expression as:
    \begin{equation}
        \begin{split}
            \mathbb{E}_{\rho \leftarrow \mu_n}\Big(\prod_{t = 1}^T\frac{2^n c_t \langle \psi_{s_t}^{u_t} | \sum_i E_i \rho E_i^{\dagger} |\psi_{s_t}^{u_t}\rangle}{2^n c_t \langle \psi_{s_t}^{u_t}| \frac{I}{2^n}|\psi_{s_t}^{u_t}\rangle}\Big) &= \mathbb{E}_{\rho \leftarrow \mu_n}\Big(\prod_{t = 1}^T2^n \langle \psi_{s_t}^{u_t} | \sum_i E_i \rho E_i^{\dagger} |\psi_{s_t}^{u_t}\rangle \Big) \\
            &= \mathbb{E}_{\rho \leftarrow \mu_n} 2^{nT} \tr \Big [\Big(\sum_i E_i \rho E_i^{\dagger}\Big)^{\otimes T} \Big(\otimes_{t = 1}^T |\psi_{s_t}^{u_t}\rangle \langle \psi_{s_t}^{u_t}|\Big)\Big].
        \end{split}
    \end{equation}
    Now, using the cyclic property of trace, $\tr [\sum_i E_i \rho E_i^{\dagger} |\psi_{s_t}^{u_t}\rangle \langle \psi_{s_t}^{u_t}|] = \tr [\rho \sum_i E_i^{\dagger} |\psi_{s_t}^{u_t}\rangle \langle \psi_{s_t}^{u_t}| E_i]$. Using this, we get:
    \begin{equation}
        \begin{split}
            \mathbb{E}_{\rho \leftarrow \mu_n} 2^{nT} \tr \Big [\Big(\sum_i E_i \rho E_i^{\dagger}\Big)^{\otimes T} \Big(\otimes_{t = 1}^T |&\psi_{s_t}^{u_t}\rangle \langle \psi_{s_t}^{u_t}|\Big)\Big] = \mathbb{E}_{\rho \leftarrow \mu_n} 2^{nT} \tr \Big [\Big(\rho \Big)^{\otimes T} \Big(\otimes_{t = 1}^T \sum_i E_i^{\dagger}|\psi_{s_t}^{u_t}\rangle \langle \psi_{s_t}^{u_t}|E_i\Big)\Big].
        \end{split}
    \end{equation}
    Since the CPTP map $\mathcal{E}$ is unital thus we can define the new map $\mathcal{F}[\rho] = \sum_j E^{\dagger}_j \rho E_j$ which will also be CPTP. Using this fact, we can write $\sum_i E_i^{\dagger} |\psi_{s_t}^{u_t}\rangle \langle \psi_{s_t}^{u_t}|E_i = \sum_j p_{j_t} |\psi^t_j\rangle \langle \psi^t_j|$ for every $t$ where $\{|\psi_j^t\rangle\}_j$ are some pure states and $\sum_{j_t} p_{j_t} = 1$. Then we can write:
    \begin{equation}
        \begin{split}
            \mathbb{E}_{\rho \leftarrow \mu_n} 2^{nT} \tr \Big [\Big(\rho \Big)^{\otimes T}& \Big(\otimes_{t = 1}^T \sum_i E_i^{\dagger}|\psi_{s_t}^{u_t}\rangle \langle \psi_{s_t}^{u_t}|E_i\Big)\Big] \\
            &= \mathbb{E}_{\rho \leftarrow \mu_n} 2^{nT} \tr \Big [\Big(\rho \Big)^{\otimes T} \Big(\otimes_{t = 1}^T \sum_j p_{j}^t |\psi^t_j\rangle \langle \psi^t_j|\Big)\Big] \\
            & = \mathbb{E}_{\rho \leftarrow \mu_n} 2^{nT}\sum_{j_1, j_2 \ldots j_T} p_{j_1}p_{j_1}\ldots p_{j_T} \tr \Big [\Big(\rho \Big)^{\otimes T} \Big(\otimes_{t = 1}^T  |\psi^t_{j}\rangle \langle \psi^t_{j}|\Big)\Big] \\
            &= 2^{nT}\sum_{j_1, j_2 \ldots j_t} p_{j_1}p_{j_1}\ldots p_{j_T} {2^n + T - 1 \choose T} \frac{1}{T!} \sum_{\pi \in S_T} \tr \Big [\pi \Big(\otimes_{t = 1}^T  |\psi^t_{j}\rangle \langle \psi^t_{j}|\Big)\Big] ,
        \end{split}
    \end{equation}
    where in the last equality, we used the fact that $\mathbb{E}_{\rho \leftarrow \mu_n} \rho^{\otimes T} = {2^n + T - 1 \choose T}\frac{1}{T!}\sum_{\pi \in S_T} \pi $ where $\pi = \sum_{\bar{x}} |\pi(\bar{x})\rangle \langle \bar{x}|$. Now, using Fact \ref{Permutation product}, we have:
    \begin{equation}
    \begin{split}
        2^{nT}\sum_{j_1, j_2 \ldots j_t} p_{j_1}p_{j_1}\ldots p_{j_T} {2^n + T - 1 \choose T} \frac{1}{T!}& \sum_{\pi \in S_T} \tr \Big [\pi \Big(\otimes_{t = 1}^T  |\psi^t_{j}\rangle \langle \psi^t_{j}|\Big)\Big]\\ & \geq 2^{nT} {2^n + T - 1 \choose T}\frac{1}{T!} \prod_{t = 1}^T (\sum_{j_t} p_{j_t})\\
       & = \frac{2^{nT}}{(2^n) \ldots (2^n + T-1)}\\
        & \geq  \Big ( 1 - \frac{T}{2^n}\Big )^T .
    \end{split}
    \end{equation}
    Assume that $\frac{T^2}{2^n} = \operatorname{negl}(n)$ which also covers the case when $T = \text{poly}(n)$. In this case, we have:
    \begin{equation}
        \mathbb{E}_{\rho \leftarrow \eta_H}\Big(\prod_{t = 1}^T\frac{2^n c_t \langle \psi_{s_t}^{u_t} | \sum_i E_i \rho E_i^{\dagger} |\psi_{s_t}^{u_t}\rangle}{2^n c_t \langle \psi_{s_t}^{u_t}| \frac{I}{2^n}|\psi_{s_t}^{u_t}\rangle}\Big) \geq \Big ( 1 - \frac{T^2}{2^n}\Big) = 1 - \operatorname{negl}(n)
    \end{equation}    
    where we use binomial approximation to obtain the first inequality. Then, using this, we have the probability to distinguish noisy Haar and maximally mixed state according to Le Cam one point method is:
    \begin{equation}
        \delta \leq 1 - \mathbb{E}_{\rho \leftarrow \eta_H}\Big(\prod_{t = 1}^T\frac{2^n c_t \langle \psi_{s_t}^{u_t} | \sum_i E_i \rho E_i^{\dagger} |\psi_{s_t}^{u_t}\rangle}{2^n c_t \langle \psi_{s_t}^{u_t}| \frac{I}{2^n}|\psi_{s_t}^{u_t}\rangle}\Big) = \operatorname{negl}(n).
    \end{equation}
    But to efficiently distinguish the given ensembles, we require $\delta \geq 1/2 + \Omega(n^{-c})$, which is not possible for the case when $T^2/2^n = \operatorname{negl}(n)$. Thus, $T = \Omega(2^{n/2})$.
\end{proof}

\begin{fact}[Maximally mixed state and Haar random state are indistinguishable without quantum memory~\cite{chen2022exponential}] \label{CCHL_theorem}
In the absence of quantum memory, any learning algorithm requires $T = \Omega(2^{n/2})$ samples to distinguish whether it is sampled from Haar ensemble or singleton maximally mixed state.
\end{fact}

\begin{theorem}\label{thm:noise_robust_memoryless_PRS} An adversary without quantum memory cannot efficiently distinguish between PRS affected by arbitrary unital noise channels, and Haar random states.
\end{theorem}
\begin{proof}
    Using theorem \ref{memory less PRS proof 1}, we get that noisy Haar random state ensemble is indistinguishable from maximally mixed state for a memoryless adversary. Now, using Fact \ref{CCHL_theorem} and triangle inequality, we get that noisy Haar random states are statistically indistinguishable from noiseless Haar random states for memoryless adversary.  Now, this translates to PRS, from the fact that noisy Haar random state ensemble is computationally indistinguishable from noisy PRS thus providing computational indistinguishability of noisy PRS and Haar random states. 
\end{proof}
The above lemma gives us other notion of noise-robust pseudorandomness. Note that, the above proof is information theoretic which implies computational indistinguishability, but we also provide a weaker version of the above theorem simply by computational indistinguishability arguments based on the fact that maximally mixed state is indistinguishable from Haar ensemble in the many vs one distinguishability task without quantum memory.

\begin{theorem}[Weaker version of theorem \ref{thm:noise_robust_memoryless_PRS}]\label{memoryless PRS proof 2}
    Any adversary without quantum memory cannot efficiently distinguish between PRS affected by efficient unital noise channels $\Gamma$, and Haar random states.
\end{theorem}
\begin{proof}
    Let the PRS family be $\{|\phi_k\rangle\}_{k \in \mathcal{K}}$ and let the noise in PRS be given by the unital CPTP channel $\Gamma$. Then, we create 3 hybrids as follows:
    \begin{enumerate}
        \item \textbf{Hybrid 1}: $t$ copies of $\Gamma(|\phi_k\rangle)$ with $k \leftarrow \mathcal{K}$ and provide adversary with one copy at a time. 
        \item \textbf{Hybrid 2}: $t$ copies of $\Gamma(|\psi\rangle)$ with $|\psi\rangle \leftarrow \mu_{n}$ and provide adversary with one copy at a time. 
        \item \textbf{Hybrid 3}: $t$ copies of maximally mixed state on $n$ qubits: $\frac{I}{2^n}$ and provide adversary with one copy at a time. 
        \item \textbf{Hybrid 4}: $t$ copies of Haar random state $|\psi\rangle \leftarrow \mu_n$ and provide adversary with one copy at a time. 
    \end{enumerate}
    Now computational indistinguishability of $H_3$ and $H_4$ follows from Fact \ref{CCHL_theorem} as shown in~\cite{chen2022exponential}. To see the computational indistinguishability of $H_2$ and $H_3$, we first assume that the hybrids are distinguishable and there is a quantum polynomial time distinguisher $\mathcal{A}$ which given only $t = \text{poly}(n)$ copies of hybrids, can distinguish between them. Then we can design a new distinguisher $\mathcal{A}^{\prime}$ which first acts the efficient unital noise channel $\Gamma$ and then use the distinguisher $\mathcal{A}$ as subroutine to distinguish hybrids $H_3$ and $H_4$ in polynomial time with $t = \text{poly}(n)$ samples which contradicts with computational indistinguishability of hybrids $H_3$ and $H_4$. Thus, hybrid $H_2$ and $H_4$ are also computationally indistinguishable. Finally computational indistinguishability of $H_1$ and $H_2$ follows from definition of PRS and from the fact that $\Gamma$ is efficiently implementable.
\end{proof}

\begin{lemma}\label{examples of memoryless PRS that are not standard PRS}
        PRDMs are also memoryless PRS, but PRDMs is not the same as PRS.
\end{lemma}
\begin{proof}
 To prove that PRDM is also a memoryless PRS, it is easy to see that we can obtain $n$ qubit PRDM from $n+m$ qubits Haar ensemble by partial trace on $m$ qubits, and as partial trace is computationally efficient, we will have that $n$ qubit PRDM is close to $n$ qubit maximally mixed state. Then, by using transitivity of computational indistinguishability, it will also be computationally close to memoryless PRS. Furthermore, PRDM is not the same as PRS which can be seen by the fact that one can distinguish two copies of PRDM from two copies of PRS using the SWAP test. 
\end{proof}

\section{PRS with noisy quantum memory}\label{Pseudorandomness with noisy quantum memory}
Similar to standard PRS, we can define PRS with noisy quantum memory. Firstly, we define the notion of noisy quantum memory as inspired from~\cite{chen2023complexity}. The authors in~\cite{chen2023complexity} considered the adversary's quantum memory with two registers: state loading register $A$ with $n$ qubits and workspace register  $B$ with $\text{poly}(n)$ qubits. Then, the adversary can query the state preparing oracle $O$ and the noise channel acts on both state loading and memory register of adversary just after the query. Between any two queries, the adversary has access to noiseless operations and measurements. We allow the adversary to receive two copies of state from the oracle in one query. Then, an entangling noise $D_q$ acts on the adversary's quantum memory. We mention that this allows us to look at more general noise, particulalry entangling noise channels then just one copy noise channels.

\begin{definition}[Adversary with noisy quantum memory]\label{def:noisy_quantum_memory}
    An adversary with noisy quantum memory has an initial state $\sigma$ on $n^{\prime} \geq 2n$ qubits and it can query the oracle $O_i$ which prepares the state:
    \begin{equation}
        O_i(\sigma) = \rho_i \otimes \rho_i \otimes \tr_2 \sigma
    \end{equation}
    where $\tr_2\sigma$ is partial trace on state loading register of state $\sigma$. Then, the noise channel $D_q$ with noise probability $q$ acts on all qubits of adversary just after querying the oracle and prepares the state:
    \begin{equation}
        D_qO_i(\sigma) = D_q(\rho_i \otimes \rho_i \otimes \tr_2 \sigma).
    \end{equation}
\end{definition}

We can now formally define pseudorandom states with noisy quantum memory as below.
\begin{definition}[PRS with noisy quantum memory]
    Let $\lambda$ be the security parameter and let $\mathcal{H}$ and $\mathcal{K}$ be the Hilbert space and key space respectively both dependent on the security parameter. Then, a keyed family of pure quantum states $\{|\phi_k\rangle\}_{k \leftarrow \mathcal{K}}$ is called a pseudorandom state family with noisy quantum memory if:
    \begin{enumerate}
        \item Efficient Generation: There exist an efficient quantum algorithm $S$ such that $S(k, 1^{\lambda}) = |\phi_k\rangle$.
        \item Computational Indistinguishability: For a random key $k \in K$, given $t = poly(\lambda)$ copies of $|\phi_k\rangle$ are computationally indistinguishable from t copies of Haar random states for any quantum polynomial time algorithm $B$ with noisy quantum memory as defined in Def. \ref{def:noisy_quantum_memory}:
        \begin{equation}
            |\Pr_{k \leftarrow K}[B(|\phi_k\rangle^{\otimes m}) = 1] - \Pr_{|\psi\rangle \leftarrow \eta_H}[B(|\psi\rangle^{\otimes m}) = 1]| = \operatorname{negl}(\lambda).
        \end{equation}
    \end{enumerate}
\end{definition}
It is trivial to see that standard PRS is an example of PRS with noisy quantum memory when the noise probability is $O(\operatorname{negl}(n))$. Thus, this notion of PRS is much more general. Now, we will see that even this does not help in noise robustness. Particularly, we show it using Holevo Helstrom bound~\cite{bae2015quantum} for distinguishability on two copies of PRS subject to noise. 
For simplicity, we assume only $2n$ qubit quantum memory for the adversary, but it can be readily generalized to any $n^{\prime} \geq 2n$. We will denote the noise channel representing noise in the quantum memory using $D_q$ and the noise channel representing noise in the PRS as $\Gamma(\rho)$ in following results.

\begin{theorem}\label{noisymemory}
    The probability of distinguishing a noisy PRS from noiseless PRS for a noisy quantum memory adversary is lower bounded by $1/2 + \Omega(n^{-c})$ for the probability of noise $\Theta(1/poly(n))$ in the quantum memory.
\end{theorem}
\begin{proof}
     We compute the lower bound on TD$[D_qO_0(\sigma), D_qO_1(\rho)]$ where $D_q$ is general noise channel of form $D_q(\rho) = q\rho + (1-q)\mathcal{E}(\rho)$ where $\mathcal{E}$ is general CPTP map and $O_i$ prepares the two copies of noiseless PRS, $\rho \otimes \rho $ and two copies of noisy PRS, $\Gamma (\rho) \otimes \Gamma (\rho)$ for $i = 0$ and $i = 1$ respectively. Now,
     \begin{equation}
        \begin{split}
            \text{TD}[D_qO_0(\sigma), D_qO_1(\sigma)] = \frac{1}{2}\tr \Big{|}& q\mathbb{E}_{\rho \leftarrow \eta_H}\rho \otimes \rho + (1-q) \mathbb{E}_{\rho \leftarrow \eta_H}\mathcal{E}(\rho \otimes \rho) \\
            &- q \mathbb{E}_{\rho \leftarrow \eta_H}\Gamma(\rho) \otimes \Gamma(\rho) - (1-q)\mathbb{E}_{\rho \leftarrow \eta_H} \mathcal{E}(\Gamma(\rho)\otimes \Gamma(\rho)) \Big{|}.
        \end{split}
    \end{equation}
    Now, using the reverse triangle inequality given by $\tr |x - y| \geq | \tr |x| - \tr |y| | \geq \tr |x| - \tr |y|$, we obtain 
    \begin{equation}
        \begin{split}
            \text{TD}[D_qO_0(\sigma), D_qO_1(\sigma)] \geq \frac{1}{2}\Big[ & q \tr |\mathbb{E}_{\rho \leftarrow \eta_H}\rho\otimes \rho - \mathbb{E}_{\rho \leftarrow \eta_H}\Gamma(\rho) \otimes \Gamma(\rho)| \\
            & - (1-q) \tr |\mathbb{E}_{\rho \leftarrow \eta_H}\mathcal{E}(\rho \otimes \rho) - \mathbb{E}_{\rho \leftarrow \eta_H}\mathcal{E}(\Gamma(\rho) \otimes \Gamma(\rho))| \Big].
        \end{split}
    \end{equation}
    Now, we will use the contractivity of trace distance under completely positive trace-preserving maps: $\text{TD}(\mathcal{E}(\rho), \mathcal{E}(\sigma)) \leq \text{TD}(\rho, \sigma)$ to obtain:
    \begin{equation}\label{eq:tr_dist_noisy_noiseless_PRS_noisy_memory}
        \begin{split} 
            \text{TD}[D_qO_0(\sigma), D_qO_1(\sigma)] \geq \frac{1}{2}\Big[ &q \tr |\mathbb{E}_{\rho \leftarrow \eta_H}\rho\otimes \rho - \mathbb{E}_{\rho \leftarrow \eta_H}\Gamma(\rho) \otimes \Gamma(\rho)| \\
            & - (1-q) \tr |\mathbb{E}_{\rho \leftarrow \eta_H}\rho \otimes \rho - \mathbb{E}_{\rho \leftarrow \eta_H}\Gamma(\rho) \otimes \Gamma(\rho)| \Big].
        \end{split}
    \end{equation}
Now, we can just recombine the terms and use the fact that two copies of noisy PRS for probability of noise $p = \Omega(n^{-c})$, are distinguishable from two copies of noiseless PRS  using the SWAP test with high probability. From efficient distinguishability and the Holevo-Helstrom bound, we get, $\tr|\mathbb{E}_{\rho \leftarrow \eta_H}\rho \otimes \rho - \mathbb{E}_{\rho \leftarrow \eta_H} \Gamma(\rho) \otimes \Gamma(\rho)| = \Omega(n^{-c})$. Putting this back in \eqref{eq:tr_dist_noisy_noiseless_PRS_noisy_memory}
\begin{equation}
    \begin{split}
        \text{TD}[D_qO_0(\sigma), D_qO_1(\sigma)] &\geq \frac{1}{2}\Big[ (2q - 1) \tr |\mathbb{E}_{\rho \leftarrow \eta_H}\rho \otimes \rho- \mathbb{E}_{\rho \leftarrow \eta_H}\Gamma(\rho) \otimes \Gamma(\rho) | \Big] \\
        & = \frac{1}{2} \Big[ (2q - 1) \Omega(n^{-c})\Big].
    \end{split}
\end{equation}
Now, for $q = 1 - \Theta\Big(\frac{1}{\text{poly}(n)}\Big)$, we get that $\text{TD}[D_qO_0(\sigma), D_qO_1(\sigma)] = \Omega(n^{-c})$ for some $c > 0$. Thus, using Holevo-Helstrom bound~\cite{HOLEVO1973337, Helstrom1969-ny,bae2015quantum},
\begin{equation}
    \Pr = \frac{1}{2} + \Omega(n^{-c}),
\end{equation}
which suggests efficient distinguishability.
\end{proof}
We further give an example to show the distinguishability for the case when the noise in quantum memory acts on one copy only. For simplicity, we assume $2n$ qubit memory:
\begin{lemma}
   An adversary with noisy quantum memory can distinguish noisy PRS from noiseless PRS if the noise in PRS and memory is modeled as local depolarisation noise with noise probability $r$, 
   \begin{equation}
       \Lambda_r(\rho) = r\rho + \frac{1-r}{2}(I_1 \otimes \tr_1 \rho)
   \end{equation}
   where $r = \{p, q\}$ for noise in PRS and noise in memory respectively.
\end{lemma}
\begin{proof}
    We prove this using the SWAP test~\cite{barenco1997stabilization, garcia2013swap, beckey2023multipartite}. In this model, we can realize the SWAP test by first querying the oracle, using the SWAP operation to swap the workspace register and memory register, then querying the oracle again, and finally performing the SWAP test. Now, we assume a single copy noise in the quantum memory here, a local depolarisation channel acts on the copy in the quantum memory. Hence, we will have a state like $\rho_i \otimes \Lambda_q (\rho_i)$ and in the SWAP test, we will obtain $\tr (\rho_i \Lambda_q (\rho_i))$. In other words,
    \begin{equation}
        \sigma \xrightarrow{O_i} \rho_i \otimes \rho_i \xrightarrow{\Lambda_q} \rho_i \otimes \Lambda_q(\rho_i) \xrightarrow{\text{SWAP test}} \tr (\rho_i\Lambda_q(\rho_i)).
    \end{equation}
     It is easy to check for the oracles $O_0(\sigma) = \rho \otimes \tr \sigma$, preparing noisless PRS and $O_1(\sigma) = \Lambda_p(\rho) \otimes \tr \sigma$, preparing noisy PRS, we will have $\tr (\rho \Lambda_q(\rho))$ and $\tr (\Lambda_p(\rho) \Lambda_{pq}(\rho))$ respectively and the absolute value of average difference between these for $\Lambda$ being local depolarisation channel is $|p^2q - q|$ up to a constant multiplicative factor. This term is negligible if and only if either $q = \operatorname{negl}(n)$ or $p = 1 - \operatorname{negl}(n)$ where the former case would mean very high noise in memory and the latter case mean very low noise in PRS.
\end{proof}

Thus, even for simple unital noise models such as local depolarizing noise, we find that PRS where the observer has noisy quantum memory are not noise robust.

\section{Noise robustness of private key quantum money based on PRS} \label{sec:quantummoney}

\revA{Here, we describe our noise-robust private key quantum money scheme based on PRS and our modified completeness amplification scheme.

First, let us define private key quantum money and its security:
\begin{definition}[Private key quantum money~\cite{ji2018pseudorandom}]
   Let $\lambda$ be the security parameter. A private key quantum money scheme consists of following algorithms:
    \begin{enumerate}
        \item $\mathsf{KeyGen}$: Takes in the unary $1^{\lambda}$ and outputs a key uniformly randomly.
        \item $\mathsf{Bank}$: Takes a key as input and generates a quantum state, called banknote in this context.
        \item $\mathsf{Ver}$: Takes a key and an alleged bank note and either it accepts or rejects.
    \end{enumerate}
\end{definition}

Let $\mathsf{Count}$ be the money counter algorithm which takes in private key $k$ and $t^{\prime}$ copies of the alleged banknotes, calls the algorithm $\mathsf{Ver}$ on each copy and returns the number of times that it accepts.

\begin{definition}[Security of private key quantum money scheme~\cite{ji2018pseudorandom}] Let $\lambda$ be the security parameter. Given a private key money scheme $\mathcal{A}$ as defined above,
\begin{enumerate}
    \item Completeness error: The private key quantum money scheme $\mathcal{A}$ has completeness error $\epsilon$ if $ \mathsf{Ver}(k, \rho_k)$ accepts with
probability at least $1 - \epsilon$ for all valid banknotes $\rho_k$.
    \item Soundness error:  The private key quantum money scheme $\mathcal{A}$ is defined to have soundness error $\delta$ if for any polynomial time counterfeiter algorithm $C$, which maps $t$ banknotes to $t^{\prime} \geq t$ banknotes, 
    \begin{equation}
        \Pr[\mathsf{Count}(k, C(\rho_1, \rho_2, \ldots, \rho_t)) > t] \leq \delta .
    \end{equation}
\end{enumerate}
\end{definition}
A private key quantum money scheme is said to be secure if it has soundness error $\delta = \operatorname{negl}(\lambda)$. Further, we demand to have completeness error $\epsilon =\text{negl}(n)$ such that valid money is nearly always accepted. Note that in the original definition, a completeness error of up to $\epsilon<1/2$ is allowed, however this implies that valid quantum money is often rejected by the bank, and money effectively looses value.

\subsection{Noise-robust quantum money}

Now, we illustrate our scheme that has negligible completeness error $\epsilon=\text{negl}(n)$ even under noise.
The bank issues banknote 
\begin{equation}
    \$=(s,\ket{\psi_{q_1}},\ket{\psi_{q_2}},\dots,\ket{\psi_{q_L}})\,,
\end{equation}
which is composed of serial number $s$ and $L=\text{poly}(n)$ PRS $\ket{\psi_{q_i}}$ with bitstring $q_i$. The $q_i$ are generated using a secret master key $k$ (of size $\text{poly}(n)$): We generate  $q_1,q_2,\dots,q_L=f_k(s)$ from a (quantum-secure) pseudorandom function $f_k: \{0,1\}^{\text{poly}(n)}\rightarrow \{0,1\}^{\text{poly}(n)}$. 
Crucially, we choose $s$ large enough such that the probability that two banknotes $\$_1$, $\$_2$ have the same $s_1$, $s_2$ is negligible. Note that one can generate $s$ via a pseudorandom number generator.
For verification of $\$$, the bank projects each of the $L$ PRS with the projector $\ket{\psi_{q_i}}\bra{\psi_{q_i}}$ which can be efficiently done when knowing key $k$.

When
$T\geq (F_\text{min}+\eta)L$ of the projections succeed (with $\eta=1/\text{poly}(n)$, then the bank accepts the banknote, else rejects it. Here, we assume each PRS of the composite banknote is subject to noise channel $\Gamma(.)$ and
\begin{equation}
F_\text{min}=\min_{i}F(\Gamma(\ket{\psi_{q_i}}),\ket{\psi_{q_i}})
\end{equation}
is the minimal fidelity of all the noisy PRS with the noise-free ones. Here, we demand that $F_\text{min}>1/2+1/\text{poly}(n)$.

First, we consider the completeness error:
When $L=1$, the bank would only accept a valid noisy banknote with probability $F_\text{min}$, i.e. the money loses value with increasing noise with completeness error $\epsilon=1-F_\text{min}$. How to choose $L$ and $T$ such that the completeness error is $\epsilon=\text{negl}(n)$?
The trick is that by compositing the banknote into $L$ different PRS, we perform $L$ projections, where at least $LF_\text{min}$ succeed on average, where the law of large numbers guarantees that the composite banknote is accepted with high probability when $T\geq (F_\text{min}+\eta)L$ projections suceed.

Formally, we perform in total $L$ projections on the composite banknote 
for verification, with outcome $X_i=1$ when the projection succeeds, and $X_i=0$ otherwise.
The Chernoff bound for the average over $L$ Bernoulli trials $\hat{X}=\frac{1}{L}\sum_{i=1}^L X_i$  is given by
\begin{equation}
    \epsilon\equiv P(\hat{X}\leq \mu- \eta)\le P(\vert\hat{X} - \mu\vert \ge \eta) \le \exp(-\frac{L\eta^2}{3\mu})\,,
\end{equation}
where $\mu=F_\text{min}$ the average fraction of the projections succeeding and $\eta$ is the deviation from the average. 
The number of projections succeeding is given by $T\ge (F_\text{min}+\eta)L$.
To have negligible completeness error, we choose
\begin{equation}
    \epsilon< \exp(-\frac{L\eta^2}{3F_\text{min}})=\text{negl}(n)
\end{equation}
which we achieve by 
\begin{equation}
    L=\frac{3F_\text{min}}{\eta^2}\text{polylog}(n)\,.
\end{equation}
In particular, by choosing $\eta=1/\text{poly}(n)$ and $L=\text{poly}(n)$ we get $\epsilon=\text{negl}(n)$.

Next, we regard the soundness error $\delta$, i.e. whether there exist an efficient counterfeit algorithm $C$ that from $t$ banknotes can create additional counterfeit banknotes which are accepted by the bank. 
As the banknotes are PRS, they cannot be cloned~\cite{ji2018pseudorandom}, which holds true even under noise. Thus, there is no efficient algorithm that can create counterfeit smaller banknotes, as long as $F_\text{min}>1/2$. Further, from a single composite banknote no counterfeit money can be generated as shown in Ref.~\cite{aaronson_quantum_money_2012_arxiv}. We also find security against an embezzling attack that breaks the original completeness amplification scheme by Ref.~\cite{aaronson_quantum_money_2012_arxiv}. While we did not find a general proof of security against all possible attacks involving many banknotes, we believe generic security should hold.

\subsection{Embezzling attack}
We now describe an embezzling strategy that aims to forge fake banknotes. As we will see, this strategy does not work for our scheme, but works for the original completeness amplification scheme by Ref.~\cite{aaronson_quantum_money_2012_arxiv}.
In this attack, we are are given a valid (noise-free) composite banknote $\$$, and switch one of its smaller banknotes $\ket{\psi_{q_i}}$ with some random state $\rho$.  As the banknote consists of PRS, $\rho$ will have negligible fidelity with the PRS. 
However, the composite banknote remains valid (for $T\leq L-1$), but one has embezzled a smaller banknote with PRS $\ket{\psi_{q_i}}$. %
We now assume that we have in total $L$ different composite banknotes available, and apply our attack on all $L$. We gain $L$ embezzled smaller banknotes, which we composite into a counterfeit composite banknote $\$_\text{cf}=(s',\ket{\psi_{q_1}},\ket{\psi_{q_2}},\dots,\ket{\psi_{q_L}})$, where each $\ket{\psi_{q_i}}$ comes from a different composite banknote with different serial number. 
However, as the counterfeiter cannot efficiently compute the correct serial number $s'$ that matches the embezzled PRS (this would require knowledge of secret key $k$), $\$_\text{cf}$ has an negligible acceptance probability by the bank.

We highlight that the original amplification protocol by Ref.~\cite{aaronson_quantum_money_2012_arxiv} has a major vulnerability: In particular, it breaks under our previously described embezzling strategy:
In Ref.~\cite{aaronson_quantum_money_2012_arxiv}, the composite banknote $\$_\text{a}=((s_1,\ket{\psi_{s_1}}),(s_2,\ket{\psi_{s_2}}),\dots,(s_L,\ket{\psi_{s_L}}))$ consists of $L$ smaller banknotes, but where each smaller banknote has their individual serial number $s_i$ and state $\ket{\psi_{s_i}}$. Now, we take one of the smaller banknotes $(s_i,\ket{\psi_{s_i}})$, and replace it with some $(s_i,\rho)$. The resulting composite banknote $\$_\text{a}'$ is still accepted by the bank (for $T\leq L-1$), but now the counterfeiter has  gained a valid smaller banknote $(s_i,\ket{\psi_{s_i}})$. 
The counterfeiter now performs this strategy on $L$ valid composite banknotes, gaining $L$ embezzled smaller banknotes. These $L$ smaller banknotes into a counterfeit composite banknote $\$_\text{a,cf}$ which is accepted by the bank with probability $1$, effectively turning $L$ valid banknotes into $L+1$ banknotes accepted by the bank. 

Note that this counterfeit strategy does not work when replacing the individual $s_i$ with a single serial number $s$ that is valid only for the full composite banknote, as we cannot determine the correct serial number of the embezzled composite banknote.
}

\revA{
\section{Verifiable PRDM}\label{sec:VPRDM}
In this section, we give a short review on the verifiable PRDM (VPRDM), which were introduced after this work. 
VPRDMs are PRDMs, with the additional property that their correct preparation can be efficiently verified using key $k$~\cite{haug2024pseudorandom}. This property is for example important for applying EFI pairs for bit commitment.

We reproduce the formal definition of VPRDM in the followng:
\begin{definition}[VPRDM~\cite{haug2024pseudorandom}]\label{def:VPRDM_sup}
  Let $\lambda=\operatorname{poly}(n)$ be the security parameter with keyspace $\mathcal{K}=\{0,1\}^{\lambda}$. A keyed family of $n$-qubit density matrices $\{\rho_{k,m}\}_{k \in \mathcal{K}}$ is a VPRDM when:
    \begin{enumerate}
        \item {Efficiently preparable}: There exists an efficient quantum algorithm $\mathcal{G}$ such that $\mathcal{G}(1^{\lambda}, k,m) = \rho_{k,m}$.
        \item {Computational Indistinguishability}: $t=\mathrm{poly}(n)$ copies of $\rho_{k,m}$ are computationally indistinguishable from the GHSE $\eta_{n,m}$. In particular, for any efficient quantum algorithm $\mathcal{A}$ we have
        \begin{equation}
            \Big{|}\Pr_{k \leftarrow \mathcal{K}}[\mathcal{A}(\rho_{k,m}^{\otimes t}) = 1] - \Pr_{\rho \leftarrow \eta_{n,m}}[\mathcal{A}(\rho^{\otimes t}) = 1]\Big{|} = \operatorname{negl}(\lambda).
        \end{equation}
    \item {Efficient verification}: There is an efficient quantum algorithm $\mathcal{V}(\rho, k,m)$ to verify that $\rho_{k,m}$ is indeed the correct VPRDM generated by key $k$. In particular, it must fulfill the completeness condition (i.e. correct states are accepted)
    \begin{equation}
    \mathcal{V}(\rho_{k,m}, k,m)=1
     \end{equation}
     and soundness condition (i.e. wrong states are rejected with high probability)
     \begin{equation}
            \Pr_{k' \leftarrow \mathcal{K}/\{k\}}[\mathcal{V}(\rho_{k,m}, k',m) = 1] = \operatorname{negl}(\lambda)\,.
    \end{equation}
    \end{enumerate}
\end{definition}
Note that VPRDMs differ from PRDMs only by the additional third condition. 

An efficient construction of VPRDMs is known via~\cite{haug2024pseudorandom}
\begin{equation}\label{eq:VPRDM_construction}
\{\rho_{k,m}\}=\{U_k (\ket{0}\bra{0})^{\otimes n-m}\otimes \sigma_m U_k^\dagger\}_{k\in\mathcal{K}}\,,
\end{equation}
where $\sigma_m=I_m/2^m$ is the maximally mixed state, $\{U_k\}_{k\in\mathcal{K}}$ is PRU with keyspace $\mathcal{K}=\{0,1\}^{\mathrm{poly}(n)}$, and $m<n-\mathrm{polylog}(n)$.

For this VPRDM construction, the efficient verification is performed by applying the inverse $U_k^\dagger$ and projecting onto $(\ket{0}\bra{0})^{\otimes n-m}$, i.e. 
\begin{equation}
\mathcal{V}(\rho, k,m)=\text{tr}((\ket{0}\bra{0})^{\otimes n-m}\operatorname{tr}_{m}(U_k^\dagger \rho U_k))\,,
\end{equation}
where $\operatorname{tr}_{m}(\cdot)$ traces out the last $m$ qubits.

}
\end{document}